\documentclass[twoside]{article}

\usepackage[preprint]{aistats2027}

\usepackage[square,numbers]{natbib}

\usepackage{amsthm}
\usepackage{amsmath}
\usepackage{amssymb}
\usepackage{algorithm}
\usepackage{algpseudocode}
\usepackage{bm}
\usepackage{graphicx} 
\usepackage{mathtools}
\usepackage{mathrsfs}
\usepackage{multirow}
\usepackage{caption}
\usepackage{color}
\usepackage{booktabs}
\usepackage{hyperref}       
\usepackage{booktabs} 
\usepackage{array}
\usepackage{tabularx}

\begin{document}

%
\runningtitle{Bayesian Nonparametric Clustering of Marked Point Patterns}

%
\runningauthor{Minsung Choi, Seonghyun Jeong}

\twocolumn[

\aistatstitle{
Asymptotically Valid Laplace Variational Inference for Bayesian Nonparametric Clustering of Marked Point Patterns
}

\aistatsauthor{
  Minsung Choi
  \And
  Seonghyun Jeong\textsuperscript{*}
}

\aistatsaddress{
  Yonsei University
}
] 

\begingroup
\renewcommand{\thefootnote}{\fnsymbol{footnote}}
\footnotetext[1]{Corresponding author: \texttt{sjeong@yonsei.ac.kr}}
\endgroup

\begin{abstract}
We introduce Dirichlet process mixtures of marked Poisson point processes, a Bayesian nonparametric model for clustering replicated marked point patterns. The model jointly infers latent cluster structure, the number of clusters, and continuous mark-specific intensity surfaces with uncertainty quantification, a combination existing methods provide only in part. 
Using a squared-link intensity whose likelihood is tractable without gridding or thinning, we develop, to our knowledge, the first variational Bayes method for clustering marked Poisson point processes. 
Its constrained Laplace approximation eliminates sign ambiguity and nodal lines by turning each coefficient update into a well-posed convex problem, attaining computational efficiency and inferential fidelity with theoretical guarantees of mode consistency, exponential-weight concentration, and uniform total-variation accuracy. 
Our method outperforms competitors in accuracy and speed on synthetic data, and an NBA shot-chart analysis uncovers offensive archetypes.
\end{abstract}


\section{INTRODUCTION}
\label{sec:intro}

Marked point-process data arise whenever each observational unit generates a finite random collection of events on a continuous domain with event-level attributes. A Marked Poisson point process (MPPP) provides a natural probabilistic formalism for such data \citep{kingman1993poisson, daley2003introduction}.
Examples arise in biology \citep{Li2019TumorPathology}, cosmology \citep{Martinez2010MarkConnection}, social sciences \citep{Mohler2014MarkedCrime}, and sports analytics \citep{Jiao2021Bayesian}. 
Often, one observes such a process for each subject, and the central inferential goal is to identify latent groups sharing similar event intensity and mark distribution.

Despite the prevalence of such data, clustering replicated MPPPs remains underdeveloped. In this paper, we propose Dirichlet process mixtures of marked Poisson point processes 
model for clustering replicated marked point data. Our framework clusters process realizations without discretizing the domain, while  learning the number of clusters, modeling mark-specific intensity structure, and quantifying uncertainty in both cluster assignments and intensity functions, a combination existing approaches provide only in part. Moreover, a subject-specific offset separates overall volume from 
spatial and mark structure for clustering. Mark-specific intensities induce spatially varying mark probabilities. Joint estimation within each cluster borrows statistical strength across subjects, stabilizing inference in sparse regimes.

We develop an efficient variational inference (VI) algorithm, to the best of our knowledge the first variational Bayes method for clustering MPPPs. Posterior computation poses a further challenge, as existing intensity inference methods sacrifice either computational efficiency or inferential fidelity. Our algorithm attains both by combining VI with a constrained Laplace approximation (LA) for the nonconjugate basis coefficients, while exploiting the computational advantages of the squared link. We establish consistency of the constrained mode estimator
, prove that the exponential weight asymptotically concentrates on the feasible region, and prove uniform total-variation convergence of the exact positive-chamber update to its LA. The positivity restriction resolves the global sign ambiguity and the nodal-line pathology inherent to squared-link models. The algorithm is grid-free, stable, efficient yet accurate, and suited to high dimensions.

\subsection{Related Works}

\paragraph{Inference for Poisson and Cox Processes.}
Under the canonical log link, the log-Gaussian Cox process (LGCP) is a standard and flexible model for point-process intensity, but exact inference is hindered by the intractable likelihood integral, motivating discretization, finite-element representations, and integrated nested Laplace approximation (INLA)-type methods \citep{moller1998log,moller2004statistical,Rue2009INLA,Simpson2016GoingOffGrid,Teng2017Bayesian}. Alternative links and variational approximations mitigate these costs. \citet{Adams2009Tractable} introduced the sigmoidal Gaussian Cox process with thinning-based augmentation, while \citet{lloyd2015variational} proposed the first fully variational approach for continuous Gaussian-process-modulated Poisson processes under a squared link. \citet{John2018LargeScaleCP} extended this line to large-scale spatiotemporal settings using variational Fourier features, emphasizing scalability and the nodal-line issues. \citet{Aglietti2019Structured} further developed structured VI for continuous sigmoidal Cox processes. Multi-process models such as \citet{Lloyd2016LPPA} and \citet{Ding2018BaNPPA} also use VI, but  decompose each process into shared latent functions rather than clustering replicated point processes.

\paragraph{Clustering of Point Processes.}
Existing point-process clustering methods also leave important gaps. \citet{Barrack2017NHPP} considered finite mixtures of non-homogeneous Poisson processes for event-time clustering, but it is restricted to one-dimensional unmarked data, requires the number of components and spline complexity to be specified a priori, and cannot quantify uncertainty of intensity surfaces. \citet{Phung2017RFS} proposed a Dirichlet process mixture of Poisson random finite sets, but the component model is tied to conjugate families and includes an explicit cardinality, so clustering is driven jointly by set size and simple feature distributions rather than by continuous intensity structure. \citet{Yin2023MatrixClustering} cluster subjects via a mixture of finite mixtures of matrix normal distributions, but their procedure is two-stage on discretized intensity matrices. 
The closest analogue is \citet{Yin2021RowClustering}, who clustered repeated marked event sequences via a mixture of multi-level marked LGCPs, but it is tailored to matrix-valued repeated observations, point-estimate centered, and requires selecting the number of clusters and a smoothing bandwidth. Hawkes- and neural-temporal-point-process clustering targets excitation or neural sequence dynamics instead \citep{Xu2017DMHP,Zhang2022NTPPMix,Ding2023CNTPP}.

Section~\ref{sec:prelim} reviews preliminaries, Section~\ref{sec:model} defines the proposed model, Section~\ref{sec:inference} presents the VI algorithm and theory, Section~\ref{sec:experiments} reports experiments and data analysis, and Section~\ref{sec:conclusion} concludes.

\section{PRELIMINARIES}
\label{sec:prelim}

\subsection{Marked Poisson Point Process}


Let marked point data be denoted by $(\bm{Y},\bm{m}) = \{(\bm{y}_1,m_1), \ldots,(\bm{y}_N,m_N)\}$, where the set of points $\bm{Y} = \{\bm{y}_1, \ldots, \bm{y}_N\}$ represents a realization of a point process within a domain $\mathcal{B} \subseteq \mathbb{R}^H, H\geq 1$, and the marks $\bm{m} = \{m_1, \ldots ,m_N\}$ taking values in some space $\mathcal{M}$ denote characteristics associated with each point. 
Thus, $(\bm{y}_i, m_i)$ is a random point in the product space $\mathcal{B} \times \mathcal{M}$. 
We model the marked point data $(\bm{Y}, \bm{m})$ as an MPPP, where the points $\bm{Y}$ are modeled as a realization of a Poisson point process (PPP) \citep{prekopa1958secondary, kingman1993poisson, daley2003introduction}.
PPP is defined via a non-atomic finite measure on the Borel $\sigma$-algebra of $\mathcal{B}$, $\Lambda(\cdot) : \mathscr{B}(\mathcal{B}) \rightarrow \mathbb{R}^+$, called the mean measure. 
For any Borel set $A \subseteq \mathcal{B}$, the number of points in $A$, $N(A) \in \mathbb{N}$, is a Poisson random variable with rate $\Lambda(A)$. If $\Lambda$ is absolutely continuous with respect to the Lebesgue measure, then $\Lambda(A)=\int_A \lambda(\bm y)\,d\bm{y}$, where $\lambda=d\Lambda/d\bm{y}:\mathcal{B}\to \mathbb{R}^+$ is the Radon--Nikodym derivative, called the intensity function.
For any disjoint Borel sets $A_1, A_2 \subseteq \mathcal{B}$, $N(A_1)$ and $N(A_2)$ are independent.
We write
$\bm Y \sim \mathcal{PP}(\Lambda)$ to denote \(\bm Y\) is a realization from a PPP with mean measure $\Lambda$.

Assuming the marks $\bm{m}$ are conditionally independent given $\bm{Y}$, the joint likelihood of the MPPP is
\begin{align*}
\bm{L}(\bm \Theta|\bm{Y}, \bm{m}) = \exp(-\Lambda(\mathcal{B})) \prod_{i=1}^N \lambda(\bm{y}_i) \prod_{i=1}^N \mathrm p(m_i \mid \bm{y}_i),
\end{align*}
where the first two terms are the PPP likelihood of $\bm{Y}$, $\mathrm p(m_i \!\mid\! \bm{y}_i)$ is the conditional density of the mark given point with respect to a dominating measure on $\mathcal{M}$, and $\bm \Theta \!=\! (\lambda, \mathrm p)$. 
By the marking theorem \citep{kingman1993poisson}, this factorization yields an MPPP. In Section~\ref{sec:model}, mark probabilities are induced by the relative magnitude of mark-specific intensities.

\subsection{Dirichlet Process and Dirichlet Process Mixture}

A Dirichlet process (DP) \citep{Ferguson1973} is a distribution over a random probability measure. A random probability measure $G$ on a measurable space $(\mathcal{S}, \mathcal{F})$ 
follows a DP if, for any finite measurable partition $\{A_1, \ldots, A_k\}$ of $\mathcal{S}$, 
$(G(A_1), \ldots, G(A_k)) \!\sim\! \mathrm{Dir}\bigl(\alpha G_0(A_1), \ldots, \alpha G_0(A_k)\bigr)$.
We denote $G \!\sim\! \mathcal{DP}(\alpha,\! G_0)$, where $G_0$ is the base measure satisfying $\mathbb{E}[G(A)] \!=\! G_0(A)$ for all measurable $A$, and $\alpha\! >\! 0$ is the concentration parameter governing the degree of shrinkage of $G$ toward $G_0$. 

The Dirichlet process mixture (DPM) \citep{Antoniak1974} is an infinite mixture model that learns the number of components from data, providing a principled nonparametric Bayesian framework for density estimation, latent model, and clustering \citep{Ferguson1973, Neal2000}. 
The DPM is defined hierarchically as 
$G\! \sim\! \mathcal{DP}(\alpha, G_0), \;
\Theta_i \!\mid\! G \!\sim\! G, \;
y_i \!\mid\! \Theta_i \!\sim\! F(\Theta_i)$,
for $i \!=\! 1, \ldots, n$, where $F(\Theta_i)$ denotes the distribution of $y_i$ given $\Theta_i$.

A constructive definition of $G$ via the stick-breaking representation \citep{Sethuraman1994} formulates the discrete random measure $G$ as an infinite weighted sum of point masses:
\begin{equation}\label{eq:dp-stick-breaking}
\begin{aligned}
G = \sum_{k=1}^{\infty}\pi_k\,\delta_{\Theta_k}, \quad&
\pi_k = \phi_k\prod_{j=1}^{k-1}(1-\phi_j),\\[-4pt]
\Theta_k \stackrel{\text{iid}}{\sim} G_0,\quad &
\phi_k \stackrel{\text{iid}}{\sim} \mathrm{Beta}(1,\alpha),
\end{aligned}
\end{equation}
where $\delta_{\Theta_k}$ is the Dirac measure at $\Theta_k$, and the weights $\{\pi_k\}$ satisfy $\sum_{k=1}^{\infty} \pi_k = 1$ almost surely. 
The DPM then has an equivalent form \citep{Ishwaran2001}, with latent membership $z_i \in \mathbb{N}$ of the $i$-th observation: for $i \!=\! 1,\ldots,n$, 
$z_i\!\mid\!\{\pi_k\}_{k\ge1}
\!\stackrel{\text{iid}}{\sim}\!\sum_{k=1}^{\infty}\pi_k\,\delta_{\Theta_k}, \;
y_i\!\mid\!z_i,\!\{\Theta_k\}_{k\ge1}
\!\stackrel{\text{ind}}{\sim}\!F(\cdot \mid \Theta_{z_i}).
$

\subsection{Variational Inference}

VI recasts the estimation of an intractable posterior as an optimization problem \citep{Jordan1999, Blei2017}: it approximates the exact posterior $p(\Theta \mid y)$ within a family of distributions $\mathcal{Q}$ by finding the optimal distribution $q^*(\Theta) \in \mathcal{Q}$ that minimizes the Kullback-Leibler (KL) divergence to the true posterior:
$q^*(\Theta) = \arg\min_{q \in \mathcal{Q}} \text{KL}(q(\Theta) \parallel p(\Theta \mid y))$.
As computing the KL divergence inherently requires intractable integration, one equivalently maximizes the evidence lower bound (ELBO): 
$\text{ELBO}(q) = \mathbb{E}_{q(\Theta)}[\log p(y, \Theta)] - \mathbb{E}_{q(\Theta)}[\log q(\Theta)]$.
%
The mean-field approximation restricts $\mathcal{Q}$ to mutually independent distributions, $q(\Theta) = \prod_{j=1}^J q_j(\Theta_j)$, under which the optimal variational distribution for each component $\Theta_j$ takes the explicit form 
$q_j^*(\Theta_j) \propto \exp(\mathbb{E}_{-j}[\log p(y, \Theta)])$ \citep{Wainwright2008, Blei2017},
where $\mathbb{E}_{-j}[\cdot] \equiv \mathbb{E}_{q_{-j}(\Theta_{-j})}[\cdot]$ denotes the expectation with respect to $q_{-j}(\Theta_{-j}) = \prod_{\ell \neq j} q_\ell(\Theta_\ell)$.

\section{MODEL}
\label{sec:model}

In this section, we formulate the 
Dirichlet process mixtures of marked Poisson point processes (DPM-MPPP) 
model, which provides a principled Bayesian approach for the joint inference of latent cluster assignments and MPPPs.
For each subject $i \in \{1, 2, \ldots, n\}$, we observe $N_i$ pairs of marked point data, denoted as $(\bm{Y}_i, \bm{m}_i) = \{(\bm y_{ij},m_{ij})\}_{j=1}^{N_i}$, and a scalar variable $T_i$. Here, $T_i$ serves as an offset in the Poisson intensity formulation, and thus represents the level of exposure--such as time, area, or size--for subject $i$.
The point data $\bm{Y}_i$, observed within the domain $\mathcal{B} \subseteq \mathbb{R}^H, H\geq 1$, is modeled as a realization of a PPP
, and the mark $\bm{m}_i$ is assumed to take values from a finite discrete set. In this work, we assume a binary mark $m_{ij}\in\{0,1\}$ for simplicity. $z_i \in \{1,2,3\ldots\}$ indicates the cluster assignment of subject $i$.
For notational convenience, we define the index set of events for subject $i$ associated with mark $m \in \{0,1\}$ as $\mathcal{I}_{im} = \{j \in \{1, \ldots, N_i\} : m_{ij} = m\}$ and the corresponding subset of points as $\bm{Y}_{im} = \{\bm{y}_{ij} : j \in \mathcal{I}_{im}\}$. We denote the mean measure of our model as $\Lambda_{km}, \ k \in \mathbb N, \ m\in\{0,1\}$, where $\Lambda_{km}(A) = \int_A \lambda_{km}(\bm y) d\bm{y}$ and $\lambda_{km}(\bm y)$ is the intensity function.
Crucially, we assume that subjects assigned to the same cluster $k$ share an identical baseline mean measure for each mark $m$, $\Lambda_{k m}$. 
By jointly estimating each cluster’s baseline intensity from aggregated data of its affiliated subjects, the model mitigates data sparsity and borrows statistical strength across subjects. 
The model is formulated as follows. 
Conditional on the cluster assignment $z_i$, we assume that $\bm{Y}_{i0}$ and $\bm{Y}_{i1}$ are modeled as realizations from independent PPPs with mean measures $T_i\Lambda_{z_i0}$ and $T_i\Lambda_{z_i1}$, respectively:
\begin{align*}
\bm{Y}_{im} \mid z_i, \Lambda_{z_im} &\sim \mathcal{PP}(T_i\Lambda_{z_im}),
\end{align*}
Scaling the intensity by offset $T_i$ disentangles it into the scale and unit intensities, allowing the DPM to cluster subjects based purely on their spatial tendencies, marginalizing the effect of varying total volume.
Assuming an independent Poisson process for each mark is likelihood-equivalent to a hierarchical formulation where the locations follow a Poisson process, and the marks follow a spatially varying Bernoulli process:
\begin{align*}
\bm{Y}_i \mid z_i, \Lambda_{z_i} &\sim \mathcal{PP}(T_i\Lambda_{z_i}),\\[-1pt]
m_{ij} \mid \bm{y}_{ij}, z_i &\sim \text{Bernoulli}(p_{z_i}(\bm y_{ij})),
\end{align*}
where $\Lambda_{z_i}(\cdot)\! =\! \Lambda_{z_i0}(\cdot) \!+\! \Lambda_{z_i1}(\cdot)$, and the mark
probability surface is $p_{z_i}\!(\bm y_{ij}) \!=\! \frac{\lambda_{z_i,1}(\bm{y}_{ij})}{\lambda_{z_i,0}(\bm{y}_{ij}) + \lambda_{z_i,1}(\bm{y}_{ij})}$. 
This follows from the superposition theorem and the marking theorem \citep{kingman1993poisson}. 
This equivalence reveals an advantage of our approach: the model naturally induces spatially varying mark probabilities, with event--mark dependence governed by the relative magnitudes of the intensities.

The baseline intensity is modeled as
$\lambda_{km}(\bm y) \!=\! h\left(W_{km}(\bm y)\right)$
, $W_{km}(\bm y) \!=\! \bm{B}(\bm y)^\top\bm{\theta}_{km}$,
where the squared link $h(x) = x^2$ maps the latent process \(W_{k m}\) to a nonnegative intensity, $\bm{B}(\bm y) \in \mathbb R^d$ is a vector of basis functions evaluated at $\bm y$, and $\bm{\theta}_{km} \in \mathbb R^d$ is the coefficient vector. 
The basis $\bm{B}(\cdot)$ must be nonnegative with no common zero, to guarantee a nonempty feasible set defined by linear constraints. 
The squared link and the related theories are discussed in Sections~\ref{subsec:squared-link} and~\ref{subsec:theory}.

We assign a Normal prior on $\bm{\theta}_{km}$ and an Inverse Gamma prior on the
variance parameters $\tau_{km}^2$
\begin{align*}
\bm{\theta}_{km} \mid \tau^2_{km} \sim
  \mathrm{Normal}\bigl(\bm{0},\, \tau^2_{km}\bm\Omega^{-1}\bigr),
\;
\tau^2_{km} \sim \mathrm{IG}(a_0, b_0),
\end{align*}
where $\bm\Omega$ is a positive definite matrix, so $\bm{\theta}_{km}$ has
mean $\bm 0$ and precision matrix $\tau_{km}^{-2}\bm\Omega$, consequently
inducing a Gaussian process prior on
$W_{km}(\cdot) = \bm{B}(\cdot)^\top\bm{\theta}_{km}$
\citep{rasmussen2006gaussian, bishop2006pattern}.
This formulation is highly flexible, as $\bm\Omega$ can be tailored to the basis $\bm{B}(\cdot)$; e.g., for a B-spline basis expansion \citep{deboor2001practical}, $\bm \Omega$ can be built from a p-spline penalty matrix $\bm{\mathrm{P}}$ \citep{lang2004bayesian}. Since $\bm{\mathrm{P}}$ is
singular, we take $\bm\Omega \!=\! \bm{\mathrm{P}} \!+\! c_{\bm\Omega}\bm{I}_d$,
with a small constant $c_{\bm\Omega} \!>\! 0$, making 
$\bm\Omega$ positive definite and the prior proper, while
leaving the roughness penalty of $\bm{\mathrm{P}}$ essentially unchanged.

Following the stick-breaking representation of the DP in \eqref{eq:dp-stick-breaking}, let
$\{\pi_k\}_{k \geq 1}$ be the random mixture weights. 
We draw cluster indicators 
as: $z_i \mid \{\pi_k\}_{k \geq 1} \overset{\mathrm{iid}}{\sim} \sum_{k=1}^{\infty} \pi_k \, \delta_k.$

\subsection{Resolving Fidelity--Efficiency Trade-Off}
\label{subsec:squared-link}

Poisson process intensity inference has faced a trade-off between inferential fidelity and computational efficiency. Existing methods either preserve accuracy at a substantial computational cost or gain efficiency through approximations that compromise fidelity.
Under the canonical log link \citep{moller1998log, moller2004statistical}, the likelihood integral is doubly intractable, so LGCP methods require additional approximation layers, including gridding, finite-element representations, and INLA-type approximations \citep{Taylor2013lgcp, Simpson2016GoingOffGrid,Rue2009INLA,Teng2017Bayesian}. 
These approaches suffer from inherent challenges regarding the choice of grid or mesh, becoming increasingly unmanageable as the geometry of the domain becomes irregular or the dimension of the domain increases \citep{Simpson2016GoingOffGrid,Teng2017Bayesian,John2018LargeScaleCP}.
Sigmoid-link methods \citep{Adams2009Tractable, Gunter2014Efficient, Donner2018Sigmoidal, Aglietti2019Structured} avoid direct evaluation of the integral through thinning-based augmentation, but require posterior inference over latent thinned points, whose number can grow substantially in large domains and critically determines the computational cost. The sigmoid link does not yield a log-concave likelihood, so the theory developed in Section~\ref{subsec:theory} is not available. 
A linear link makes the integral analytically trivial, but positivity, $\inf_{\bm y\in\mathcal B}\bm B(\bm y)^\top\!\bm\theta_{km}\!>\!0$, must then be enforced, which severely complicates variational inference. Positive variational families, however, either lack closed-form expectations (e.g., multivariate truncated Normals) or do not retain posterior dependence among basis coefficients (e.g., factorized Gammas). 

The squared link \citep{lloyd2015variational} removes the integration bottleneck, but leaves an optimization bottleneck in the nonconjugate basis-coefficient block. Its key computational advantage is that the integral and expected data term of the likelihood are analytic under the Gaussian variational family, giving a closed-form ELBO. 
Lacking closed-form coordinate-ascent updates for the basis coefficients, \citet{lloyd2015variational}  numerically optimize the full variational mean and covariance ($d+d(d+1)/2$ parameters per block), which is prohibitive and unstable in high dimensions. Moreover, the squared link introduces both a global sign ambiguity and the nodal-line problem, making optimization sensitive to initialization and susceptible to undesirable local modes \citep{John2018LargeScaleCP}.

Our strategy is to retain the tractable-integral advantage of the squared link while resolving the optimization bottleneck with a constrained LA. Restricting the mode search to the positive chamber selects one representative of the sign-symmetric pair, excludes nodal lines, and renders the objective concave, so each coefficient update reduces to a $d$-dimensional convex problem with a unique solution. 
Section~\ref{subsec:theory} proves that this restriction is justified: the constrained mode is consistent, the excluded region carries asymptotically negligible weight, and the positive-chamber Gaussian approximation uniformly converges in total variation. 
The trade-off is resolved: inference is tractable (analytic likelihood, ELBO), efficient (convex mode update, analytic covariance), and theoretically justified. 

\section{INFERENCE}
\label{sec:inference}

We propose an efficient VI algorithm 
to jointly infer the latent cluster assignments and the MPPP intensities on the continuous domain, 
incorporating an LA to enable highly efficient inference while maintaining the fidelity of the approximate posterior. 
We provide rigorous theoretical guarantees for the framework.

\subsection{Variational Inference}
\label{subsec:vi}

For tractable inference with the infinite-dimensional DP, we adopt the 
truncated stick-breaking representation \citep{blei2006variational} with truncation level $K$. 
Let 
$\bm\phi\!:=\!(\phi_1,\!\ldots,\!\phi_{K\!-\!1}), \ \bm z\!:=\!(z_1,\!\ldots,\!z_n), \ \bm\theta\!:=\!\{\bm\theta_{km}\!:\! k\!=\!1,\!\ldots\!,K,\ m\!\in\!\{0,\!1\}\}, \ \bm\tau^{2}\!:=\!\{\tau_{km}^2\!:\!k\!=\!1,\!\ldots,\!K,\ m\!\in\!\{0,\!1\}\}$.
We adopt the mean-field assumption: 
\begin{align*}
 &q(\bm{\phi}, \bm{z}, \bm{\theta}, \bm{\tau^2}) \\[-2pt]
 &=
 \prod_{k=1}^{K-1} q(\phi_k)
 \prod_{i=1}^n q(z_i)
 \prod_{k=1}^K q(\bm{\theta}_{k0})q(\bm{\theta}_{k1})
  q(\tau_{k0}^2)q(\tau_{k1}^2),
\end{align*}
with 
$q(z_i)\!=\!\mathrm{Discrete}(\nu_{i 1},\!\dots,\!\nu_{i K})$, $q(\bm{\theta}_{k m})\!=\!\mathrm{Normal}(\bm{\mu}_{\bm \theta_{k m}}, \bm{\Sigma}_{\bm \theta_{km}})$, $q(\tau_{km}^2)\!=\!\mathrm{IG}(\alpha_{q_{km}}, \beta_{q_{km}})$,  $q(\phi_k)\!=\!\mathrm{Beta}(\gamma_{k1},\gamma_{k2}), \ k\!=\!1,\!\ldots\!,K\!-\!1$.
The full ELBO and closed-form coordinate updates for the conjugate blocks are derived in Appendix~\ref{app:elbo} and Appendix~\ref{app:cavi}.

\subsection{Constrained Laplace Approximation}
\label{subsec:laplace}

In our VI framework, the basis coefficients $\bm{\theta}_{km}$ do not admit 
closed-form updates, while the other blocks do. Rather than numerically optimizing the full ELBO over the variational mean and covariance as in \citep{lloyd2015variational}, which is prohibitive in high dimensions, we propose an LA to the variational update of $\bm{\theta}_{km}$ \citep{wang2013variational}. 

Let $J(\bm{\theta}_{km})$ denote the variational objective, the kernel of the expected log-joint probability under all other variational factors, $\mathbb{E}_{q_{-\bm{\theta}_{km}}}[\log p(\bm{Y}, \bm{\phi}, \bm{z}, \bm{\theta}, \bm{\tau^2})]$:
\begin{align*}
J(\bm{\theta}_{km})\!=\! &-\bm{\theta}_{km}^\top \bm{A}_{km} \bm{\theta}_{km}\!+\!2\sum_{i=1}^n \nu_{ik}\!\!\sum_{j \in \mathcal{I}_{im}}\! \!\log\!\left|s_{\bm{\theta_{km}}}(\bm y_{ij}) \right| ,
\end{align*}
where $\bm{A}_{km} \!=\! (\sum_{i=1}^n \nu_{ik} T_i) \bm M \!+\! (\eta_{km}/2)\bm \Omega$, with $\bm M \!=\! \int_{\!\mathcal{B}} \!\bm B(\!\bm y\!)\! \bm B(\!\bm y\!)^{\!\top} \!d\bm{y}$, expected precision $\eta_{km} \!=\! \mathbb{E}_{q}[{\tau_{\!km}^{\!-\!2}}]$, and $s_{\bm{\theta}}(\bm y) \!:=\! \bm B(\bm y)^\top \!\bm\theta$. The derivation is in Appendix~\ref{app:theta-objective}.

An LA requires a regular mode of \(J(\bm{\theta}_{km})\) and its Hessian, 
but unconstrained mode finding over \(\mathbb R^d\) is
problematic. The squared link induces a global sign symmetry, as
\(\bm{\theta}_{km}\) and \(-\bm{\theta}_{km}\) give the same intensity, 
so any mode has a sign-symmetric counterpart. Moreover, the
zero-crossing hyperplanes
$s_{\bm{\theta_{km}}}(\bm y_{ij})\!=\!0$
partition the parameter space into multiple sign regions. Extra
local modes may arise in sign-changing regions, leading to nodal-line
solutions and sensitivity to initialization
\citep{John2018LargeScaleCP}. Hence the unrestricted objective is not 
suitable for a stable Laplace update.

To obtain an identifiable Laplace center, we restrict the
mode-finding problem to the strictly positive chamber. For fixed constants
\(\delta>0\) and \(R>0\), define
$D_{\delta} \!\coloneq\! \left\{\bm{\theta}\!\in\!\mathbb{R}^d\!:\!\|\bm{\theta}\|_2 \!\le\! R,\  \inf_{\bm y\in\mathcal B}s_{\bm{\theta}}(\bm y) \!\ge\! \delta \right\}$.
Clearly, $ J(\bm\theta)$ is concave on \(D_\delta\) whenever
\(\bm A_{km}\) is positive semidefinite, though it is not globally concave on $\mathbb R^d$. Hence positive-chamber optimization is a constrained concave maximization, which rules out spurious local modes
within the chamber. 
The constraint also keeps $s_{\bm{\theta}}(\bm y)$ uniformly away from zero, so the linear term never crosses zero on \(\mathcal B\), eliminating
nodal-line configurations. Moreover, the
positive-chamber restriction selects one representative from the
sign-symmetric pair \(\pm\bm\theta\). 
Section~\ref{subsec:theory} establishes its theoretical validity.

Once the positive mode $\bm \mu_{\bm{\theta}_{km}}^* \!=\! \arg\max_{\bm{\theta}_{km} \!\in\! D_{\delta}} J(\bm{\theta}_{km})$ is obtained, the LA is completed by computing the Hessian matrix $\bm H(\bm{\mu}_{\bm{\theta}_{km}}^*)$ evaluated at this mode:
\begin{align*}
\bm H(\bm{\mu}_{\bm{\theta}_{km}}^*) = -2\bm{A}_{km} - 2 \sum_{i=1}^n \nu_{ik} \sum_{j \in \mathcal{I}_{im}} \frac{\bm{B}(\bm{y}_{ij}) \bm{B}(\bm{y}_{ij})^\top}{(\bm{\mu}_{\bm{\theta}_{km}}^{*\top} \bm{B}(\bm{y}_{ij}))^2}.
\end{align*}
Finally, the variational covariance of $\bm \theta_{km}$ is analytically updated as $\bm{\Sigma}_{\bm{\theta}_{km}}^* = [-\bm H(\bm{\mu}_{\bm{\theta}_{km}}^*)]^{-1}$.

\subsection{Theoretical Support}
\label{subsec:theory}

We now provide theoretical support for the constrained LA in Section~\ref{subsec:laplace}. 
This requires addressing three questions.
The first question is whether restricting the mode-finding problem to
\(D_\delta\) is justified, that is, whether boundary and sign-changing
coefficients can be excluded. Near the zero boundary of the positive
chamber, the logarithmic criterion becomes singular and a standard
Laplace expansion may become nonregular. The
squared link is invariant under \(\bm\theta\!\mapsto\!-\bm\theta\), so
the positive and negative chambers represent the same intensity, whereas
sign-changing coefficients may create nodal-line pathologies. We show
that, under mild regularity conditions, the population objective is
maximized only at a sign-symmetric pair of coefficients that lies in the
positive and negative chambers and is separated from the zero boundary,
and that the empirical maximizer is asymptotically close to this pair.

The second question is whether the excluded boundary and sign-changing
coefficients carry asymptotically negligible
mass. Mode separation alone is insufficient, since a region with a
smaller supremum can still carry nonnegligible weight.
We show that the exponential weight of these coefficients is
asymptotically negligible relative to that of the positive chamber.
These results justify using the local mass associated with the constrained mode
for the Laplace variational update. 

The third question is whether, within the positive chamber, the exact
coordinate update is accurately approximated by the Gaussian at the mode. The first two answers locate the
mode and the relevant mass, but not the shape of the distribution
within \(D_\delta\). We show that these two distributions are
asymptotically indistinguishable in total variation, uniformly over the
admissible responsibilities and expected precisions.

Throughout this subsection, we fix cluster--mark block $(k,m)$ and suppress the block indices. Define
$a := \sum_{i=1}^n \nu_{ik} T_i$, $\eta := \mathbb{E}_q(\tau_{km}^{-2})$, $\lambda_P := {\eta}/({2a})$, $
\bm M := \int_{\mathcal B} \bm B(\bm y)\bm B(\bm y)^\top d\bm{y}$. 
Assume that, for each $i$, $\bm Y_{im}=\{\bm y_{ij}:j\in\mathcal I_{im}\}$ is a realization of a Poisson process on $\mathcal B$ with intensity $T_i\lambda_i$, where $\lambda_i$ is a fixed baseline intensity. If the data follow the model, $\lambda_i=\lambda_{z_im}$ for the true assignment $z_i$. Asymptotic statements refer to $T_{\min}:=\min_iT_i\to\infty$ with $n$ and $d$ fixed.

For $a>0$, the normalized weights $w:=(\nu_{ik}T_i/a)_{i=1}^n$ lie in the simplex $\Delta_{n\!-\!1}\!:=\!\{w\!\in\![0,\!1]^n\!:\!\sum_i \!w_i\!=\!1\}$. Since the data term of $J/a$ is the $w$-weighted average of the per-exposure subject sums $2T_i^{-1}\sum_{j\in\mathcal I_{im}}\log|s_{\bm\theta}(\bm y_{ij})|$, the update targets the weighted intensity $\lambda_{0,w}:=\sum_iw_i\lambda_i$. Population quantities below are deterministic for each fixed $w$; we write $\lambda_0=\lambda_{0,w}$, $Q=Q_w$, and $\bm\theta^\star=\bm\theta^\star(w)$ when $w$ is clear, and keep the index $w$ otherwise (e.g., $v^\star_w$, $\omega_w$, $G_\delta(w)$).

Because $(\bm\nu,\eta)$ is computed from the observations that enter $J$, it is data-dependent, although it is held fixed in the maximization over $\bm\theta$. We therefore state all results uniformly over a deterministic class, so that they apply to any data-dependent $(\bm\nu,\eta)$ in 
$\mathcal A\!:=\!\bigl\{(\bm\nu,\!\eta):\bm\nu\!\in\![0,1]^n,\ a\!\ge\!\underline a,\ w\!\in\!\mathcal W,\ 0\!\le\!\eta\!\le\! C_\eta\bigr\}$,
where $C_\eta\!:=\!(a_0\!+\!d/2)/b_0$, $\underline a\!>\!0$ is deterministic with $\underline a\!\to\!\infty$, and $\mathcal W\!\subseteq\!\Delta_{n\!-\!1}$ is fixed, nonempty, and compact. 

All regularity conditions, lemmas, proofs, and verification for B-splines are deferred to Appendix~\ref{app:local-theory}.

\makeatletter
\@ifundefined{c@maintheorem}{\newcounter{maintheorem}}{}
\makeatother
\setcounter{maintheorem}{0}
\renewcommand{\themaintheorem}{\arabic{maintheorem}}

\subsubsection{Separation and Consistency} 
 
We define the normalized version of $J$ on $\bar{\mathbb B}(\bm 0\!,\!R)$, where \(\bar{\mathbb B}(\bm a,\!r\!)\! \coloneqq \!\{\bm u\!\in\!\mathbb R^d:\|\bm u\!-\!\bm a\|_2 \!\leq\! r\}\) denotes the closed ball:
\begin{equation*}
\widetilde J(\bm \theta) \coloneqq -\bm \theta^\top (\bm M + \lambda_P\,\bm\Omega)\bm \theta + \dfrac{2}{a}\displaystyle\sum_{i=1}^n \nu_{ik}\!\sum_{j\in\mathcal I_{im}}\!\log|s_{\bm{\theta}}(\bm{y}_{ij})|,
\end{equation*}
with $\widetilde J(\bm \theta) \coloneqq -\infty$ whenever
$s_{\bm\theta}(\bm y_{ij}) = 0$ for some $i \in \{1,\ldots,n\}$ with $\nu_{ik}>0$ and
$j \in \mathcal I_{im}$.
The population criterion associated with $\widetilde J$, for each fixed $w$, is
\begin{align*}
Q(\bm \theta) \coloneqq
-\bm \theta^\top (\bm M + \bar\lambda\bm\Omega) \bm\theta + 2\displaystyle\int_{\mathcal B}\lambda_0(\bm y)\log|s_{\bm{\theta}}(\bm y)|\,d\bm{y},
\end{align*}
with $Q(\bm \theta)\!\coloneqq\!-\!\infty$ for $\lambda_0(\cdot)|\log|s_{\bm\theta}(\cdot)||\notin L^1(\mathcal B)$,
where 
$\lambda_0=\lambda_{0,w}$ denotes the weighted baseline intensity, and $\bar\lambda\ge 0$ is the deterministic population benchmark of
the empirical penalty level $\lambda_P$.
We only require that
\(\lambda_P\) be close to \(\bar\lambda\) 
under the penalty discrepancy condition in Assumption~\ref{ass:app-fullsign-penalty}. In standard settings,
where the normalized prior penalty vanishes with the exposure
scale, one may take \(\bar\lambda\!=\!0\).

We define the positive chamber $D\coloneqq\{\bm\theta\in\bar{\mathbb B}(\bm 0,R):\inf_{\mathcal B}s_{\bm\theta}>0\}$, the positive boundary strip $D_\delta^{\mathrm{bd}}\coloneqq\{\bm\theta\in\bar{\mathbb B}(\bm 0,R):0<\inf_{\mathcal B}s_{\bm\theta}<\delta\}$, 
and the sign-symmetric counterparts $-D$, $-D_\delta$, and $-D_\delta^{\mathrm{bd}}$.
We further introduce the sign-changing region $S^{\mathrm{sc}}\!\coloneqq\!
\left\{\!\bm \theta\in\mathbb R^d\!:\!\|\bm\theta\|_2\!\le\!R,\ \sup_{\bm{y}\in\mathcal B}\!s_{\bm\theta}(\bm y)\!>\!0,\ \inf_{\bm{y}\in\mathcal B}\!s_{\bm\theta}(\bm y)\!<\!0\!\right\}$.
The full comparison class is $\Theta^{\pm,\mathrm{sc}} \coloneqq D\cup(-D)\cup S^{\mathrm{sc}}$, and the inadmissible set is $\mathcal X_\delta \coloneqq S^{\mathrm{sc}}\cup D_\delta^{\mathrm{bd}}\cup(-D_\delta^{\mathrm{bd}})$, so that $\Theta^{\pm,\mathrm{sc}} = \left(D_\delta\cup(-D_\delta)\right)\cup \mathcal X_\delta$.
 
The first result supports the LA centered at the 
positive chamber mode. 
The positive chamber $D$ contains sequences approaching
the zero boundary, where a quadratic Laplace expansion may
become nonregular. If \(s_{\bm\theta}\) changes sign, then it crosses zero and
creates nodal lines, which are known to cause instability in the
squared-link formulation \citep{John2018LargeScaleCP}. So the boundary strip $D_\delta^{\mathrm{bd}}$ and the sign-changing region $S^{\mathrm{sc}}$ must be
excluded. The following theorem therefore
compares the inadmissible set $\mathcal X_\delta$ with the positive and negative
chambers. It first identifies the population maximizers over
$\Theta^{\pm,\mathrm{sc}}$ and shows that they are separated from the zero
boundary and sign-changing regions, lying in $D_\delta\cup(-D_\delta)$. It then studies the
empirical mode computed over $\Theta^{\pm,\mathrm{sc}}$. Since the two
chambers are sign-symmetric representatives of the same squared-link
intensity, ruling out $\mathcal X_\delta$ leaves the strictly positive chamber
constraint to select one identifiable representative.

\refstepcounter{maintheorem}\noindent\textbf{Theorem~\themaintheorem\ (Empirical Separation of Inadmissible Coefficients and Mode Consistency).}\label{thm:main-chamber-gap}
\textit{
Assume Assumptions~\ref{ass:app-growth}, \ref{ass:app-fullsign-penalty}, and~\ref{ass:app-uniform-family} hold, with the penalty convergence uniform over $\mathcal A$. Then the maximization of $Q$ over $\Theta^{\pm,\mathrm{sc}}$ has exactly two global maximizers $\pm\bm\theta^\star$, with $\bm\theta^\star\in D_\delta$ and $-\bm\theta^\star\in -D_\delta$. Moreover,}
\begin{align}
\sup_{\bm\theta\in \mathcal X_\delta}\widetilde J(\bm\theta)
<
\sup_{\bm\theta\in D_\delta\cup(-D_\delta)}\widetilde J(\bm\theta)
\label{eq:main-sample-gap}
\end{align}
\textit{
simultaneously for all $(\bm\nu,\eta)\in\mathcal A$ with probability tending to one. Consequently every maximizer of $\widetilde J$ over $\Theta^{\pm,\mathrm{sc}}$ belongs to $D_\delta\cup(-D_\delta)$ with probability tending to one. And}
\begin{align}
\sup_{(\bm\nu,\eta)\in\mathcal A}
\operatorname{dist}
\!\left(
\arg\max_{\bm \theta\in \Theta^{\pm,\mathrm{sc}}}\widetilde J(\bm\theta),
\{\pm \bm\theta^\star\}
\right)
\xrightarrow{p} 0,
\label{eq:main-chamber-gap-argmax}
\end{align}
\textit{where $\operatorname{dist}(A,B)\!\coloneqq\!\sup_{\bm a\in A}\inf_{\bm b\in B}\!\|\bm a\!-\!\bm b\|_2$ denotes the one-sided Hausdorff distance from set $A$ to set $B$. In particular, $\widehat{\bm\theta}_\delta\!=\!\arg\max_{\bm\theta\in D_\delta}\widetilde J(\bm\theta)$ satisfies
$\|\widehat{\bm\theta}_\delta-\bm\theta^\star(w)\|_2\xrightarrow{p}0$ uniformly over $\mathcal A$.}

Theorem~\ref{thm:main-chamber-gap} first shows that the population
maximizers over $\Theta^{\pm,\mathrm{sc}}$ are separated from the zero
boundary and lie in $D_\delta\cup(-D_\delta)$. It then shows that, on the
full comparison class $\Theta^{\pm,\mathrm{sc}}$, neither sign-changing
coefficients nor boundary-strip coefficients can contain an empirical
maximizer with probability tending to one. The result covers
$\bar{\mathbb B}(\bm 0,R)$ up to a measure-zero boundary. Moreover, every
maximizer over $\Theta^{\pm,\mathrm{sc}}$ is asymptotically close to either
$\bm\theta^\star$ or $-\bm\theta^\star$, and the positive chamber
constraint selects the intended representative of this sign-symmetric
pair. Thus, optimizing over \(D_\delta\) preserves the asymptotic
positive chamber mode and provides a regular Laplace center.

\subsubsection{Empirical Dominance} 

Theorem~\ref{thm:main-chamber-gap} rules out sign-changing and boundary coefficients
as empirical mode candidates. However, mode separation alone does not
show that the sign-changing chamber and boundary are negligible for the distributional
approximation. A region may have a lower supremum of the objective but
still contribute nonnegligible weight if its volume is large
enough. In that case, approximating only the local mass around the
positive chamber mode could miss a substantial contribution from
sign-changing and boundary regions.

The next result strengthens the mode-separation statement to an
exponential-weight comparison. It shows that the unnormalized weights of
the sign-changing region and boundary are asymptotically negligible relative to the
positive chamber. Therefore, excluding sign-changing and boundary configurations does
not discard asymptotically relevant mass, and the LA
centered at the positive chamber mode captures the relevant local
contribution of the positive representative.

To state the result, 
let \(\mathbb B(\bm a,r) \coloneqq \{\bm u\in\mathbb R^d:\|\bm u-\bm a\|_2<r\}\) denote the open ball, and 
define
\begin{equation}\label{eq:app-omega}
\begin{aligned}
v^\star(r) &\coloneqq \left|D_\delta\cap \mathbb B(\bm\theta^\star,r)\right|,\\[-1pt]
\omega(r) &\coloneqq \sup_{\bm\theta\in D_\delta\cap \mathbb B(\bm\theta^\star,r)} \left|Q(\bm\theta)-Q(\bm\theta^\star)\right|.
\end{aligned}
\end{equation}
Let \(\widetilde J_{\delta}\) and \(Q_{\delta}\) denote the truncated versions of
\(\widetilde J\) and \(Q\), obtained by replacing $\log |s_{\bm\theta}(\cdot)|$ with $\max\{\log |s_{\bm\theta}(\cdot)|,\, c_{\mathrm{tr}}\log\delta\}$. 
Fix $\delta\in(0,1)$ and $c_{\mathrm{tr}}>1$, with $c_{\mathrm{tr}}$ chosen large enough so that
$G_\delta := \min\{G^{\mathrm{sc}}_\delta,\, G^{\mathrm{bd}}_\delta\} > 0$, where $G^{\mathrm{sc}}_{\delta}
:= Q(\bm{\theta}^\star) -
\sup_{\bm{\theta}\in S^{\mathrm{sc}}}Q_{\delta}(\bm{\theta}),\
G^{\mathrm{bd}}_{\delta}
:= Q(\bm{\theta}^\star) -
\sup_{\bm{\theta}\in D_\delta^{\mathrm{bd}}}Q_{\delta}(\bm{\theta})$.
Define
$\Xi_{\delta,a} \coloneqq \sup_{\bm\theta\in \bar{\mathbb B}(\bm 0,R)} | \widetilde J_{\delta}(\bm\theta) - Q_{\delta}(\bm\theta) |$.
For the uniform statements, write
$\underline G_\delta:=\inf_{w\in\mathcal W}G_\delta(w)>0$.

Finally, for measurable $C\!\subseteq\!\bar{\mathbb B}(\bm 0,R)$ define the empirical exponential weight as
$\widehat Z_a(C)\!\coloneqq\!\int_C \exp\!\left\{a\,\widetilde J(\bm\theta)\right\}d\bm\theta$.

\refstepcounter{maintheorem}\noindent\textbf{Theorem~\themaintheorem\ (Empirical Dominance of the Admissible Chambers).}\label{thm:main-dominance}
\textit{Assume the conditions of Theorem~\ref{thm:main-chamber-gap} and Proposition~\ref{prop:app-pop-dominance-truncated} hold. Then, for every fixed $r>0$ with $v^\star(r)>0$ and every measurable $C\subseteq \mathcal X_\delta$,}
\begin{align}
\frac{\widehat Z_a(C)}{\widehat Z_a(D_\delta)}
\le
\frac{|C|}{v^\star(r)}
\exp\Bigl\{
-a\bigl(G_{\delta}-\omega(r)-2\Xi_{\delta,a}\bigr)
\Bigr\}
\label{eq:main-sample-ratio-bound}
\end{align}
\textit{on the event $\{2\Xi_{\delta,a}<G_{\delta}-\omega(r)\}$. 
The bound is pathwise for every admissible pair. In particular, for a fixed radius $r_0>0$ satisfying $\inf_w v_w^\star(r_0)>0$ and $\sup_w\omega_w(r_0)<\underline G_\delta$,}
\begin{align}
\sup_{(\bm\nu,\eta)\in\mathcal A}
\frac{\widehat Z_a(\mathcal X_\delta)}{\widehat Z_a(D_\delta)}
\xrightarrow{p} 0.
\label{eq:main-sample-ratio-goes-zero}
\end{align}
\textit{The same conclusion holds with $D_\delta$ replaced by $-D_\delta$.
By Proposition~\ref{prop:app-admissible-radius}, there is a common $\underline r^\star(\delta)>0$ such that these conditions hold for every $r_0\in(0,\underline r^\star(\delta))$.}


Theorem~\ref{thm:main-dominance} provides the mass-level counterpart to Theorem~\ref{thm:main-chamber-gap}. It shows that the exponential weight of the boundary strip and the sign-changing region is asymptotically negligible relative to the strictly positive chamber. Thus, after selecting the positive representative, the constrained LA does not omit asymptotically relevant sign-changing contribution.

\subsubsection{Validity of the Laplace Approximation}
Theorems~\ref{thm:main-chamber-gap} and~\ref{thm:main-dominance} show that, uniformly over $\mathcal A$, the inadmissible set $\mathcal X_\delta$ contains no empirical mode with probability tending to one and carries asymptotically negligible exponential weight relative to $D_\delta$. It remains to ask whether the mass within $D_\delta$ is accurately approximated by the Gaussian at the constrained mode. We therefore compare the exact coordinate update restricted to the positive chamber, $q^*_{a,+}(\bm\theta):=\frac{\exp\{J(\bm\theta)\}\mathbf1_{\{\bm\theta\in D_\delta\}}}{\int_{D_\delta}\exp\{J(\bm u)\}\,d\bm u}$, with the Laplace update of Section~\ref{subsec:laplace}, $q_a^{\mathrm L}:=\operatorname{Normal}\!\left(\widehat{\bm\theta}_\delta,[-\bm H(\widehat{\bm\theta}_\delta)]^{-1}\right)$.
Both depend on $(\bm\nu,\eta)$ through $J$.
The next result shows that the two distributions are asymptotically indistinguishable in total variation, uniformly over $\mathcal A$.

\refstepcounter{maintheorem}\noindent\textbf{Theorem~\themaintheorem\ (Uniform Total-Variation Accuracy of the Positive-Chamber Laplace Approximation).}\label{thm:main-laplace}
\textit{Under the conditions of Theorem~\ref{thm:main-dominance},}
\begin{equation}
\sup_{(\bm\nu,\eta)\in\mathcal A}
\left\|q^*_{a,+}-q_a^{\mathrm L}\right\|_{\mathrm{TV}}
\xrightarrow{p}0,
\label{eq:main-laplace-tv}
\end{equation}
\textit{where $\|\cdot\|_{\mathrm{TV}}$ denotes total variation distance.}

In words, \eqref{eq:main-laplace-tv} states that $q_a^{\mathrm L}$ and $q^*_{a,+}$ assign every event the same probability up to an $o_p(1)$ term that is uniform over all events and over $\mathcal A$.
Theorem~\ref{thm:main-laplace} thus completes the justification of the constrained Laplace step, whose center and mass are controlled by Theorems~\ref{thm:main-chamber-gap} and~\ref{thm:main-dominance}. Since \eqref{eq:main-laplace-tv} is uniform over $\mathcal A$, it holds simultaneously at all VI iterations whose $(\bm\nu,\eta)$ lies in $\mathcal A$, including a data-dependent stopping iteration, and at any fixed point of the iterations that lies in $\mathcal A$.

By Corollary~E.1, with $\mathcal W=\Delta_{n-1}$ and $\underline a=T_{\min}$, every coefficient update of an active block ($\sum_i\nu_{ik}>1$) in Algorithm~\ref{alg:cavi} is such an iteration.

\subsection{Inference Algorithm}
\label{subsec:algorithm}

Our inference algorithm is summarized in Appendix~\ref{app:vi-algorithm}.
While $(\bm{\phi}, \bm{z}, \bm{\tau^2})$ are updated via standard closed-form coordinate ascent steps, 
$\bm\theta$ are updated using the constrained LA. 


\section{EXPERIMENTS AND ANALYSIS}
\label{sec:experiments}

\paragraph{Synthetic Data: Intensity Estimation Experiment.}
We first isolate the intensity model from the mixture, comparing 
our estimator against five Poisson process baselines
: INLA \citep{Rue2009INLA}, LGCP \citep{Taylor2013lgcp}, SGCP \citep{Donner2018Sigmoidal}, 
STVB \citep{Aglietti2019Structured}, and  VBPP \citep{lloyd2015variational}. 
Six intensity surfaces are each observed at four expected counts ($n_{\mathrm{exp}}$) between $500$ and $50\,000$,
with $100$ replicates per configuration.
Figure~\ref{fig:synthetic results}(a) shows the true intensity surfaces,
the relative integrated squared error (RISE), pointwise coverage rate, and computation time. 
Detailed definitions are deferred to Appendix~\ref{app:intensity}. 
The proposed method attains the lowest RISE in 22 of the 24 configurations; 
the two exceptions both occur in Scenario~D
, where the gap to the best method is within $1.4$ standard errors.
It also yields a pointwise coverage closest to the nominal level (0.94 on average), 
and has the lowest mean computation time, running faster than INLA, the next
fastest method, in 18 of the 24 configurations, and remaining comparable in the other 6.

\begin{figure*}[t!]
  \centering
  \includegraphics[width=\textwidth]{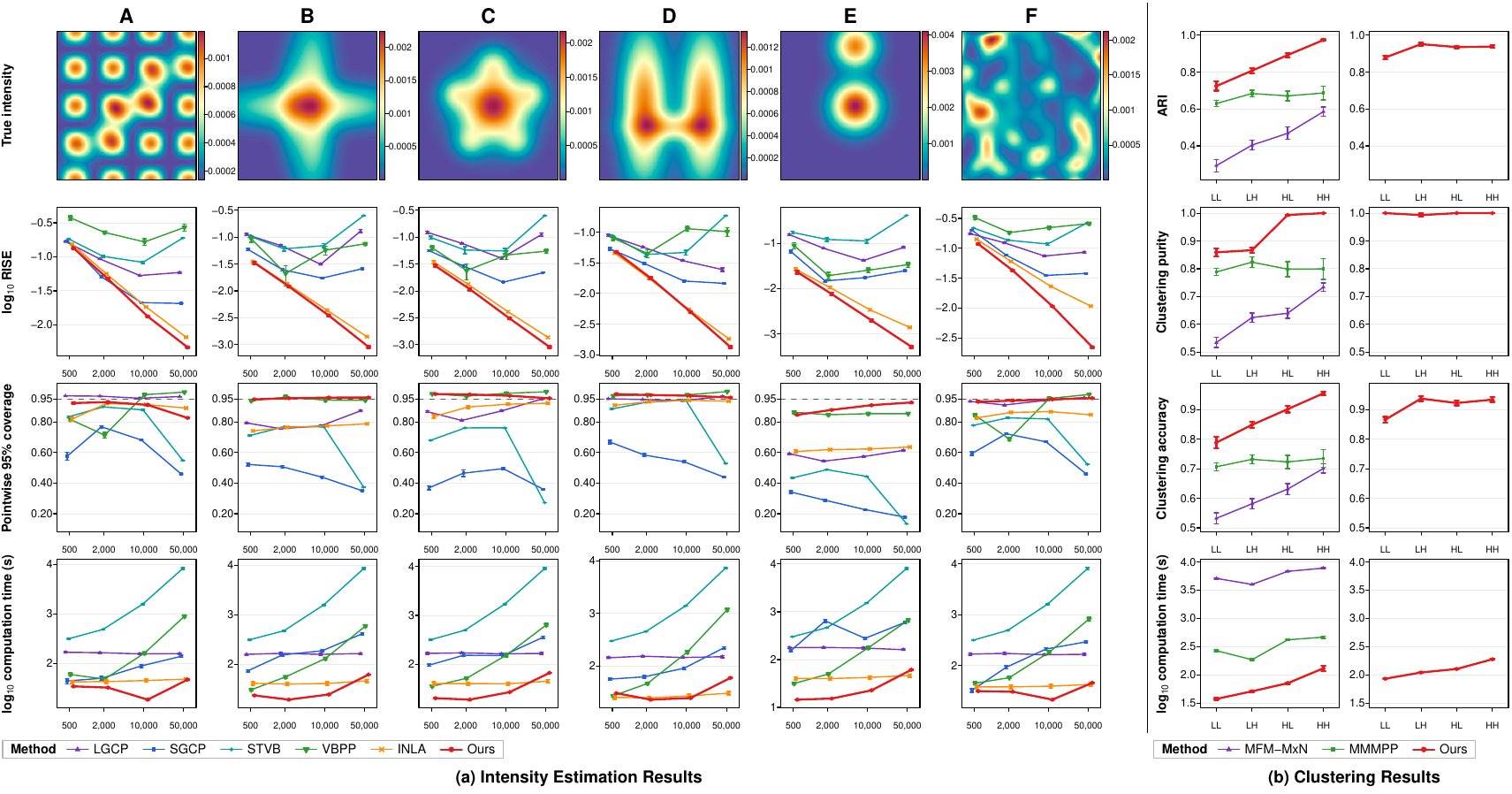}
  \caption{Synthetic experiment results. (a)~Intensity estimation: columns are the 
  six scenarios, and the top row shows true surfaces. The horizontal axis is
  $n_{\mathrm{exp}}$. (b)~Clustering results: unmarked comparison (left)
  and marked experiment (right). In all panels, points are means over $100$ replicates; whiskers
  are $95\%$ confidence intervals.}
  \label{fig:synthetic results}
\end{figure*}

\paragraph{Synthetic Data: Clustering Experiment.}
We evaluate subject-level clustering against MM-MPP \citep{Yin2021RowClustering}
and MFM-MxN \citep{Yin2023MatrixClustering}, without marks. 
Six clusters defined by the same intensity surfaces of the previous experiment
are observed under four regimes (HH, HL, LH, LL), with $100$ replicates each; 
the first letter denotes a high (H) or low (L) number of subjects and the second denotes 
a high or low expected number of observations per subject.
The left column of Figure~\ref{fig:synthetic results}(b) reports the adjusted Rand index (ARI) \citep{Hubert1985ARI},
clustering purity, clustering accuracy, and computation time, with $95\%$ confidence intervals.
Further detailed design and definitions are given in Appendix~\ref{app:clustering}. 
The proposed model ranks first in every regime and metric, 
improving ARI over the better baseline by $0.10$ to $0.29$.

\paragraph{Synthetic Data: Marked Point Process Clustering Experiment.}
The previous experiment omitted marks for a fair comparison with baseline models; here we fit the full marked model. Six clusters, A--B, B--A, C--D, D--C, E--F and F--E, use the six surfaces in the top row of Figure~\ref{fig:synthetic results}(a) for marks $0$ and $1$ (e.g., A--B uses A for mark $0$ and B for mark $1$), so the two clusters of each pair have identically distributed event locations and differ only in their mark probabilities. Across four regimes (HH, HL, LH, LL), with $100$ replicates each (right column of Figure~\ref{fig:synthetic results}(b)), ARI ranges from $0.88$ to $0.95$. 
Details are in Appendix~\ref{app:marked}.

\paragraph{Real-World Data.}

We also apply the proposed model to National Basketball Association 
shot-chart data as a real-world example of replicated marked point-process clustering, where each player is represented by shot locations with make--miss marks.
The result yields 13 major clusters corresponding to interpretable offensive archetypes, such as frontcourt players, ace scorers, and shooters, while further separating players within each role by their spatial shot allocation and location-specific success probabilities.
The data description, fitting details, estimated intensity and success-probability surfaces, representative players, and interpretation are provided in Appendix~\ref{app:nba_real_data}.

\section{CONCLUSIONS}
\label{sec:conclusion}

We proposed DPM-MPPP, a Bayesian nonparametric model that jointly infers cluster assignments, the number of clusters, and mark-dependent intensity surfaces without discretization. For posterior inference, we developed, to our knowledge, the first variational Bayes method for clustering MPPPs. Its constrained Laplace step turns the nonconjugate basis-coefficient update into a well-posed convex problem, eliminating the sign ambiguity and nodal-line pathologies of squared-link models, and our theory guarantees mode consistency and exponential-weight concentration on the feasible region, together with uniform convergence of the positive-chamber Laplace update in total variation. In synthetic experiments, our method outperforms existing methods in both accuracy and speed, resolving the fidelity--efficiency trade-off, and an NBA shot-chart analysis uncovers interpretable offensive archetypes.

A practical limitation is that, as with many Bayesian nonparametric mixture models, DPM-MPPP can mildly over-split clusters. 
A natural extension is to continuous mark distributions, allowing the framework to model richer event-level attributes.




\renewcommand{\bibsection}{\subsubsection*{References}}
\bibliographystyle{plainnat} 
\bibliography{references}    

@inproceedings{Adams2009Tractable,
  author    = {Adams, Ryan Prescott and Murray, Iain and MacKay, David J. C.},
  title     = {Tractable Nonparametric {Bayesian} Inference in {Poisson} Processes with {Gaussian} Process Intensities},
  booktitle = {Proceedings of the 26th Annual International Conference on Machine Learning},
  pages     = {9--16},
  year      = {2009},
  publisher = {ACM},
  doi       = {10.1145/1553374.1553376}
}

@article{Donner2018Sigmoidal,
  author  = {Donner, Christian and Opper, Manfred},
  title   = {Efficient {Bayesian} Inference of Sigmoidal {Gaussian} {Cox} Processes},
  journal = {Journal of Machine Learning Research},
  year    = {2018},
  volume  = {19},
  number  = {67},
  pages   = {1--34}
}

@inproceedings{Lloyd2016LPPA,
  author    = {Lloyd, Chris and Gunter, Tom and Osborne, Michael and Roberts, Stephen and Nickson, Tom},
  title     = {Latent Point Process Allocation},
  booktitle = {Proceedings of the 19th International Conference on Artificial Intelligence and Statistics},
  series    = {Proceedings of Machine Learning Research},
  volume    = {51},
  pages     = {389--397},
  year      = {2016},
  publisher = {PMLR}
}

@inproceedings{Ding2018BaNPPA,
  author    = {Ding, Hongyi and Khan, Mohammad and Sato, Issei and Sugiyama, Masashi},
  title     = {{Bayesian} Nonparametric {Poisson}-Process Allocation for Time-Sequence Modeling},
  booktitle = {Proceedings of the 21st International Conference on Artificial Intelligence and Statistics},
  series    = {Proceedings of Machine Learning Research},
  volume    = {84},
  pages     = {1108--1116},
  year      = {2018},
  publisher = {PMLR}
}

@article{Barrack2017NHPP,
  author  = {Barrack, Duncan and Preston, Simon},
  title   = {Classification and Clustering for Observations of Event Time Data Using Non-Homogeneous {Poisson} Process Models},
  journal = {arXiv preprint arXiv:1703.02111},
  year    = {2017},
  doi     = {10.48550/arXiv.1703.02111}
}

@article{Phung2017RFS,
  author  = {Phung, Dinh and Vo, Ba-Ngu},
  title   = {A Random Finite Set Model for Data Clustering},
  journal = {arXiv preprint arXiv:1703.04832},
  year    = {2017},
  doi     = {10.48550/arXiv.1703.04832}
}

@inproceedings{Yin2021RowClustering,
  author    = {Yin, Lihao and Xu, Ganggang and Sang, Huiyan and Guan, Yongtao},
  title     = {Row-Clustering of a Point Process-Valued Matrix},
  booktitle = {Advances in Neural Information Processing Systems},
  volume    = {34},
  pages     = {20028--20039},
  year      = {2021},
  publisher = {Curran Associates, Inc.}
}

@inproceedings{Xu2017DMHP,
  author    = {Xu, Hongteng and Zha, Hongyuan},
  title     = {A {Dirichlet} Mixture Model of {Hawkes} Processes for Event Sequence Clustering},
  booktitle = {Advances in Neural Information Processing Systems},
  volume    = {30},
  pages     = {1354--1363},
  year      = {2017},
  publisher = {Curran Associates, Inc.}
}

@inproceedings{Zhang2022NTPPMix,
  author    = {Zhang, Yunhao and Yan, Junchi and Zhang, Xiaolu and Zhou, Jun and Yang, Xiaokang},
  title     = {Learning Mixture of Neural Temporal Point Processes for Multi-Dimensional Event Sequence Clustering},
  booktitle = {Proceedings of the 31st International Joint Conference on Artificial Intelligence},
  pages     = {3766--3772},
  year      = {2022},
  publisher = {International Joint Conferences on Artificial Intelligence Organization},
  doi       = {10.24963/ijcai.2022/523}
}

@inproceedings{Ding2023CNTPP,
  author    = {Ding, Fangyu and Yan, Junchi and Wang, Haiyang},
  title     = {C-{NTPP}: Learning Cluster-Aware Neural Temporal Point Process},
  booktitle = {Proceedings of the 37th AAAI Conference on Artificial Intelligence},
  pages     = {7369--7377},
  year      = {2023},
  publisher = {AAAI Press},
  doi       = {10.1609/aaai.v37i6.25897}
}

@article{Reich2006Spatial,
  author  = {Reich, Brian J. and Hodges, James S. and Carlin, Bradley P. and Reich, Adam M.},
  title   = {A Spatial Analysis of Basketball Shot Chart Data},
  journal = {The American Statistician},
  year    = {2006},
  volume  = {60},
  number  = {1},
  pages   = {3--12},
  doi     = {10.1198/000313006X90305}
}

@inproceedings{Miller2014Factorized,
  author    = {Miller, Andrew and Bornn, Luke and Adams, Ryan and Goldsberry, Kirk},
  title     = {Factorized Point Process Intensities: A Spatial Analysis of Professional Basketball},
  booktitle = {Proceedings of the 31st International Conference on Machine Learning},
  series    = {Proceedings of Machine Learning Research},
  volume    = {32},
  pages     = {235--243},
  year      = {2014},
  publisher = {PMLR}
}

@inproceedings{Yin2022SpatialHomogeneity,
  author    = {Yin, Fan and Jiao, Jieying and Yan, Jun and Hu, Guanyu},
  title     = {{Bayesian} Nonparametric Learning for Point Processes with Spatial Homogeneity: A Spatial Analysis of {NBA} Shot Locations},
  booktitle = {Proceedings of the 39th International Conference on Machine Learning},
  series    = {Proceedings of Machine Learning Research},
  volume    = {162},
  pages     = {25523--25551},
  year      = {2022},
  publisher = {PMLR}
}

@article{Jiao2021Bayesian,
  author  = {Jiao, Jieying and Hu, Guanyu and Yan, Jun},
  title   = {A {Bayesian} Marked Spatial Point Processes Model for Basketball Shot Chart},
  journal = {Journal of Quantitative Analysis in Sports},
  year    = {2021},
  volume  = {17},
  number  = {2},
  pages   = {77--90},
  doi     = {10.1515/jqas-2019-0106}
}

@article{Hu2021GroupLearning,
  author  = {Hu, Guanyu and Yang, Hou-Cheng and Xue, Yishu},
  title   = {{Bayesian} Group Learning for Shot Selection of Professional Basketball Players},
  journal = {Stat},
  year    = {2021},
  volume  = {10},
  number  = {1},
  pages   = {e324},
  doi     = {10.1002/sta4.324}
}

@article{Yin2023MatrixClustering,
  author  = {Yin, Fan and Hu, Guanyu and Shen, Weining},
  title   = {Analysis of Professional Basketball Field Goal Attempts via a {Bayesian} Matrix Clustering Approach},
  journal = {Journal of Computational and Graphical Statistics},
  year    = {2023},
  volume  = {32},
  number  = {1},
  pages   = {49--60},
  doi     = {10.1080/10618600.2022.2085727}
}

@article{Hu2023ZIPClustered,
  author  = {Hu, Guanyu and Yang, Hou-Cheng and Xue, Yishu and Dey, Dipak K.},
  title   = {Zero-Inflated {Poisson} Model with Clustered Regression Coefficients: Application to Heterogeneity Learning of Field Goal Attempts of Professional Basketball Players},
  journal = {Canadian Journal of Statistics},
  year    = {2023},
  volume  = {51},
  number  = {1},
  pages   = {157--172},
  doi     = {10.1002/cjs.11684}
}

@article{WongToi2023Joint,
  author  = {Wong-Toi, Eliot and Yang, Hou-Cheng and Shen, Weining and Hu, Guanyu},
  title   = {A Joint Analysis for Field Goal Attempts and Percentages of Professional Basketball Players: {Bayesian} Nonparametric Resource},
  journal = {Journal of Data Science},
  year    = {2023},
  volume  = {21},
  number  = {1},
  pages   = {68--86},
  doi     = {10.6339/22-JDS1062}
}

@inproceedings{Aglietti2019Structured,
  author    = {Aglietti, Virginia and Bonilla, Edwin V. and Damoulas, Theodoros and Cripps, Sally},
  title     = {Structured Variational Inference in Continuous {Cox} Process Models},
  booktitle = {Advances in Neural Information Processing Systems},
  volume    = {32},
  pages     = {12458--12468},
  year      = {2019},
  publisher = {Curran Associates, Inc.}
}

@article{Li2019TumorPathology,
  author  = {Li, Qiwei and Wang, Xinlei and Liang, Faming and Xiao, Guanghua},
  title   = {A {Bayesian} Mark Interaction Model for Analysis of Tumor Pathology Images},
  journal = {The Annals of Applied Statistics},
  year    = {2019},
  volume  = {13},
  number  = {3},
  pages   = {1708--1732},
  doi     = {10.1214/19-AOAS1254}
}

@article{Martinez2010MarkConnection,
  author  = {Mart{\'\i}nez, Vicent J. and Arnalte-Mur, Pablo and Stoyan, Dietrich},
  title   = {Measuring Galaxy Segregation with the Mark Connection Function},
  journal = {Astronomy \& Astrophysics},
  year    = {2010},
  volume  = {513},
  pages   = {A22},
  doi     = {10.1051/0004-6361/200912922}
}

@article{Mohler2014MarkedCrime,
  author  = {Mohler, George},
  title   = {Marked Point Process Hotspot Maps for Homicide and Gun Crime Prediction in {Chicago}},
  journal = {International Journal of Forecasting},
  year    = {2014},
  volume  = {30},
  number  = {3},
  pages   = {491--497},
  doi     = {10.1016/j.ijforecast.2014.01.004}
}

@book{kingman1993poisson,
  author    = {Kingman, John Frank Charles},
  title     = {Poisson Processes},
  series    = {Oxford Studies in Probability},
  publisher = {Clarendon Press},
  year      = {1993},
  doi       = {10.1093/oso/9780198536932.001.0001}
}

@article{prekopa1958secondary,
  author  = {Pr{\'e}kopa, Andr{\'a}s},
  title   = {On Secondary Processes Generated by a Random Point Distribution of {Poisson} Type},
  journal = {Annales Universitatis Scientiarum Budapestinensis de Rolando E{\"o}tv{\"o}s Nominatae, Sectio Mathematica},
  year    = {1958},
  volume  = {1},
  pages   = {153--170}
}

@book{daley2003introduction,
  author    = {Daley, Daryl J. and Vere-Jones, David},
  title     = {An Introduction to the Theory of Point Processes: Volume {I}: Elementary Theory and Methods},
  series    = {Probability and Its Applications},
  edition   = {2nd},
  publisher = {Springer},
  year      = {2003},
  doi       = {10.1007/b97277}
}

@article{Ferguson1973,
  author  = {Ferguson, Thomas S.},
  title   = {A {Bayesian} Analysis of Some Nonparametric Problems},
  journal = {The Annals of Statistics},
  year    = {1973},
  volume  = {1},
  number  = {2},
  pages   = {209--230},
  doi     = {10.1214/aos/1176342360}
}

@article{Antoniak1974,
  author  = {Antoniak, Charles E.},
  title   = {Mixtures of {Dirichlet} Processes with Applications to {Bayesian} Nonparametric Problems},
  journal = {The Annals of Statistics},
  year    = {1974},
  volume  = {2},
  number  = {6},
  pages   = {1152--1174},
  doi     = {10.1214/aos/1176342871}
}

@article{Sethuraman1994,
  author  = {Sethuraman, Jayaram},
  title   = {A Constructive Definition of {Dirichlet} Priors},
  journal = {Statistica Sinica},
  year    = {1994},
  volume  = {4},
  number  = {2},
  pages   = {639--650}
}

@article{Ishwaran2001,
  author  = {Ishwaran, Hemant and James, Lancelot F.},
  title   = {{Gibbs} Sampling Methods for Stick-Breaking Priors},
  journal = {Journal of the American Statistical Association},
  year    = {2001},
  volume  = {96},
  number  = {453},
  pages   = {161--173},
  doi     = {10.1198/016214501750332758}
}

@article{Neal2000,
  author  = {Neal, Radford M.},
  title   = {{Markov} Chain Sampling Methods for {Dirichlet} Process Mixture Models},
  journal = {Journal of Computational and Graphical Statistics},
  year    = {2000},
  volume  = {9},
  number  = {2},
  pages   = {249--265},
  doi     = {10.1080/10618600.2000.10474879}
}

@article{Jordan1999,
  author  = {Jordan, Michael I. and Ghahramani, Zoubin and Jaakkola, Tommi S. and Saul, Lawrence K.},
  title   = {An Introduction to Variational Methods for Graphical Models},
  journal = {Machine Learning},
  year    = {1999},
  volume  = {37},
  number  = {2},
  pages   = {183--233},
  doi     = {10.1023/A:1007665907178}
}

@article{Wainwright2008,
  author  = {Wainwright, Martin J. and Jordan, Michael I.},
  title   = {Graphical Models, Exponential Families, and Variational Inference},
  journal = {Foundations and Trends in Machine Learning},
  year    = {2008},
  volume  = {1},
  number  = {1--2},
  pages   = {1--305},
  doi     = {10.1561/2200000001}
}

@article{Blei2017,
  author  = {Blei, David M. and Kucukelbir, Alp and McAuliffe, Jon D.},
  title   = {Variational Inference: A Review for Statisticians},
  journal = {Journal of the American Statistical Association},
  year    = {2017},
  volume  = {112},
  number  = {518},
  pages   = {859--877},
  doi     = {10.1080/01621459.2017.1285773}
}

@article{blei2006variational,
  author  = {Blei, David M. and Jordan, Michael I.},
  title   = {Variational Inference for {Dirichlet} Process Mixtures},
  journal = {Bayesian Analysis},
  year    = {2006},
  volume  = {1},
  number  = {1},
  pages   = {121--143},
  doi     = {10.1214/06-BA104}
}

@book{deboor2001practical,
  author    = {de Boor, Carl},
  title     = {A Practical Guide to Splines},
  series    = {Applied Mathematical Sciences},
  volume    = {27},
  publisher = {Springer},
  year      = {2001}
}

@article{lang2004bayesian,
  author  = {Lang, Stefan and Brezger, Andreas},
  title   = {{Bayesian} {P}-Splines},
  journal = {Journal of Computational and Graphical Statistics},
  year    = {2004},
  volume  = {13},
  number  = {1},
  pages   = {183--212},
  doi     = {10.1198/1061860043010}
}

@article{moller1998log,
  author  = {M{\o}ller, Jesper and Syversveen, Anne Randi and Waagepetersen, Rasmus Plenge},
  title   = {Log {Gaussian} {Cox} Processes},
  journal = {Scandinavian Journal of Statistics},
  year    = {1998},
  volume  = {25},
  number  = {3},
  pages   = {451--482},
  doi     = {10.1111/1467-9469.00115}
}

@book{moller2004statistical,
  author    = {M{\o}ller, Jesper and Waagepetersen, Rasmus Plenge},
  title     = {Statistical Inference and Simulation for Spatial Point Processes},
  publisher = {Chapman and Hall/CRC},
  year      = {2004},
  doi       = {10.1201/9780203496930}
}

@inproceedings{John2018LargeScaleCP,
  author    = {John, S. T. and Hensman, James},
  title     = {Large-Scale {Cox} Process Inference Using Variational {Fourier} Features},
  booktitle = {Proceedings of the 35th International Conference on Machine Learning},
  series    = {Proceedings of Machine Learning Research},
  volume    = {80},
  pages     = {2362--2370},
  year      = {2018},
  publisher = {PMLR}
}

@inproceedings{lloyd2015variational,
  author    = {Lloyd, Chris and Gunter, Tom and Osborne, Michael and Roberts, Stephen},
  title     = {Variational Inference for {Gaussian} Process Modulated {Poisson} Processes},
  booktitle = {Proceedings of the 32nd International Conference on Machine Learning},
  series    = {Proceedings of Machine Learning Research},
  volume    = {37},
  pages     = {1814--1822},
  year      = {2015},
  publisher = {PMLR}
}

@article{wang2013variational,
  author  = {Wang, Chong and Blei, David M.},
  title   = {Variational Inference in Nonconjugate Models},
  journal = {Journal of Machine Learning Research},
  year    = {2013},
  volume  = {14},
  pages   = {1005--1031}
}

@article{Simpson2016GoingOffGrid,
  author  = {Simpson, Daniel and Illian, Janine B. and Lindgren, Finn and S{\o}rbye, Sigrunn H. and Rue, H{\aa}vard},
  title   = {Going Off Grid: Computationally Efficient Inference for Log-{Gaussian} {Cox} Processes},
  journal = {Biometrika},
  year    = {2016},
  volume  = {103},
  number  = {1},
  pages   = {49--70},
  doi     = {10.1093/biomet/asv064}
}

@article{Rue2009INLA,
  author  = {Rue, H{\aa}vard and Martino, Sara and Chopin, Nicolas},
  title   = {Approximate {Bayesian} Inference for Latent {Gaussian} Models by Using Integrated Nested {Laplace} Approximations},
  journal = {Journal of the Royal Statistical Society: Series B (Statistical Methodology)},
  year    = {2009},
  volume  = {71},
  number  = {2},
  pages   = {319--392},
  doi     = {10.1111/j.1467-9868.2008.00700.x}
}

@article{Taylor2013lgcp,
  author  = {Taylor, Benjamin M. and Davies, Tilman M. and Rowlingson, Barry S. and Diggle, Peter J.},
  title   = {{lgcp}: An {R} Package for Inference with Spatial and Spatio-Temporal Log-{Gaussian} {Cox} Processes},
  journal = {Journal of Statistical Software},
  year    = {2013},
  volume  = {52},
  number  = {4},
  pages   = {1--40},
  doi     = {10.18637/jss.v052.i04}
}

@article{Teng2017Bayesian,
  author  = {Teng, Ming and Nathoo, Farouk S. and Johnson, Timothy D.},
  title   = {{Bayesian} Computation for Log-{Gaussian} {Cox} Processes: A Comparative Analysis of Methods},
  journal = {Journal of Statistical Computation and Simulation},
  year    = {2017},
  volume  = {87},
  number  = {11},
  pages   = {2227--2252},
  doi     = {10.1080/00949655.2017.1326117}
}

@inproceedings{Gunter2014Efficient,
  author    = {Gunter, Tom and Lloyd, Chris and Osborne, Michael A. and Roberts, Stephen J.},
  title     = {Efficient {Bayesian} Nonparametric Modelling of Structured Point Processes},
  booktitle = {Proceedings of the 30th Conference on Uncertainty in Artificial Intelligence},
  pages     = {310--319},
  year      = {2014},
  publisher = {AUAI Press}
}

@book{bishop2006pattern,
  author    = {Bishop, Christopher M.},
  title     = {Pattern Recognition and Machine Learning},
  publisher = {Springer},
  year      = {2006}
}

@book{rasmussen2006gaussian,
  author    = {Rasmussen, Carl Edward and Williams, Christopher K. I.},
  title     = {{Gaussian} Processes for Machine Learning},
  publisher = {MIT Press},
  year      = {2006},
  doi       = {10.7551/mitpress/3206.001.0001}
}

@article{Hubert1985ARI,
  author  = {Hubert, Lawrence and Arabie, Phipps},
  title   = {Comparing Partitions},
  journal = {Journal of Classification},
  year    = {1985},
  volume  = {2},
  number  = {1},
  pages   = {193--218},
  doi     = {10.1007/BF01908075}
}

\clearpage
\appendix
\thispagestyle{empty}

\onecolumn
\aistatstitle{Asymptotically Valid Laplace Variational Inference for Bayesian Nonparametric Clustering of Marked Point Patterns: Supplementary Materials}

\section{FULL ELBO}
\label{app:elbo}

The variational objective is
\begin{align*}
\mathcal L(q)
=
\mathbb E_q\!\left[\log p(\bm y,\bm m,\bm z,\bm\phi,\bm\theta,\bm\tau^2)\right]
-
\mathbb E_q\!\left[\log q(\bm z,\bm\phi,\bm\theta,\bm\tau^2)\right].
\end{align*}
It is convenient to decompose
\begin{align*}
\mathcal L(q)=\mathcal L_{\mathrm{lik}}+\mathcal L_{\bm z}+\mathcal L_{\bm\phi}+\mathcal L_{\bm\theta,\bm\tau}.
\end{align*}


\paragraph{Likelihood Term.}
For subject $i$ and cluster $k$, the log-likelihood contribution under the squared link is
\begin{align*}
\log p(\bm{Y}_i \mid m_i,\bm\theta_k)
&=
\sum_{m=0}^1
\left\{
-T_i\,\bm\theta_{km}^\top \bm{M} \bm\theta_{km}
+
\sum_{j\in \mathcal{I}_{im}}
\log\!\left[\bigl(\bm{B}(\bm{y}_{ij})^\top \bm\theta_{km}\bigr)^2\right]
\right\}
+
N_i\log T_i.
\end{align*}

Taking expectation with respect to $q(\bm{\theta}_k)=q(\bm{\theta}_{k0})q(\bm\theta_{k1})$ gives
\begin{align*}
\mathbb E_{q(\bm\theta_k)}\!\left[\log p(\bm{Y}_i \mid m_i,\bm\theta_k)\right]
&=
\sum_{m=0}^1
\left\{
-T_i R_{km}
+
\sum_{j\in \mathcal{I}_{im}}
\mathcal H(\mu_{ij,km},\sigma^2_{ij,km})
\right\}
+
N_i\log T_i,
\end{align*}
where
\begin{align*}
R_{km}
&:=
\mathbb E_{q(\bm\theta_{km})}\!\left[\bm\theta_{km}^\top \bm{M} \bm\theta_{km}\right]
=
\text{tr}\!\left\{\bm{M}\bigl(\bm\Sigma_{\bm\theta km}+\bm\mu_{\bm\theta km}\bm\mu_{\bm\theta km}^\top\bigr)\right\},
\\
\mu_{ij,km}
&:=
\bm{B}(\bm{y}_{ij})^\top \bm\mu_{\bm\theta km},
\
\sigma^2_{ij,km}
:=
\bm{B}(\bm{y}_{ij})^\top \bm\Sigma_{\bm\theta km} \bm{B}(\bm{y}_{ij}),
\end{align*}
and
\begin{align*}
\mathcal H(\mu,\sigma^2)
:=
\mathbb E_{X\sim \mathcal N(\mu,\sigma^2)}[\log X^2]
=
-\widetilde G\!\left(-\frac{\mu^2}{2\sigma^2}\right)
+
\log\!\left(\frac{\sigma^2}{2}\right)
-
C_{\mathrm E},
\end{align*}
where $C_{\mathrm E}$ is the Euler--Mascheroni constant, and $\tilde{G}$ is defined via the confluent hypergeometric function as in \citep{lloyd2015variational}.

Hence the likelihood contribution to the ELBO is
\begin{align*}
\mathcal L_{\mathrm{lik}}
=
\sum_{i=1}^n \sum_{k=1}^K \nu_{ik}
\Biggl[
\sum_{m=0}^1
\left\{
-T_i R_{km}
+
\sum_{j\in \mathcal{I}_{im}}
\mathcal H(\mu_{ij,km},\sigma^2_{ij,km})
\right\}
+
N_i\log T_i
\Biggr].
\end{align*}

\paragraph{Assignment Term.}
Define
\begin{align}
\zeta_k:=\mathbb E_q[\log \pi_k]
=
\begin{cases}
\mathbb E_q[\log \phi_k]+\sum_{\ell=1}^{k-1}\mathbb E_q[\log(1-\phi_\ell)], & k=1,\ldots,K-1,\\[0.2cm]
\sum_{\ell=1}^{K-1}\mathbb E_q[\log(1-\phi_\ell)], & k=K,
\end{cases}
\label{eq:app-zeta}
\end{align}
where, under $q(\phi_k)=\mathrm{Beta}(\gamma_{k1},\gamma_{k2})$ for $k=1,\ldots,K-1$,
\begin{align*}
\mathbb E_q[\log\phi_k]=\psi(\gamma_{k1})-\psi(\gamma_{k1}+\gamma_{k2}),
\qquad
\mathbb E_q[\log(1-\phi_k)]=\psi(\gamma_{k2})-\psi(\gamma_{k1}+\gamma_{k2}),
\end{align*}
and $\psi(x)=\frac{d}{dx}\log\Gamma(x)$ denotes the digamma function.
Then
\begin{align*}
\mathcal L_z
=
\sum_{i=1}^n \sum_{k=1}^K \nu_{ik}\zeta_k
-
\sum_{i=1}^n \sum_{k=1}^K \nu_{ik}\log \nu_{ik}.
\end{align*}

\paragraph{Stick-breaking Term.}
For $k=1,\ldots,K-1$,
\begin{align*}
\mathcal L_\phi
&=
\sum_{k=1}^{K-1}
\Bigl\{
\log \alpha + (\alpha-1)\mathbb E_q[\log(1-\phi_k)]
-
\mathbb E_q[\log q(\phi_k)]
\Bigr\},
\\
\mathbb E_q[\log q(\phi_k)]
&=
\log\Gamma(\gamma_{k1}+\gamma_{k2})
-
\log\Gamma(\gamma_{k1})
-
\log\Gamma(\gamma_{k2})
\nonumber\\
&\qquad
+
(\gamma_{k1}-1)\mathbb E_q[\log\phi_k]
+
(\gamma_{k2}-1)\mathbb E_q[\log(1-\phi_k)].
\end{align*}

\paragraph{Coefficient and Variance Terms.}
Define
\begin{align*}
S_{km}:=\text{tr}(\bm\Omega\bm\Sigma_{\bm\theta km})+\bm\mu_{\bm\theta km}^\top\bm\Omega\bm\mu_{\bm\theta km}.
\end{align*}
Under $q(\tau_{km}^2)=\mathrm{IG}(\alpha_{qkm},\beta_{qkm})$, with density proportional to
$(\tau_{km}^2)^{-\alpha_{qkm}-1}\exp(-\beta_{qkm}/\tau_{km}^2)$, we have
\begin{align}
\mathbb E_q[\tau_{km}^{-2}]
&=
\frac{\alpha_{qkm}}{\beta_{qkm}},
\label{eq:app-Einvtausq}
\\
\mathbb E_q[\log \tau_{km}^2]
&=
\log\beta_{qkm}-\psi(\alpha_{qkm}).
\label{eq:app-Elogtausq}
\end{align}

Therefore,
\begin{align}
\mathcal L_{\theta,\tau}
&=
\sum_{k=1}^K\sum_{m=0}^1
\Biggl[
-\frac{d}{2}\log(2\pi)
+\frac12 \log |\bm\Omega|
+a_0 \log b_0
-\log\Gamma(a_0)
\nonumber\\
&\qquad\qquad
-
\left(\frac{d}{2}+a_0+1\right)\mathbb E_q[\log \tau_{km}^2]
-
\left(b_0+\frac12 S_{km}\right)\mathbb E_q[\tau_{km}^{-2}]
\nonumber\\
&\qquad\qquad
-
\mathbb E_q[\log q(\tau_{km}^2)]
-
\mathbb E_q[\log q(\bm\theta_{km})]
\Biggr],
\label{eq:app-Lthetatau-raw}
\end{align}
where
\begin{align*}
\mathbb E_q[\log q(\tau_{km}^2)]
&=
\alpha_{qkm}\log \beta_{qkm}
-
\log\Gamma(\alpha_{qkm})
-
(\alpha_{qkm}+1)\mathbb E_q[\log \tau_{km}^2]
-
\alpha_{qkm},
\\
\mathbb E_q[\log q(\bm\theta_{km})]
&=
-\frac{d}{2}\log(2\pi e)-\frac12 \log |\bm\Sigma_{\bm\theta km}|.
\end{align*}
Substituting \eqref{eq:app-Einvtausq} and \eqref{eq:app-Elogtausq} into
\eqref{eq:app-Lthetatau-raw} yields the full ELBO in closed form.

\section{CLOSED-FORM CAVI UPDATES FOR CONJUGATE BLOCKS}
\label{app:cavi}

The blocks $q(\phi_k)$, $q(z_i)$, and $q(\tau_{km}^2)$ admit exact coordinate updates.

\paragraph{Stick-breaking Weights.}
For $k=1,\ldots,K-1$,
\begin{align*}
q^*(\phi_k)=\mathrm{Beta}(\gamma_{k1},\gamma_{k2}),
\qquad
\gamma_{k1}=1+\sum_{i=1}^n \nu_{ik},
\qquad
\gamma_{k2}=\alpha+\sum_{i=1}^n \sum_{\ell=k+1}^K \nu_{i\ell}.
\end{align*}
The last stick is fixed, $\phi_K\equiv 1$, and therefore has no variational update.

\paragraph{Variance Parameters.}
For each $(k,m)$,
\begin{align*}
q^*(\tau_{km}^2)=\mathrm{IG}(\alpha_{qkm},\beta_{qkm}),
\qquad
\alpha_{qkm}=a_0+\frac{d}{2},
\qquad
\beta_{qkm}=b_0+\frac12 S_{km}.
\end{align*}

\paragraph{Cluster Allocations.}
For each subject $i$ and cluster $k$, define the responsibility score
\begin{align*}
\log \tilde\nu_{ik}
=
\zeta_k + \operatorname{Like}_{ik},
\end{align*}
where $\zeta_k$ is given by \eqref{eq:app-zeta} and
\begin{equation}
\operatorname{Like}_{ik}
=
\sum_{m=0}^1
\left[
-T_i R_{km}
+
\sum_{j\in \mathcal{I}_{im}}
\mathcal H(\mu_{ij,km},\sigma^2_{ij,km})
\right].
\label{eq:app-Likeik}
\end{equation}
The subject-specific constant $N_i\log T_i$ is omitted from \eqref{eq:app-Likeik} because it does
not depend on $k$ and therefore cancels when the responsibilities are normalized. The update is
\begin{align*}
\nu_{ik}
=
\frac{\tilde\nu_{ik}}{\sum_{\ell=1}^K \tilde\nu_{i\ell}},
\qquad
k=1,\ldots,K.
\end{align*}

\section{OBJECTIVE FUNCTION}
\label{app:theta-objective}

For a fixed cluster--mark block $(k,m)$, the optimal mean-field factor satisfies
\begin{align*}
\log q^*(\bm\theta_{km})
=
\mathbb E_{q_{-\bm\theta_{km}}}
\!\left[
\log p(\bm y,\bm m,\bm z,\bm\phi,\bm\theta,\bm\tau^2)
\right]
+
\mathrm{const}.
\end{align*}
Retaining only the terms that depend on $\bm\theta_{km}$ yields
\begin{align*}
\mathbb E_{q_{-\bm\theta_{km}}}
\!\left[
\log p(\bm y,\bm m,\bm z,\bm\phi,\bm\theta,\bm\tau^2)
\right]
&=
\sum_{i=1}^n \nu_{ik}
\left[
-T_i\,\bm\theta_{km}^\top \bm{M}\bm\theta_{km}
+
\sum_{j\in \mathcal{I}_{im}}
\log\!\left\{\bigl(\bm{B}(\bm{y}_{ij})^\top\bm\theta_{km}\bigr)^2\right\}
\right]
\nonumber\\
&\qquad
-
\frac{1}{2}\,
\mathbb E_q[\tau_{km}^{-2}]\,
\bm\theta_{km}^\top \bm\Omega \bm\theta_{km}
+
C,
\end{align*}
where $C$ is constant in $\bm\theta_{km}$. Define
\begin{align*}
\eta_{km}
&:=
\mathbb E_q[\tau_{km}^{-2}]
=
\frac{\alpha_{qkm}}{\beta_{qkm}},
\qquad
a_{km}
:=
\sum_{i=1}^n \nu_{ik}T_i,
\\
\bm{A}_{km}
&:=
a_{km} \bm{M} + \frac{1}{2}\eta_{km}\bm\Omega.
\end{align*}
Then the $\bm\theta_{km}$-dependent objective function is
\begin{align*}
J(\bm\theta_{km})
&=
-\bm\theta_{km}^\top \bm{A}_{km}\bm\theta_{km}
+
\sum_{i=1}^n \nu_{ik}
\sum_{j\in \mathcal{I}_{im}}
\log\!\left\{\bigl(\bm{B}(\bm{y}_{ij})^\top \bm\theta_{km}\bigr)^2\right\}.
\end{align*}

\section{VARIATIONAL INFERENCE ALGORITHM}
\label{app:vi-algorithm}

\begin{algorithm}[H]
\small
\caption{Variational Inference Algorithm (DPM-MPPP)}
\label{alg:cavi}
\begin{algorithmic}[1]
\State \textbf{Input:} Data $(\bm{Y}_i,\bm{m}_i)$ and offset variables $T_i$ for subjects $i = 1, \dots, n$. Basis functions $\bm{B}(\cdot)$. Precision matrix $\bm\Omega$. Hyperparameters $\alpha, a_0, b_0$. Truncation level $K$.
\State \textbf{Initialize:} Variational parameters $\gamma_{k1}, \gamma_{k2}, \nu_{ik}, \alpha_{q_{km}}, \beta_{q_{km}}, \bm\mu_{\bm\theta_{km}}, \bm\Sigma_{\bm\theta_{km}}$. Set $\alpha_{q_{km}} \gets a_0 + d/2$, and $\beta_{q_{km}}\gets b_0+\frac12\{\operatorname{Tr}(\bm\Omega\bm\Sigma_{\bm\theta_{km}})+\bm\mu_{\bm\theta_{km}}^\top\bm\Omega\bm\mu_{\bm\theta_{km}}\}$ for every $(k,m)$.
\State \textbf{Define:} Index sets for marks $\mathcal{I}_{im} = \{j \in \{1, \dots, N_i\} : m_{ij} = m\}$.
\Repeat
    \State \textit{// 1. Update cluster coefficients (Laplace Approximation)}
    \For{$k = 1, \dots, K$ and $m \in \{0, 1\}$}
        \State Compute $\eta_{km} \gets \alpha_{q_{km}} / \beta_{q_{km}}$, $\bm{A}_{km} \gets \left(\sum_{i=1}^n \nu_{i k} T_i\right) \bm{M} + \frac{1}{2} \eta_{km} \bm\Omega$
        \State Find positive mode $\bm\mu_{\bm\theta_{km}}^* \gets \arg\max_{\bm\theta \in D_{\delta}} \left[ -\bm\theta^\top \bm{A}_{km} \bm\theta + \sum_{i=1}^n \nu_{ik} \sum_{j \in \mathcal{I}_{im}} 2 \log(\bm\theta^\top \bm{B}(\bm{y}_{ij})) \right]$
        \State Compute Hessian $\bm H(\bm\mu_{\bm\theta_{km}}^*) \gets -2 \bm A_{km} - \sum_{i=1}^n \nu_{ik} \sum_{j \in \mathcal{I}_{im}} \frac{2}{(\bm\mu_{\bm\theta_{km}}^{*\top} \bm{B}(\bm{y}_{ij}))^2} \bm{B}(\bm{y}_{ij}) \bm{B}(\bm{y}_{ij})^\top$
        \State Update $\bm\mu_{\bm\theta_{km}} \gets \bm\mu_{\bm\theta_{km}}^*$ and $\bm\Sigma_{\bm\theta_{km}} \gets -\bm H(\bm\mu_{\bm\theta_{km}}^*)^{-1}$
    \EndFor
    
    \State \textit{// 2. Update hyperparameters for coefficients}
    \For{$k = 1, \dots, K$ and $m \in \{0, 1\}$}
        \State Update $\beta_{q_{km}} \gets b_0 + \frac{1}{2}\left( \text{Tr}(\bm\Omega \bm\Sigma_{\bm\theta_{km}}) + \bm\mu_{\bm\theta_{km}}^\top \bm\Omega \bm\mu_{\bm\theta_{km}} \right)$
    \EndFor
    
    \State \textit{// 3. Update stick-breaking weights}
    \For{$k = 1, \dots, K-1$}
        \State Update $\gamma_{k1} \gets 1 + \sum_{i=1}^n \nu_{ik}$
        \State Update $\gamma_{k2} \gets \alpha + \sum_{i=1}^n \sum_{l=k+1}^{K} \nu_{il}$
    \EndFor
    \State Set $\phi_K \equiv 1$

    \State \textit{// 4. Update cluster assignments}
    \For{$i = 1, \dots, n$ and $k = 1, \dots, K$}
        \State Compute 
        $\mathrm{Prior}_{ik} \gets
        \begin{cases}
            \psi(\gamma_{k1})-\psi(\gamma_{k1}+\gamma_{k2})
            +\sum_{\ell=1}^{k-1}\bigl\{\psi(\gamma_{\ell 2})-\psi(\gamma_{\ell 1}+\gamma_{\ell 2})\bigr\}, & k<K,\\[0.2cm]
            \sum_{\ell=1}^{K-1}\bigl\{\psi(\gamma_{\ell 2})-\psi(\gamma_{\ell 1}+\gamma_{\ell 2})\bigr\}, & k=K
        \end{cases}$
        \State and compute $\mathrm{Like}_{ik}:= \sum_{m=0}^1 \mathbb{E}_{q(\bm\theta_{km})}\!\left[\log p\!\left(\bm{Y}_{im} \mid \bm\theta_{km}\right)\right]$
        \State Update unnormalized assignment $\tilde\nu_{ik} \gets \exp(\mathrm{Prior}_{ik}+\mathrm{Like}_{ik})$
    \EndFor
    \State Normalize $\nu_{ik} \gets \tilde\nu_{ik} / \sum_{t=1}^K \tilde\nu_{it}$ for all $i, k$
    
    \State \textit{// 5. Convergence Check}
    \State Compute total ELBO $\mathcal{L}(q)$ and evaluate max parameter difference.
\Until{relative absolute change in ELBO $< \epsilon$}
\State \textbf{Output:} Optimal variational parameters $q^*(\bm\phi), q^*(\bm z), q^*(\bm\theta), q^*(\bm \tau^2)$.
\end{algorithmic}
\end{algorithm}

The positive-mode step in Algorithm~\ref{alg:cavi} can be solved as a convex optimization problem with linear positivity constraints. For computational efficiency, our implementation uses a coefficient-wise positive box constraint and L-BFGS-B, which was faster and more stable for the nonnegative B-spline bases used here, optionally followed by a few gradient-polishing steps. This box constraint is only a computational device and need not coincide with $D_\delta$. The L-BFGS-B solution therefore serves as a candidate mode rather than the final box-constrained mode. Starting from this candidate, we take additional gradient-based optimization steps on the original positive-chamber objective to obtain the mode over $D_\delta$. The specific optimizer is not essential; any reliable constrained optimizer for the same positive chamber objective can be used.

\makeatletter
\@ifundefined{proof}{%
  \newenvironment{proof}[1][Proof]{\par\medskip\noindent\textit{#1. }}{\hfill$\square$\par\medskip}
}{}
\makeatother

\theoremstyle{plain}
\newtheorem{appassumption}{Assumption}[section]
\newtheorem{applemma}{Lemma}[section]
\newtheorem{appproposition}{Proposition}[section]
\newtheorem{apptheorem}{Theorem}[section]
\newtheorem{appcorollary}{Corollary}[section]

\newtheorem{appremark}{Remark}[section]
\newtheorem{appexample}{Example}[section]

\section{TECHNICAL APPENDIX FOR SECTION~4.3}
\label{app:local-theory}

This appendix collects the theory that underlies Section~4.3. Throughout, we fix one cluster--mark block $(k,m)$ and suppress the block indices whenever no confusion arises. Thus
\[
\nu_i := \nu_{ik},
\qquad
\mathcal I_i := \mathcal I_{im},
\qquad
\eta := \mathbb{E}_q(\tau_{km}^{-2}),
\qquad
a := \sum_{i=1}^n \nu_i T_i,
\qquad
\lambda_P := \frac{\eta}{2a}.
\]
Let $N_i:=\sum_{j\in\mathcal I_i}\delta_{\bm y_{ij}}$ have Poisson intensity $T_i\lambda_i(\bm y)\,d\bm y$, where the $\lambda_i$ are fixed, nonnegative, bounded, integrable baseline intensities; in the mixture model this statement is conditional on the true subject intensities. The deterministic exposures satisfy $T_i\to\infty$ for every $i$, with fixed $n,d$. Dependence on these exposures is suppressed throughout. Let $\mathcal W$ be a fixed nonempty compact subset of the probability simplex
$\Delta_{n-1}:=\{w\in[0,1]^n:\sum_iw_i=1\}$, and let $\underline a\to\infty$ be a deterministic function of the exposures. Define
\begin{equation}
\begin{split}
\mathcal A:=\Bigl\{(\bm\nu,\eta):\;&\bm\nu\in[0,1]^n,\, a\ge\underline a,\,
w=(\nu_iT_i/a)_{i=1}^n\in\mathcal W, \, 0\le\eta\le C_\eta\Bigr\},
\qquad C_\eta:=\frac{a_0+d/2}{b_0}.
\end{split}
\label{eq:app-admissible-pairs}
\end{equation}
Assume this set is nonempty for all sufficiently large exposures. All empirical suprema over $\mathcal A$ use the same observed point patterns; zero-weight terms in the log sums are omitted. In population formulas throughout this appendix, $\lambda_0=\lambda_{0,w}:=\sum_iw_i\lambda_i$, $Q=Q_w$, $g_0=\sqrt{\lambda_{0,w}}$, and $\bm\theta^\star=\bm\theta^\star(w)$; their weight indices are omitted to preserve the notation. The common population geometry required by the main theorems is stated in Assumption~\ref{ass:app-uniform-family}. Pointwise sufficient conditions below are evaluated at the corresponding weighted target. For the bounded-precision VI updates, the common penalty benchmark is $\bar\lambda=0$; the existing population calculations remain stated for a general fixed $\bar\lambda\ge0$. No conditional Poisson law given the fitted responsibilities is assumed.

We write
\begin{align*}
s_{\bm\theta}(\bm y) := \bm B(\bm y)^\top\bm\theta,
\qquad
\lambda_{\bm\theta}(\bm y) := s_{\bm\theta}(\bm y)^2,
\end{align*}
so the squared link is sign-symmetric:
\begin{align*}
\lambda_{\bm\theta}(\bm y)=\lambda_{-\bm\theta}(\bm y).
\end{align*}
The normalized empirical criterion is
\begin{align*}
\widetilde J(\bm \theta)
\coloneqq
\begin{cases}
-\bm\theta^\top \bm{M}\bm\theta
+\dfrac{2}{a}\displaystyle\sum_{i=1}^n \nu_i \sum_{j\in\mathcal I_i}\log|s_{\bm\theta}(\bm{y}_{ij})|
-\lambda_P\,\bm\theta^\top\bm\Omega\,\bm\theta,
& 
\text{if }s_{\bm\theta}(\bm y_{ij})\neq0\ \text{whenever }\nu_i>0,\\
-\infty,
& \text{otherwise,}
\end{cases}
\end{align*}
defined for $\bm\theta\in\bar{\mathbb B}(\bm 0,R)$, where
\begin{align*}
\bm{M} := \int_{\mathcal B} \bm B(\bm y)\bm B(\bm y)^\top\,d\bm{y}.
\end{align*}
The corresponding 
benchmark criterion for a fixed $w$ is
\begin{align*}
Q(\bm \theta)
\coloneqq
\begin{cases}
-\bm\theta^\top \bm{M}\bm\theta
+2\displaystyle\int_{\mathcal B}\lambda_0(\bm y)\log|s_{\bm\theta}(\bm y)|\,d\bm{y}
-\bar\lambda\,\bm\theta^\top\bm\Omega\,\bm\theta,
& \text{if } \lambda_0|\log|s_{\bm\theta}||\in L^1(\mathcal B),\\
-\infty,
& \text{otherwise,}
\end{cases}
\end{align*}
where $\bar\lambda\in[0,\infty)$ is deterministic.
On $D$ where $s_{\bm\theta}>0$, $\log|s_{\bm\theta}|=\log s_{\bm\theta}$, so $\widetilde J$ and $Q$ coincide with their standard log-form expressions. 
By the sign-invariance of the squared link, $\widetilde J(\bm\theta)=\widetilde J(-\bm\theta)$ and $Q(\bm\theta)=Q(-\bm\theta)$ on $\bar{\mathbb B}(\bm 0,R)$.
\begin{equation}
Q(\bm\theta)=Q(-\bm\theta),\quad \widetilde J(\bm\theta)=\widetilde J(-\bm\theta).
\label{eq:app-sign-symmetry}
\end{equation}

\subsection{Positive Chambers and Compact Cores}

This subsection prepares the population optimization problem for the positive branch of the squared-link model. Although the squared link gives the same intensity after a global sign change, the logarithmic likelihood becomes unstable when the underlying linear predictor approaches zero. We therefore first establish that the population maximizer is separated from this zero boundary. Under a strictly positive target intensity and an interior approximation condition, candidates close to the zero boundary are uniformly worse than an interior candidate. This reduction allows the subsequent consistency argument to be carried out on a compact interior region where the objective is well behaved.

Fix constants $\delta>0$ and $R>0$. Define
\begin{align*}
D
&:=
\left\{
\bm\theta\in\mathbb R^d:\ \|\bm\theta\|_2\le R,\ \inf_{\bm y\in\mathcal B}s_{\bm\theta}(\bm y)>0
\right\},
\\
D_\delta
&:=
\left\{
\bm\theta\in\mathbb R^d:\ \|\bm\theta\|_2\le R,\ \inf_{\bm y\in\mathcal B}s_{\bm\theta}(\bm y)\ge \delta
\right\},
\\
D_\delta^{\mathrm{bd}}
&:=
\left\{
\bm\theta\in\mathbb R^d:\ \|\bm\theta\|_2\le R,\ 0<\inf_{\bm y\in\mathcal B}s_{\bm\theta}(\bm y)< \delta
\right\}.
\end{align*}
Then
\begin{align*}
D = D_\delta \cup D_\delta^{\mathrm{bd}}.
\end{align*}
The set $D$ is convex because it is the intersection of the Euclidean ball with the half-space constraints $\bm B(\bm y)^\top\bm\theta>0$ for all $\bm y\in\mathcal B$. The set $D_\delta$ is convex and closed.

Lemma~\ref{lem:app-bounded-basis} verifies the feasibility of the positivity constraint. Under continuity, nonnegativity, and the absence of a common zero for the basis functions, one can choose a positive coefficient vector such that the linear predictor is uniformly bounded away from zero on the compact domain. This ensures that the positive interior region used later is nonempty.

\begin{applemma}[Basis boundedness and nonemptiness of positive chambers]
\label{lem:app-bounded-basis}
Suppose that $\mathcal B$ is compact and each component of
$\bm B:\mathcal B\to\mathbb R^d$ is continuous. Then
\[
B_{\max}:=\sup_{\bm y\in\mathcal B}\|\bm B(\bm y)\|_2<\infty.
\]
Assume further that $\bm B$ is componentwise nonnegative and has no common zero on
$\mathcal B$, in the sense that
\[
\kappa_{\bm B}
:=
\inf_{\bm y\in\mathcal B}\sum_{\ell=1}^d  B_\ell(\bm y)
>0 .
\]
Then the positive chamber
$D$
is nonempty for every $R>0$. Moreover,
$\left\{
\bm\theta:\|\bm\theta\|_2\le R,\ 
\bm\theta_\ell>0\ \text{for all }\ell
\right\}
\subseteq D$ .
For the compact positive core
$D_\delta$,
we have the inclusion
\[
\left\{
\bm\theta:\|\bm\theta\|_2\le R,\ 
\theta_\ell\ge \frac{\delta}{\kappa_{\bm B}}
\ \text{for all }\ell
\right\}
\subseteq D_\delta .
\]
Consequently, if
$R\ge \frac{\sqrt d\,\delta}{\kappa_{\bm B}}$,
then $D_\delta$ is nonempty.
In particular, if $\bm B$ is a nonnegative partition-of-unity basis, namely
\[
\sum_{\ell=1}^d B_\ell(\bm y)=1
\qquad\text{for all } \bm y\in\mathcal B,
\]
then $\kappa_{\bm B}=1$, one may take $B_{\max}=1$, and
\[
\left\{
\bm \theta:\|\bm\theta\|_2\le R,\ 
\theta_\ell\ge\delta
\ \text{for all }\ell
\right\}
\subseteq D_\delta .
\]
\end{applemma}

\begin{proof}
Since each component of $\bm B$ is continuous and $\mathcal B$ is compact, the map
$\bm y\mapsto \|\bm B(\bm y)\|_2$ is continuous on $\mathcal B$ and therefore attains its
maximum. Hence $B_{\max}<\infty$.

Now suppose $\bm B$ is componentwise nonnegative and $\kappa_{\bm B}>0$. For any
$\bm \theta$ with $\theta_\ell>0$ for all $\ell$, let
\[
\theta_{\min}:=\min_{\ell=1,\ldots,d}\theta_\ell>0.
\]
Then, for every $\bm y\in\mathcal B$,
\[
\bm B(\bm y)^\top\bm\theta
=
\sum_{\ell=1}^d B_\ell(\bm y)\theta_\ell
\ge
\theta_{\min}\sum_{\ell=1}^d B_\ell(\bm y)
\ge
\theta_{\min}\kappa_{\bm B}
>0.
\]
Thus every strictly positive coefficient vector in the radius-$R$ ball belongs
to $D$. Such vectors exist for every $R>0$, for example a sufficiently small
positive multiple of $\mathbf 1_d$.

Similarly, if $\theta_\ell\ge\delta/\kappa_{\bm B}$ for all $\ell$, then
\[
\bm B(\bm y)^\top\bm\theta
\ge
\frac{\delta}{\kappa_{\bm B}}
\sum_{\ell=1}^d B_\ell(\bm y)
\ge
\delta
\]
for all $\bm y\in\mathcal B$. This proves the stated inclusion in $D_\delta$.
The nonemptiness condition follows by taking
$\bm\theta=(\delta/\kappa_{\bm B})\mathbf 1_d$, whose Euclidean norm is
$\sqrt d\,\delta/\kappa_{\bm B}$.

Finally, under the nonnegative partition-of-unity condition,
\[
\kappa_{\bm B}
=
\inf_{\bm y\in\mathcal B}\sum_{\ell=1}^dB_\ell(\bm y)
=
1.
\]
Moreover,
\[
\|\bm B(\bm y)\|_2
\le
\|\bm B(\bm y)\|_1
=
\sum_{\ell=1}^dB_\ell(\bm y)
=
1,
\]
so one may take $B_{\max}=1$.
\end{proof}

\begin{appexample}[B-spline bases]
\label{exm:bs1}
For standard B-spline bases with an open knot sequence on a compact domain, the basis functions are continuous, nonnegative, and form a partition of unity. Hence $B_{\max}=1$ and $\kappa_{\bm B}=1$. The same conclusion holds for tensor-product B-spline bases on rectangular domains.
\end{appexample}

Next, Lemma~\ref{lem:app-strong-concavity} provides the uniqueness ingredient for the population optimization problem. It shows that \(Q\) is strongly concave on \(D\), with curvature controlled by the smallest eigenvalue of \(\bm M\). Hence \(Q\) can have at most one maximizer on \(D\), and consequently at most one maximizer on the compact core \(D_\delta\).

\begin{applemma}[Strong concavity on the positive domain]
\label{lem:app-strong-concavity}
Suppose $\bm M\succ0$, $\bm \Omega\succeq0$, and $\bar\lambda\ge0$. Then $Q$ is strongly concave on $D$. More precisely, if $\mu:=\lambda_{\min}(\bm M)>0$, then for every $\bm\theta,\bm\theta'\in D$,
\begin{equation}
Q(\bm\theta')
\le
Q(\bm\theta)+\nabla Q(\bm\theta)^\top(\bm\theta'-\bm\theta)-\mu\|\bm\theta'-\bm\theta\|_2^2.
\label{eq:app-strong-concavity}
\end{equation}
In particular, $Q$ has at most one maximizer on $D$.
\end{applemma}

\begin{proof}
For $\bm\theta\in D$, differentiation under the integral gives
\begin{align*}
\nabla^2 Q(\bm\theta)
=
-2\bm M - 2\bar\lambda\,\bm\Omega
-2\int_{\mathcal B}\lambda_0(\bm y)\,
\frac{\bm B(\bm y)\bm B(\bm y)^\top}{(\bm B(\bm y)^\top\bm \theta)^2}\,d\bm{y}.
\end{align*}
Hence for every $\bm v\in\mathbb R^d$,
\[
\bm v^\top \nabla^2Q(\bm \theta)\bm v
\le -2\bm v^\top \bm M\bm v
\le -2\mu\|\bm v\|_2^2.
\]
Therefore $Q$ is strongly concave with modulus $\mu$, which yields \eqref{eq:app-strong-concavity}.
\end{proof}

Lastly, Lemma~\ref{lem:app-compact-core} provides the compactness needed for the population mode argument. Since \(D\) is defined by a strict positivity condition, it is not closed and a continuous objective need not attain its supremum on \(D\). The lemma shows that the interior set \(D_\delta\) is compact under continuity of the basis functions. This ensures that optimization over \(D_\delta\) has an actual maximizer rather than only a supremum.

\begin{applemma}[Compactness of the positive core]
\label{lem:app-compact-core}
Suppose each component of $\bm B$ is continuous on compact $\mathcal B$. Then $D_\delta$ is compact.
\end{applemma}

\begin{proof}
The ball $\{\bm\theta:\|\bm\theta\|_2\le R\}$ is compact. Since $\bm y\mapsto \bm B(\bm y)^\top\bm\theta$ is continuous on compact $\mathcal B$, the map
\[
\bm\theta \mapsto \inf_{\bm y\in\mathcal B} \bm B(\bm y)^\top \bm \theta
\]
is continuous. Thus $D_\delta$ is a closed subset of a compact set.
\end{proof}

\subsection{Population Assumptions and Positive-domain Interiority}

This subsection establishes the population target for the positive chamber optimization problem. The main goal is to show that the population criterion $Q$ has a well-defined and unique maximizer. The argument combines three ingredients: the boundary region near zero is excluded, the remaining interior feasible set is compact, and the criterion is strongly concave on the positive chamber $D$. As a result, the population mode \(\theta^\star\) exists, is unique, and lies in \(D_\delta\). This provides the deterministic target for the empirical mode consistency result in the next subsection.

We first introduce the required assumptions.
Let $g_0(\bm y):=\sqrt{\lambda_0(\bm y)}$.

\begin{appassumption}[Strict positivity of the target square root]
\label{ass:app-target}
There exist constants $0<c_0\le C_0<\infty$ such that
\begin{align*}
c_0 \le g_0(\bm y)\le C_0
\qquad\text{for all } \bm y\in\mathcal B.
\end{align*}
\end{appassumption}

\begin{appassumption}[Existence of an interior witness]
\label{ass:app-witness}
There exist $\varepsilon>0$ and $\bm\theta^+\in D_{2\delta}$ such that
\begin{align*}
\|g_0-s_{\bm\theta^+}\|_\infty\le \varepsilon,
\qquad
\varepsilon < c_0/2.
\end{align*}
\end{appassumption}

\begin{appassumption}[Boundary-strip separation]
\label{ass:app-boundary}
Define
\begin{align*}
d(\delta)
:=
\inf_{\bm\theta\in D_\delta^{\mathrm{bd}}}\|s_{\bm\theta}-g_0\|_2.
\end{align*}
Assume that
\begin{align*}
d(\delta)^2 > M_\varepsilon |\mathcal B|\,\varepsilon^2,
\end{align*}
where
\begin{align*}
M_\varepsilon := 1+\frac{C_0^2}{(c_0-\varepsilon)^2}.
\end{align*}
\end{appassumption}

Assumptions~\ref{ass:app-target}--\ref{ass:app-boundary} give primitive conditions under which the zero-boundary region can be ruled out at the population level. Assumption~\ref{ass:app-target} requires the target square root \(g_0\) to be uniformly positive and bounded, so the true intensity is separated from zero. Assumption~\ref{ass:app-witness} requires that this target can be approximated by an interior positive candidate \(\bm\theta^+\in D_{2\delta}\). Assumption~\ref{ass:app-boundary} then imposes a separation condition: candidates in the boundary strip \(D_\delta^{\mathrm{bd}}\) must remain sufficiently far from \(g_0\) in \(L^2\). Together, these assumptions ensure that the population criterion favors an interior positive candidate over points close to the zero boundary.

\begin{appexample}[B-spline verification of Assumption~\ref{ass:app-witness}]
\label{exm:app-bs-witness}
Suppose that \(\bm B\) is a tensor-product B-spline basis of degree \(p\) on a compact rectangular domain \(\mathcal B\subset\mathbb R^H\), and that Assumption~\ref{ass:app-target} holds. Assume further that \(g_0\) is \(\alpha\)-smooth with \(\alpha\le p+1\). If \(h\) is the maximal knot spacing, then standard spline approximation gives a coefficient vector \(\bm\theta^+\) such that
\[
\|g_0-\bm B^\top\bm\theta^+\|_\infty \le C h^\alpha .
\]
Taking \(h\) sufficiently small yields \(\varepsilon:=C h^\alpha<c_0/2\). Therefore
\[
s_{\bm\theta^+}(\bm y)=\bm B(\bm y)^\top\bm\theta^+
\ge g_0(y)-\varepsilon
\ge c_0-\varepsilon
\]
for all \(\bm y\in \mathcal B\). If \(\delta\) is chosen so that \(2\delta<c_0-\varepsilon\), then
\[
\inf_{\bm y\in \mathcal B}s_{\bm\theta^+}(\bm y)\ge 2\delta .
\]
Thus \(\bm\theta^+\in D_{2\delta}\), provided that \(R\) is chosen large enough to contain the approximating coefficient vector. Hence Assumption~\ref{ass:app-witness} is verified for sufficiently fine knots.
\end{appexample}

\begin{appexample}[B-spline verification of Assumption~\ref{ass:app-boundary}]
\label{exm:bs3}
Suppose that \(\bm B\) is a tensor-product B-spline basis of degree \(p\) on a compact rectangular domain \(\mathcal B\subset\mathbb R^H\), with maximal knot spacing \(h\). Let \(\bm\theta^+\in D_{2\delta}\) and \(\varepsilon\ge\|g_0-s_{\bm\theta^+}\|_\infty\) be as in the previous example, so that \(\varepsilon = C_1h^\alpha\). Assume also that \(g_0\) is Lipschitz; this is implied, for example, when \(\alpha\ge 1\).
We now lower bound the distance from \(g_0\) to the boundary strip. For the spline class with \(\|\bm\theta\|_2\le R\), there exists a constant \(C_2>0\) such that the Lipschitz constant of \(s_{\bm\theta}\) satisfies
\[
L\le C_2h^{-1}.
\]
For any \(\bm\theta\in D_\delta^{\mathrm{bd}}\), there exists \(\bm y_0\in \mathcal B\) such that
\[
s_{\bm\theta}(\bm y_0)<\delta .
\]
By Assumption~\ref{ass:app-target}, \(g_0(\bm y_0)\ge c_0\). If \(\delta<c_0/2\), then
\[
g_0(\bm y_0)-s_{\bm \theta}(\bm y_0)>c_0-\delta>c_0/2.
\]
Let \(L_0\) be a Lipschitz constant of \(g_0\). Then \(g_0-s_{\bm\theta}\) is Lipschitz with constant at most \(L_0+L\). Hence the above gap persists on a neighborhood of \(y_0\). In particular, for
\[
r:=\frac{c_0}{4(L_0+L)},
\]
we have
\[
g_0(\bm y)-s_{\bm\theta}(\bm y)>c_0/4
\]
for all \(\bm y\in \mathcal B\cap \mathbb B(\bm y_0,r)\). Since \(L\le C_2h^{-1}\), there exists a constant \(C_3>0\) such that
\[
r\ge C_3h
\]
for sufficiently small \(h\). Since the domain is rectangular, there exists \(C_4>0\) such that
\[
|\mathcal B\cap \mathbb B(\bm y_0,r)|\ge C_4 r^H .
\]
Therefore,
\[
\|s_{\bm\theta}-g_0\|_2^2
\ge
\frac{c_0^2}{16}|\mathcal B\cap \mathbb B(\bm y_0,r)|
\ge
C_5 h^H
\]
for some constant \(C_5>0\). Hence
\[
d(\delta)
=
\inf_{\bm \theta\in D_\delta^{\mathrm{bd}}}\|s_{\bm\theta}-g_0\|_2
\ge C_6 h^{H/2}.
\]

Combining this with \(\varepsilon\le C_1h^\alpha\), we obtain
\[
\frac{\varepsilon}{d(\delta)}
\le
C_7 h^{\alpha-H/2}.
\]
Thus, if \(\alpha>H/2\), then \(\varepsilon/d(\delta)\to0\) as \(h\to0\). More generally, the sufficient condition is
\[
M_\varepsilon h^{2\alpha-H}\to0 .
\]
In particular, if \(M_\varepsilon\) remains bounded as \(h\to0\) and \(\alpha>H/2\), then
\[
d(\delta)^2>M_\varepsilon |\mathcal{B}|\varepsilon^2
\]
for sufficiently small \(h\). Therefore Assumption~\ref{ass:app-boundary} holds.
\end{appexample}

Proposition~\ref{prop:app-boundary-strip} is the boundary-exclusion step for the population objective. Its purpose is to show that the population criterion cannot be maximized by functions that stay positive but approach the zero boundary. Under the interior approximation condition and the \(L^2\)-separation condition, every candidate in the boundary strip is uniformly worse than the interior witness \(\bm\theta^+\). Consequently, the population maximization over the positive chamber can be reduced to the interior core \(D_\delta\).

Define
\[
\Delta_P(\delta)
:=
\sup_{\bm\theta\in D_\delta^{\mathrm{bd}}}
\left|
\bm\theta^{+\top}\bm\Omega\bm\theta^+
-
\bm\theta^\top\bm\Omega\bm\theta
\right|.
\]
Since \(D_\delta^{\mathrm{bd}}\) is bounded and \(\bm\Omega\) is fixed, \(\Delta_P(\delta)<\infty\). This term controls the possible penalty difference between the interior witness and the boundary strip, and Proposition~\ref{prop:app-boundary-strip} assumes that its scaled contribution \(\bar\lambda\Delta_P(\delta)\) is sufficiently small. In the setting covered by Lemma~\ref{lem:app-operator-norm} and Proposition~\ref{prop:app-gamma} below, Assumption~\ref{ass:app-penalty} can be taken as \(\bar\lambda=0\). In that case, the scaled penalty contribution vanishes, and the required smallness condition follows directly from Assumption~\ref{ass:app-boundary}.

\begin{appproposition}[Population exclusion of the boundary strip]
\label{prop:app-boundary-strip}
Suppose Assumptions~\ref{ass:app-target}--\ref{ass:app-boundary} hold. If
\begin{equation}
\bar\lambda\,\Delta_P(\delta)
<
d(\delta)^2 - M_\varepsilon |\mathcal B|\,\varepsilon^2,
\label{eq:app-small-penalty-pop}
\end{equation}
then
\begin{equation}
\sup_{\bm\theta\in D_\delta^{\mathrm{bd}}} Q(\bm\theta)
<
Q(\bm\theta^+)
\le
\sup_{\bm\theta\in D_{2\delta}} Q(\bm\theta).
\label{eq:app-boundary-strip-gap}
\end{equation}
In particular, every maximizer of $Q$ over $D$ lies in $D_\delta$.
\end{appproposition}

\begin{proof}
Define
\begin{align*}
\ell_y(u):=u^2 - 2\lambda_0(\bm y)\log u,
\qquad u>0.
\end{align*}
Since $g_0(\bm y)^2=\lambda_0(\bm y)$,
\begin{align*}
\ell_y'(g_0(\bm y))=0,
\qquad
\ell_y''(u)=2+2\lambda_0(\bm y)u^{-2}\ge 2.
\end{align*}
Hence $u\mapsto \ell_{\bm y}(u)$ is strongly convex and minimized at $u=g_0(\bm y)$. Integrating pointwise strong convexity yields
\begin{equation}
\int_{\mathcal B}\{\ell_{\bm y}(s(\bm y))-\ell_{\bm y}(g_0(\bm y))\}\,d\bm{y}
\ge
\|s-g_0\|_2^2
\qquad\text{for every admissible } s>0.
\label{eq:app-ell-lower}
\end{equation}

Assumption~\ref{ass:app-witness} implies
\begin{align*}
s_{\bm\theta^+}(\bm y)\ge c_0-\varepsilon>c_0/2,
\qquad
s_{\bm\theta^+}(\bm y)\le C_0+\varepsilon
\qquad\text{for all } \bm y\in\mathcal B.
\end{align*}
Therefore,
\begin{align*}
\ell_{\bm y}''(u)
\le
2 + \frac{2C_0^2}{(c_0-\varepsilon)^2}
=
2M_\varepsilon
\end{align*}
on $[c_0-\varepsilon,C_0+\varepsilon]$. A second-order Taylor expansion of $\ell_{\bm y}$ around $g_0(\bm y)$ gives
\begin{align*}
\ell_{\bm y}(s_{\bm\theta^+}(\bm y))-\ell_{\bm y}(g_0(\bm y))
\le
M_\varepsilon |s_{\bm\theta^+}(\bm y)-g_0(\bm y)|^2.
\end{align*}
Integrating and using Assumption~\ref{ass:app-witness},
\begin{equation}
\int_{\mathcal B}\{\ell_{\bm y}(s_{\bm \theta^+}(\bm y))-\ell_{\bm y}(g_0(\bm y))\}\,d\bm{y}
\le
M_\varepsilon |\mathcal B|\,\varepsilon^2.
\label{eq:app-witness-loss}
\end{equation}

Now let $\bm\theta\in D_\delta^{\mathrm{bd}}$. By Assumption~\ref{ass:app-boundary} and \eqref{eq:app-ell-lower},
\begin{equation}
\int_{\mathcal B}\{\ell_{\bm y}(s_{\bm\theta}(\bm y))-\ell_{\bm y}(g_0(\bm y))\}\,d\bm{y}
\ge
d(\delta)^2.
\label{eq:app-strip-loss}
\end{equation}
Combining \eqref{eq:app-witness-loss} and \eqref{eq:app-strip-loss} gives
\begin{align*}
Q(\bm\theta^+)-Q(\bm \theta)
\ge
d(\delta)^2 - M_\varepsilon |\mathcal B|\varepsilon^2
-\bar\lambda\left(
\bm\theta^{+\top}\bm\Omega\,\bm\theta^+ - \bm\theta^\top\bm\Omega\,\bm\theta
\right).
\end{align*}
Hence
\begin{align*}
Q(\bm\theta^+)-Q(\bm\theta)
\ge
d(\delta)^2 - M_\varepsilon |\mathcal B|\varepsilon^2 - \bar\lambda\,\Delta_P(\delta).
\end{align*}
Under \eqref{eq:app-small-penalty-pop}, the right-hand side is strictly positive for every $\bm\theta\in D_\delta^{\mathrm{bd}}$, proving \eqref{eq:app-boundary-strip-gap}. Since $D=D_\delta\cup D_\delta^{\mathrm{bd}}$, no maximizer over $D$ can lie in the strip.
\end{proof}

Finally, Theorem~\ref{thm:app-positive-mode} completes the population-level argument on the positive chamber. It combines the boundary-exclusion result, compactness of the interior core, and strong concavity of \(Q\) to show that the population criterion has a well-defined and unique maximizer. Moreover, this maximizer lies in \(D_\delta\), away from the zero boundary. This result defines the deterministic target \(\bm\theta^\star\) for the empirical mode consistency result developed in the next subsection.

\begin{apptheorem}[Existence and uniqueness of the population maximizer on the positive domain]
\label{thm:app-positive-mode}
Suppose Assumptions~\ref{ass:app-target}--\ref{ass:app-boundary} hold, $\bm M\succ0$, $\bm\Omega\succeq0$, and \eqref{eq:app-small-penalty-pop} holds. Then the maximization of $Q$ over $D$ has a unique maximizer $\bm\theta^\star$. Moreover,
\begin{equation}
\bm\theta^\star\in D_\delta,
\qquad
\inf_{\bm y\in\mathcal B} \bm B(\bm y)^\top\bm\theta^\star \ge \delta.
\label{eq:app-pop-max-core}
\end{equation}
\end{apptheorem}

\begin{proof}
By Proposition~\ref{prop:app-boundary-strip}, every maximizer of $Q$ over $D$ must belong to $D_\delta$. By Lemma~\ref{lem:app-compact-core}, $D_\delta$ is compact. Since $Q$ is continuous on $D_\delta$, it attains its maximum there. By Lemma~\ref{lem:app-strong-concavity}, $Q$ is strongly concave on the convex set $D$ and therefore has at most one maximizer on $D$. This proves existence and uniqueness, and \eqref{eq:app-pop-max-core} follows from the definition of $D_\delta$.
\end{proof}

\subsection{Uniform Convergence on the Compact Positive Core}

This subsection proves the empirical consistency step. We show that \(\widetilde J\) converges uniformly to \(Q\) on \(D_\delta\). The lower bound \(s_{\bm\theta}(y)\ge \delta\) makes the logarithmic term stable, and compactness of \(D_\delta\) allows a finite-net argument to upgrade pointwise convergence to uniform convergence. This uniform convergence, combined with the uniqueness of the population maximizer, implies that the constrained empirical mode converges to \(\bm\theta^\star\).

\begin{appassumption}[Growing within-subject exposure]
\label{ass:app-growth}
The number of subjects $n$ 
and the basis dimension $d$ are fixed and
\begin{equation}
\min_{1\le i\le n} T_i \to \infty.
\label{eq:app-growth}
\end{equation}

For the coefficient blocks under consideration, $(\bm\nu,\eta)\in\mathcal A$, so $a\ge\underline a\to\infty$. For $\underline a=\min_iT_i$, this bound holds at every update of an active block, as proved in Corollary~\ref{cor:app-active-iterations}.
\end{appassumption}

\begin{appassumption}[Penalty discrepancy control on the compact core]
\label{ass:app-penalty}
There exists a deterministic benchmark level $\bar\lambda\in[0,\infty)$ such that
\begin{equation}
\Delta_{\mathrm{pen}}(\delta)
:=
\sup_{\bm\theta\in D_\delta}
\left|
(\lambda_P-\bar\lambda)\bm\theta^\top\bm\Omega\,\bm\theta
\right|
\xrightarrow{p}0.
\label{eq:app-penalty-discrepancy}
\end{equation}
For uniform statements, the convergence is required uniformly over $(\bm\nu,\eta)\in\mathcal A$, namely $\sup_{\mathcal A}\Delta_{\mathrm{pen}}(\delta)\xrightarrow{p}0$.
\end{appassumption}

Assumptions~\ref{ass:app-growth}--\ref{ass:app-penalty} are the stochastic ingredients for uniform convergence of the empirical criterion. Assumption~\ref{ass:app-growth} lets the within-subject point patterns become increasingly informative, so that the normalized Poisson sums converge to their population integrals. Assumption~\ref{ass:app-penalty} ensures that the random penalty level \(\lambda_P\) is close to its deterministic benchmark \(\bar\lambda\) uniformly over \(D_\delta\). 
Together, these assumptions allow the empirical criterion \(\widetilde J\) to approximate the population criterion \(Q\) uniformly on the compact core.

The next lemma and proposition provide a practical verification of Assumption~\ref{ass:app-penalty}. They show that, on the compact core \(D_\delta\), the penalty discrepancy is controlled by the size of \(\lambda_P-\bar\lambda\) and the operator norm of \(\bm\Omega\). Under the Gaussian coefficient prior with a Gamma prior on the precision, this control holds with \(\bar\lambda=0\). Hence the empirical penalty term becomes asymptotically negligible in the uniform convergence argument.

\begin{applemma}[A simple operator-norm sufficient condition]
\label{lem:app-operator-norm}
Assume $D_\delta\subseteq\{\bm\theta:\|\bm\theta\|_2\le R\}$. Then
\begin{equation}
\Delta_{\mathrm{pen}}(\delta)
\le
R^2\|\bm\Omega\|_{\mathrm{op}}\,|\lambda_P-\bar\lambda|.
\label{eq:app-operator-norm}
\end{equation}
Consequently Assumption~\ref{ass:app-penalty} is implied by
\begin{align*}
R^2\|\bm\Omega\|_{\mathrm{op}}\,|\lambda_P-\bar\lambda| \xrightarrow{p}0.
\end{align*}
\end{applemma}

\begin{proof}
For every $\bm\theta\in D_\delta$,
\[
\bm\theta^\top\bm\Omega\,\bm\theta \le \|\bm\Omega\|_{\mathrm{op}}\|\bm\theta\|_2^2 \le \|\bm\Omega\|_{\mathrm{op}}R^2.
\]
Taking the supremum over $\bm\theta\in D_\delta$ gives \eqref{eq:app-operator-norm}.
\end{proof}

\begin{appproposition}[Gaussian prior with a Gamma hyperprior on precision]
\label{prop:app-gamma}
Let $\kappa:=\tau^{-2}$ and suppose the conditional prior for $\bm\theta$ given $\kappa$ has kernel
\begin{align*}
p(\bm\theta\mid \kappa)\propto \kappa^{d/2}\exp\!\left\{-\frac{\kappa}{2}\bm\theta^\top\bm\Omega\,\bm\theta\right\},
\end{align*}
while
\begin{align*}
\kappa\sim\Gamma(a_0,b_0)
\qquad\text{with } a_0>0,\ b_0>0.
\end{align*}
Suppose the variational update for $\kappa$ is
\begin{align*}
q(\kappa)
=
\Gamma\!\left(
a_0+\frac{d}{2},
b_0+\frac{1}{2}\mathbb E_q[\bm\theta^\top\bm\Omega\,\bm\theta]
\right).
\end{align*}
Then
\begin{equation}
\eta
=
\mathbb E_q(\kappa)
=
\frac{a_0+d/2}{b_0+\frac12\mathbb E_q[\bm\theta^\top\bm\Omega\,\bm\theta]}
\le
\frac{a_0+d/2}{b_0},
\label{eq:app-eta-bound}
\end{equation}
and therefore
\begin{equation}
\lambda_P = \frac{\eta}{2a}
\le
\frac{a_0+d/2}{2b_0\,a}.
\label{eq:app-lambdaP-bound}
\end{equation}
In particular, if Assumption~\ref{ass:app-growth} holds and $\bm\Omega$ and $R$ are fixed, then Assumption~\ref{ass:app-penalty} holds with $\bar\lambda=0$.
\end{appproposition}

\begin{proof}
The expectation formula for a Gamma distribution gives
\[
\mathbb E_q(\kappa)
=
\frac{a_0+d/2}{b_0+\frac12\mathbb E_q[\bm\theta^\top\bm\Omega\,\bm\theta]}.
\]
Since $b_0>0$ and $\mathbb E_q[\bm\theta^\top\bm\Omega\,\bm\theta]\ge 0$, \eqref{eq:app-eta-bound} follows immediately. Then \eqref{eq:app-lambdaP-bound} is just the definition of $\lambda_P$. The final statement follows by combining \eqref{eq:app-lambdaP-bound} with Lemma~\ref{lem:app-operator-norm}. 
This bound is pathwise, even when $q$ depends on the data. In particular, $\sup_{\mathcal A}\lambda_P\le C_\eta/(2\underline a)\to0$, proving uniform penalty control on the whole radius-$R$ ball with $\bar\lambda=0$.
\end{proof}

\begin{applemma}[Uniform Poisson averages with data-dependent weights]
\label{lem:app-weighted-ulln}
Under Assumption~\ref{ass:app-growth}, let $\mathcal F$ be a uniformly bounded class of measurable functions with finite nets in the supremum norm at every positive resolution. Write $\Lambda_i:=\int_{\mathcal B}\lambda_i$ and
\[
\Delta_i(\mathcal F):=
\sup_{f\in\mathcal F}\left|\frac1{T_i}\int f\,dN_i-\int_{\mathcal B}f\lambda_i\right|.
\]
Then $\max_i\Delta_i(\mathcal F)\xrightarrow{p}0$, and, on every sample,
\begin{equation}
\sup_{\substack{\bm\nu\in[0,1]^n\,:\,a>0}}\ \sup_{f\in\mathcal F}
\left|\frac1a\sum_i\nu_i\int f\,dN_i-\int_{\mathcal B}f\lambda_{0,w}\right|
\le\max_i\Delta_i(\mathcal F).
\label{eq:app-weighted-ulln}
\end{equation}
Consequently the bound also applies to arbitrary data-dependent responsibilities.
\end{applemma}

\begin{proof}
For each fixed $f$ and each subject, the Poisson integral identities give
\[
\mathbb E\left[\frac1{T_i}\int f\,dN_i\right]=\int f\lambda_i,
\qquad
\operatorname{Var}\left(\frac1{T_i}\int f\,dN_i\right)
=\frac1{T_i}\int f^2\lambda_i\longrightarrow0.
\]
Thus
\begin{equation}
\frac1{T_i}\int f\,dN_i-\int f\lambda_i\xrightarrow{p}0.
\label{eq:app-pointwise-lln}
\end{equation}
For a finite supremum-norm $\varepsilon$-net $f_1,\ldots,f_M$, positivity of the empirical and population measures yields
\[
\Delta_i(\mathcal F)
\le\max_{1\le j\le M}\left|\frac1{T_i}\int f_j\,dN_i-\int f_j\lambda_i\right|
+\varepsilon\left\{\frac{N_i(\mathcal B)}{T_i}+\Lambda_i\right\}.
\]
The maximum tends to zero for fixed $\varepsilon$, and $N_i(\mathcal B)/T_i\xrightarrow{p}\Lambda_i$. Letting $\varepsilon\downarrow0$ and using fixed $n$ proves the first claim. For any nonnegative weights with $a>0$, the discrepancy in \eqref{eq:app-weighted-ulln} equals
\[
\sum_iw_i\left\{\frac1{T_i}\int f\,dN_i-\int f\lambda_i\right\}.
\]
Since $\sum_iw_i=1$, its absolute value is at most $\max_i\Delta_i(\mathcal F)$, independently of how the weights are chosen. This proves the pathwise bound without applying a Poisson variance formula to data-dependent weights.
\end{proof}

Theorem~\ref{thm:app-uniform-core} is the uniform convergence step for the strictly positive chamber empirical objective. It shows that the empirical criterion \(\widetilde J\) uniformly approximates the population criterion \(Q\) over the compact core \(D_\delta\). This result is needed because pointwise convergence alone is not enough to justify convergence of the maximizer. Uniform convergence ensures that the whole objective surface is close to its population limit, so that the empirical maximizer cannot move away from the unique population mode \(\bm\theta^\star\).

\begin{apptheorem}[Uniform convergence on the compact positive core]
\label{thm:app-uniform-core}
Suppose Assumptions~\ref{ass:app-growth} and~\ref{ass:app-penalty} hold. Then
\begin{equation}
\sup_{(\bm\nu,\eta)\in\mathcal A}
\sup_{\bm\theta\in D_\delta} |\widetilde J(\bm\theta)-Q(\bm\theta)| \xrightarrow{p} 0.
\label{eq:app-uniform-core}
\end{equation}
\end{apptheorem}

\begin{proof}
Write
\begin{align*}
\widetilde J(\bm\theta)-Q(\bm\theta)
=
2\left[
\frac{1}{a}\sum_{i=1}^n \nu_i \sum_{j\in\mathcal I_i}\log s_{\bm\theta}(\bm{y}_{ij})
-
\int_{\mathcal B}\lambda_0(\bm y)\log s_{\bm\theta}(\bm y)\,d\bm{y}
\right]
-
(\lambda_P-\bar\lambda)\bm\theta^\top\bm\Omega\,\bm\theta.
\end{align*}
For the stochastic term define
\[
f_{\bm \theta}(y):=\log s_{\bm\theta}(\bm y).
\]
By Lemma~\ref{lem:app-bounded-basis}, $B_{\max}<\infty$. Since $\bm\theta\in D_\delta$ implies $s_{\bm\theta}(\bm y)\ge \delta$ and $s_{\bm\theta}(\bm y)\le B_{\max}R$, we have
\begin{align*}
|f_{\bm\theta}(\bm y)|
\le
\max\{|\log\delta|,\ |\log(B_{\max}R)|\}
=: V
\qquad\text{for all } \bm\theta\in D_\delta,\ \bm y\in\mathcal B.
\end{align*}
Moreover, for $\bm \theta,\bm\theta'\in D_\delta$,
\begin{equation}
|f_{\bm\theta}(\bm y)-f_{\bm \theta'}(\bm y)|
\le
\frac{B_{\max}}{\delta}\|\bm\theta-\bm\theta'\|_2
\qquad\text{for all } \bm y\in\mathcal B.
\label{eq:app-log-Lipschitz}
\end{equation}

The boundedness and Lipschitz estimates above, together with compactness of $D_\delta$, show that $\mathcal F=\{f_{\bm\theta}:\bm\theta\in D_\delta\}$ satisfies Lemma~\ref{lem:app-weighted-ulln}. Therefore
\begin{equation}
\sup_{(\bm\nu,\eta)\in\mathcal A}\sup_{\bm\theta\in D_\delta}
\left|\frac1a\sum_i\nu_i\sum_{j\in\mathcal I_i}\log s_{\bm\theta}(\bm y_{ij})
-\int_{\mathcal B}\lambda_0(\bm y)\log s_{\bm\theta}(\bm y)\,d\bm y\right|
\xrightarrow{p}0.
\label{eq:app-uniform-log}
\end{equation}
Combining \eqref{eq:app-uniform-log} with the uniform penalty control in Assumption~\ref{ass:app-penalty} proves \eqref{eq:app-uniform-core}.
\end{proof}

\subsection{Full-sign Population Comparison}

This subsection extends the population comparison beyond the positive chamber. 
Because the squared link is invariant under the transformation \(\bm\theta\mapsto-\bm\theta\), the positive and negative chambers represent the same intensity surface. 
However, sign-changing coefficients can create nodal regions where the linear predictor crosses zero, leading to undesirable modes of the squared-link objective. 
Since $Q$ uses $\log|s_{\bm\theta}|$ and is defined on $\bar{\mathbb B}(\bm 0,R)$, it allows the positive chamber, the negative chamber, and the sign-changing region to be compared under a common objective. 
The main result of this subsection shows that sign-changing candidates are uniformly separated from the admissible chambers: the only population maximizers over the full sign-comparison class are the two symmetric modes \(\pm\bm\theta^\star\). 
This population separation provides the deterministic gap used later to prove empirical separation and exponential-weight dominance of the admissible chambers.

Define the sign-changing region inside the radius-$R$ ball by
\begin{align*}
S^{\mathrm{sc}}
:=
\left\{
\bm\theta\in\mathbb R^d:\ \|\bm\theta\|_2\le R,\ 
\sup_{\bm y\in\mathcal B}s_{\bm\theta}(\bm y)>0,\ 
\inf_{\bm y\in\mathcal B}s_{\bm\theta}(\bm y)<0
\right\},
\end{align*}
and the corresponding full sign-comparison class by
\begin{align*}
\Theta^{\pm,\mathrm{sc}}
:=
D\cup(-D)\cup S^{\mathrm{sc}}.
\end{align*}

\begin{appassumption}[Path-connected domain, neighborhood differentiability, and local volume]
\label{ass:app-slope}
The set $\mathcal B$ is path-connected. There exists an open convex neighborhood $U\supset\mathcal B$ such that each component of $B$ extends to a continuously differentiable function on $U$. Define
\begin{equation}
L := R \sup_{\bm y\in U}\|\nabla \bm B(\bm y)\|_2,
\label{eq:app-L}
\end{equation}
and assume $L>0$. Let
\begin{align*}
r^{\mathrm{sc}} := \frac{c_0}{2L},
\qquad
v^{\mathrm{sc}} := \inf_{\bm y\in\mathcal B} |\mathcal B\cap \mathbb B(\bm y,r^{\mathrm{sc}})|,
\end{align*}
and assume $v^{\mathrm{sc}}>0$.
\end{appassumption}

Assumption~\ref{ass:app-slope} provides the geometric regularity needed to turn sign changes into a uniform population loss. If \(s_{\bm\theta}\) is genuinely sign-changing on a path-connected domain, then it must vanish at some point. The differentiability and bounded-coefficient conditions imply a uniform Lipschitz bound for \(s_{\bm\theta}\), so the linear predictor remains small on a neighborhood of this nodal point. The local-volume condition ensures that this neighborhood has uniformly positive Lebesgue measure, even near the boundary of the domain. Consequently, sign-changing candidates incur a non-negligible loss in the logarithmic term relative to the strictly positive target \(g_0\), which is the key ingredient for excluding nodal-line competitors at the population level.

\begin{applemma}[Verification of Assumption~\ref{ass:app-slope} for rectangular domains]
Suppose that \(\mathcal B\subset\mathbb R^H\) is a compact rectangular domain and that each component of the basis vector \(\bm B(\cdot)\) extends to a continuously differentiable function on an open neighborhood \(U\supset\mathcal B\). Assume also that
\[
0< \sup_{\bm y\in U}\|\nabla \bm B(\bm y)\|_2 <\infty .
\]
Then Assumption~\ref{ass:app-slope} holds. 
\end{applemma}

\begin{proof}
A compact rectangular domain is path-connected. 
By the assumption, 
\[
\sup_{\bm y\in U}\|\nabla \bm B(\bm y)\|_2
\]
is finite. Hence \(0<L<\infty\) and \(r^{\mathrm{sc}}>0\).
It remains to verify the local-volume lower bound. For a compact rectangular domain, there exists a constant \(c_{\mathcal B}>0\), depending only on the domain and the dimension, such that for every \(\bm y\in\mathcal B\) and every sufficiently small \(r>0\),
\[
|\mathcal B\cap \mathbb B(\bm y,r)|\ge c_{\mathcal B}r^H .
\]
Applying this with \(r=r^{\mathrm{sc}}\) gives
\[
v^{\mathrm{sc}}
\ge c_{\mathcal B}(r^{\mathrm{sc}})^H>0 .
\]
Thus Assumption~\ref{ass:app-slope} holds.
\end{proof}

\begin{appexample}[Tensor-product B-spline bases]
\label{exm:bs4}
For tensor-product B-spline bases on a compact rectangular domain, the conditions of the preceding lemma hold when the spline degree is at least two and the knot sequence is fixed. The basis functions are then continuously differentiable and have bounded derivatives on a neighborhood of the domain. Hence Assumption~\ref{ass:app-slope} holds.
\end{appexample}

Define the sign-changing penalty discrepancy by
\begin{align*}
\Delta_P^{\mathrm{sc}}
:=
\sup_{\bm\theta\in S^{\mathrm{sc}}}
\left|
\bm\theta^{+\top}\bm\Omega\,\bm\theta^+
-
\bm\theta^\top\bm\Omega\,\bm\theta
\right|.
\end{align*}
Since $S^{\mathrm{sc}}$ is bounded and $\bm\Omega$ is fixed, one has $\Delta_P^{\mathrm{sc}}<\infty$.
Define the nodal-loss constant by
\begin{equation}
\Gamma^{\mathrm{sc}}
:=
\inf_{\bm y\in\mathcal B}\ \inf_{0<u\le c_0/2}
\left\{
\ell_y(u)-\ell_y(g_0(\bm y))
\right\}.
\label{eq:app-Gamma-sc}
\end{equation}
By Assumption~\ref{ass:app-target}, $g_0(\bm y)\in[c_0,C_0]$ for all $\bm y$. For each fixed $\bm y$, the map $u\mapsto \ell_y(u)-\ell_y(g_0(\bm y))$ is strictly decreasing on $(0,g_0(\bm y)]$, so the inner infimum in \eqref{eq:app-Gamma-sc} is attained at $u=c_0/2$. Hence
\[
\Gamma^{\mathrm{sc}}
=
\inf_{t\in[c_0,C_0]}
\left\{
\left(\frac{c_0}{2}\right)^2 - 2t^2\log\!\left(\frac{c_0}{2}\right)
-
\bigl( t^2 - 2t^2\log t \bigr)
\right\}
>0.
\]

Proposition~\ref{prop:app-signchanging} is the population-level exclusion result for sign-changing competitors. 
It plays the same role as the boundary-strip exclusion result in Proposition~\ref{prop:app-boundary-strip}, but now the bad region is the set of coefficients whose linear predictor changes sign. 
The proposition shows that, if the nodal-loss gap is large enough relative to the approximation error of the interior witness and the penalty discrepancy, then every sign-changing candidate has strictly smaller \(Q\)-value than the interior positive candidate \(\bm\theta^+\). 
Thus sign-changing coefficients cannot be population maximizers. 
This result is the key step for showing that the only relevant population modes of the squared-link objective are the two symmetric modes in the positive and negative chambers.

\begin{appproposition}[Population exclusion of sign-changing competitors]
\label{prop:app-signchanging}
Suppose Assumptions~\ref{ass:app-target}--\ref{ass:app-boundary} and~\ref{ass:app-slope} hold. If
\begin{equation}
\bar\lambda\,\Delta_P^{\mathrm{sc}}
<
\Gamma^{\mathrm{sc}}v^{\mathrm{sc}} - M_\varepsilon |\mathcal B|\,\varepsilon^2,
\label{eq:app-small-penalty-sc}
\end{equation}
then
\begin{equation}
\sup_{\bm\theta\in S^{\mathrm{sc}}} Q(\bm\theta)
<
Q(\bm\theta^+)
\le
\sup_{\bm \theta\in D} Q(\bm\theta).
\label{eq:app-signchanging-gap}
\end{equation}
\end{appproposition}

\begin{proof}
Let $\bm\theta\in S^{\mathrm{sc}}$. If $Q(\bm\theta)=-\infty$, then the desired inequality is immediate. So assume $Q(\bm\theta)>-\infty$.
Since $\mathcal B$ is path-connected and $\bm y\mapsto s_{\bm\theta}(\bm y)$ is continuous, there exists $\bm y_\theta\in\mathcal B$ such that
\begin{align*}
s_{\bm{\theta}}(\bm{y}_{\bm\theta})=0.
\end{align*}
For any $\bm y\in\mathcal B$, the line segment from $\bm y_{\bm\theta}$ to $\bm y$ lies in $U$ because $U$ is convex. By the mean-value theorem and \eqref{eq:app-L},
\begin{align*}
|s_{\bm\theta}(\bm y)|
=
|s_{\bm\theta}(\bm y)-s_{\bm\theta}(\bm{y}_{\bm\theta})|
\le
\sup_{\xi\in U}\|\nabla s_{\bm\theta}(\xi)\|_2\,\|\bm y-\bm y_{\bm \theta}\|_2
\le
L\|\bm y-\bm y_{\bm\theta}\|_2.
\end{align*}
Hence for every $\bm y\in \mathcal B\cap \mathbb B(\bm y_{\bm\theta},r^{\mathrm{sc}})$,
\begin{align*}
|s_{\bm\theta}(\bm y)| \le Lr^{\mathrm{sc}} = c_0/2.
\end{align*}
By the definitions of $\Gamma^{\mathrm{sc}}$ and $v^{\mathrm{sc}}$,
\begin{align*}
\int_{\mathcal B}\{\ell_{\bm y}(|s_{\bm\theta}(\bm y)|)-\ell_{\bm y}(g_0(\bm y))\}\,d\bm{y}
\ge
\Gamma^{\mathrm{sc}}|\mathcal B\cap \mathbb B(\bm y_{\bm\theta},r^{\mathrm{sc}})|
\ge
\Gamma^{\mathrm{sc}}v^{\mathrm{sc}}.
\end{align*}
The witness bound \eqref{eq:app-witness-loss} remains valid. Therefore
\begin{align*}
Q(\bm\theta^+) - Q(\bm\theta)
\ge
\Gamma^{\mathrm{sc}}v^{\mathrm{sc}} - M_\varepsilon |\mathcal B|\varepsilon^2
-\bar\lambda\left(
\bm\theta^{+\top}\bm\Omega\,\bm\theta^+ - \bm\theta^\top\bm\Omega\,\bm\theta
\right),
\end{align*}
hence
\begin{align*}
Q(\bm\theta^+) - Q(\bm\theta)
\ge
\Gamma^{\mathrm{sc}}v^{\mathrm{sc}} - M_\varepsilon |\mathcal B|\varepsilon^2 - \bar\lambda\,\Delta_P^{\mathrm{sc}}.
\end{align*}
Under \eqref{eq:app-small-penalty-sc}, the right-hand side is strictly positive for every $\bm\theta\in S^{\mathrm{sc}}$, which proves \eqref{eq:app-signchanging-gap}.
\end{proof}

It remains to discuss the gap condition in \eqref{eq:app-small-penalty-sc}. 
The next example shows that this condition holds for tensor-product B-spline bases under mild smoothness assumptions. 
The argument parallels Example~\ref{exm:bs3}.

\begin{appexample}[B-spline verification of condition~\eqref{eq:app-small-penalty-sc}]
Suppose that \(\bm B\) is a tensor-product B-spline basis of degree \(p\) on a compact rectangular domain \(\mathcal B\subset\mathbb R^H\), with maximal knot spacing \(h\). Let \(\bm\theta^+\in D_{2\delta}\) and \(\varepsilon\ge \|g_0-s_{\bm\theta^+}\|_\infty\) be as in Example~\ref{exm:app-bs-witness}, with
\[
\varepsilon\le C_1h^\alpha .
\]
For the spline class with \(\|\bm\theta\|_2\le R\), there exists a constant \(C_2>0\) such that
\[
L
=
R\sup_{\bm y\in U}\|\nabla \bm B(\bm y)\|_2
\le C_2h^{-1}.
\]
Hence
\[
r^{\mathrm{sc}}
=
\frac{c_0}{2L}
\ge
C_3h
\]
for some constant \(C_3>0\). Since the domain is rectangular, there exists \(C_4>0\) such that
\[
v^{\mathrm{sc}}
=
\inf_{\bm y\in \mathcal B}|B\cap \mathbb B(\bm y,r^{\mathrm{sc}})|
\ge
C_4 (r^{\mathrm{sc}})^H
\ge
C_5 h^H
\]
for some constant \(C_5>0\). Since \(\Gamma^{\mathrm{sc}}>0\), this gives
\[
\Gamma^{\mathrm{sc}}v^{\mathrm{sc}}
\ge
C_6 h^H
\]
for some constant \(C_6>0\).

Combining this lower bound with \(\varepsilon\le C_1h^\alpha\), we obtain
\[
M_\varepsilon |\mathcal{B}|\varepsilon^2
\le
C_7 M_\varepsilon h^{2\alpha}
\]
for some constant \(C_7>0\). Therefore, if
\[
M_\varepsilon h^{2\alpha-H}\to0,
\]
then
\[
\Gamma^{\mathrm{sc}}v^{\mathrm{sc}}
>
M_\varepsilon |\mathcal{B}|\varepsilon^2
\]
for sufficiently small \(h\). In particular, if \(M_\varepsilon\) remains bounded as \(h\to0\) and \(\alpha>H/2\), this condition holds. Moreover, in the setting covered by Lemma~\ref{lem:app-operator-norm} and Proposition~\ref{prop:app-gamma}, the deterministic penalty level can be taken as \(\bar\lambda=0\). Hence condition~\eqref{eq:app-small-penalty-sc} holds for sufficiently small \(h\).
\end{appexample}

Theorem~\ref{thm:app-global-modes} completes the population separation argument for the full sign-comparison class. 
The squared-link parametrization identifies \(\bm\theta\) and \(-\bm\theta\), so the positive and negative chambers should both contain valid population modes. 
The theorem shows that these are the only population maximizers: after excluding the zero-boundary strip and the sign-changing region, the population criterion \(Q\) is maximized only at the two symmetric points \(\pm\bm\theta^\star\).
Thus nodal-line competitors cannot be population optima, and the sign ambiguity of the squared-link model is reduced to the unavoidable global sign symmetry. 
This result provides the deterministic separation needed later to prove empirical separation and consistency of the full-sign empirical maximizers.

\begin{apptheorem}[Population global modes on the full sign-comparison class]
\label{thm:app-global-modes}
Suppose the conditions of Theorem~\ref{thm:app-positive-mode} and Proposition~\ref{prop:app-signchanging} hold. Then the maximization of $Q$ over $\Theta^{\pm,\mathrm{sc}}$ has exactly two global maximizers, namely
$
\pm \bm\theta^\star$.
Moreover,
\begin{align*}
\sup_{\bm\theta\in S^{\mathrm{sc}}} Q(\bm\theta)
<
Q(\bm\theta^\star)
=
Q(-\bm\theta^\star).
\end{align*}
\end{apptheorem}

\begin{proof}
By Theorem~\ref{thm:app-positive-mode}, $Q$ has the unique maximizer $\bm\theta^\star$ on $D$. By \eqref{eq:app-sign-symmetry}, the negative domain $-D$ has the unique maximizer $-\bm\theta^\star$. Proposition~\ref{prop:app-signchanging} excludes every point in $S^{\mathrm{sc}}$. Hence the only global maximizers on $\Theta^{\pm,\mathrm{sc}}$ are $\pm\bm\theta^\star$.
\end{proof}

\subsection{Empirical Full-sign Control}

This subsection transfers the population separation result to the empirical full-sign objective and the induced exponential weights. 
The main difficulty is that the absolute-log criterion \(Q\) and its empirical counterpart can be singular for sign-changing coefficients, because \(s_{\bm\theta}\) may vanish. 
To avoid this issue, we 
introduce truncated empirical and population criteria, defined on the closed ball $\bar{\mathbb B}(\bm 0,R)$, which is compact and contains $\Theta^{\pm,\mathrm{sc}}$.
The truncation makes the empirical-process comparison well behaved, while preserving the original objective on the admissible chambers \(D_\delta\cup(-D_\delta)\) and upper bounding the contribution of the inadmissible set $\mathcal X_\delta$. 
Using the population gap against sign-changing and boundary-strip competitors and a uniform truncated empirical approximation, we show that the exponential weight assigned to sign-changing candidates is asymptotically negligible relative to the admissible chambers.


Fix $\delta\in(0,1)$ and $c_{\mathrm{tr}}>1$, with $c_{\mathrm{tr}}$ chosen large enough so that $G_\delta := \min\{G^{\mathrm{sc}}_\delta,\, G^{\mathrm{bd}}_\delta\} > 0$. Define
\[
\ell_{\delta}(u)
:=
\log\bigl(|u|\vee \delta^{c_{\mathrm{tr}}}\bigr).
\]
Define
\[
\widetilde J_{\delta}(\bm\theta)
:=
-\bm\theta^\top \bm{M}\bm\theta
+
\frac{2}{a}\sum_{i=1}^n\nu_i\sum_{j\in I_i}
\ell_{\delta}\{s_{\bm\theta}(\bm y_{ij})\}
-\lambda_P\bm\theta^\top\bm\Omega\bm\theta,
\label{eq:app-Jabs-truncated}
\]
and
\[
Q_{\delta}(\bm\theta)
:=
-\bm\theta^\top \bm{M}\bm\theta
+
2\int_{\mathcal{B}}\lambda_0(\bm y)\ell_{\delta}\{s_{\bm\theta}(\bm y)\}\,d\bm{y}
-\bar\lambda\bm\theta^\top\bm\Omega\bm\theta.
\]


\begin{appremark}
\label{rem:app-truncated-fullsign}
The truncated criteria $\widetilde J_{\delta}(\bm\theta),\  Q_{\delta}(\bm\theta)$
are finite on $\bar{\mathbb B}(\bm 0,R)$.
Since $c_{\mathrm{tr}}>1$ implies $\delta^{c_{\mathrm{tr}}}<\delta$, on $D_\delta\cup(-D_\delta)$,
\[
\widetilde J_{\delta}(\bm\theta)=\widetilde J(\bm\theta),
\qquad
Q_{\delta}(\bm\theta)=Q(\bm\theta).
\]
Moreover, on $\mathcal X_\delta$,
\[
\widetilde J(\bm\theta)
\le
\widetilde J_{\delta}(\bm\theta),
\qquad
Q(\bm\theta) \le Q_{\delta}(\bm\theta).
\]
Hence any upper bound for the truncated empirical criterion on $\mathcal X_\delta$ also upper bounds the original empirical criterion on $\mathcal X_\delta$.
\end{appremark}

\begin{appassumption}[Full-sign penalty discrepancy control]
\label{ass:app-fullsign-penalty}
There exists a deterministic benchmark level \(\bar\lambda\in[0,\infty)\) such that
\[
\Delta_{\rm pen}^{\pm,\mathrm{sc}}(\delta)
:=
\sup_{\bm\theta\in \bar{\mathbb B}(\bm 0,R)}
\left|
(\lambda_P-\bar\lambda)\bm\theta^\top\bm\Omega\bm\theta
\right|
\xrightarrow{p}0 .
\]
\end{appassumption}

As in Assumption~\ref{ass:app-penalty}, require $\sup_{\mathcal A}\Delta_{\rm pen}^{\pm,\mathrm{sc}}(\delta)\xrightarrow{p}0$ for the uniform statements.

Assumption~\ref{ass:app-fullsign-penalty} is the full-sign version of Assumption~\ref{ass:app-penalty}. 
It is stated over the entire closed ball $\bar{\mathbb B}(\bm 0,R)$, which contains $\Theta^{\pm,\mathrm{sc}}$.
It follows from the operator-norm argument as Lemma~\ref{lem:app-operator-norm}. 
Indeed,
\[
\sup_{\bm\theta\in \bar{\mathbb B}(\bm 0,R)}
\left|(\lambda_P-\bar\lambda)\bm\theta^\top\bm\Omega\bm\theta\right|
\le
R^2\|\bm\Omega\|_{\mathrm{op}}|\lambda_P-\bar\lambda|.
\]
Thus the Gaussian--Gamma prior calculation in Proposition~\ref{prop:app-gamma} verifies Assumption~\ref{ass:app-fullsign-penalty} with \(\bar\lambda=0\) under Assumption~\ref{ass:app-growth}.

Proposition~\ref{lem:app-truncated-full-sign-ulln} is the full-sign analogue of the uniform convergence result on $\bar{\mathbb B}(\bm 0,R)$. 
The truncation removes the singularity of the absolute logarithm near zeros of \(s_{\bm\theta}\), making the relevant empirical-process class uniformly bounded and Lipschitz on $\bar{\mathbb B}(\bm 0,R)$. 
Together with the growing-exposure condition and the full-sign penalty discrepancy control, this yields uniform convergence of the truncated empirical criterion to its population counterpart.

\begin{appproposition}[Truncated full-sign uniform convergence]
\label{lem:app-truncated-full-sign-ulln}
Fix $c_{\mathrm{tr}}>1$. Suppose Assumption~\ref{ass:app-growth} and Assumption~\ref{ass:app-fullsign-penalty} hold. Suppose also that \(\bm B(\cdot)\) is a fixed finite-dimensional basis satisfying
\[
B_{\max}:=\sup_{\bm y\in\mathcal B}\|\bm B(\bm y)\|_2<\infty,
\]
and that 
the subject intensities $\lambda_i$ are bounded on $\mathcal B$. Then
\begin{align}
\Xi_{\delta,a}
:=
\sup_{\bm\theta\in \bar{\mathbb B}(\bm 0,R)}
\left|
\widetilde J_{\delta}(\bm\theta)
-
Q_{\delta}(\bm\theta)
\right|
\xrightarrow{p}0 .
\label{eq:app-Xi-truncated}
\end{align}

More precisely, $\overline\Xi_{\delta}:=\sup_{(\bm\nu,\eta)\in\mathcal A}\Xi_{\delta,a}\xrightarrow{p}0$.
\end{appproposition}

\begin{proof}
The function \(\ell_{\delta}\) is bounded below by $c_{\mathrm{tr}}\log\delta$. Moreover, on the parameter set \(\|\bm\theta\|_2\le R\),
\[
|s_{\bm\theta}(\bm y)|=|\bm B(\bm y)^\top\bm\theta|
\le B_{\max}R .
\]
Hence the class
\[
\mathcal F_{\delta}
:=
\{\ell_{\delta}(s_{\bm\theta}(\cdot)):\bm\theta\in \bar{\mathbb B}(\bm 0,R)\}
\]
is uniformly bounded. 
Let $C_{\delta} := \max\{c_{\mathrm{tr}}|\log\delta|,\,\log(B_{\max}R\vee\delta^{c_{\mathrm{tr}}})\}$. 
Then
\[
|\ell_{\delta}(s_{\bm\theta}(\bm y))|\le C_{\delta}
\]
for all \(\bm y\in\mathcal B\) and all \(\bm\theta\in \bar{\mathbb B}(\bm 0,R)\).

The truncated logarithm is also Lipschitz. In particular,
\[
\lvert\ell_{\delta}(u)-\ell_{\delta}(v)\rvert \le \delta^{-c_{\mathrm{tr}}}\lvert u-v\rvert
\]
for all \(u,v\in\mathbb R\). Therefore, for all \(\bm\theta,\bm\theta'\in \bar{\mathbb B}(\bm 0,R)\),
\[
\sup_{\bm y\in\mathcal B}
\left|
\ell_{\delta}(s_{\bm\theta}(\bm y))
-
\ell_{\delta}(s_{\bm\theta'}(\bm y))
\right|
\le
\delta^{-c_{\mathrm{tr}}}B_{\max}\lVert\bm\theta-\bm\theta'\rVert_2 .
\]
Thus \(\mathcal F_{\delta}\) is a uniformly bounded and Lipschitz finite-dimensional class.

Compactness of $\bar{\mathbb B}(\bm0,R)$ and the Lipschitz bound give finite supremum-norm nets for $\mathcal F_\delta$. Lemma~\ref{lem:app-weighted-ulln} therefore yields the common bound
\begin{align*}
\overline\Xi_{\delta}
&\le 2\max_i\Delta_i(\mathcal F_\delta)
+\sup_{\mathcal A}\Delta_{\rm pen}^{\pm,\mathrm{sc}}(\delta)
\xrightarrow{p}0.
\end{align*}
The quadratic likelihood term is identical in $\widetilde J_\delta$ and $Q_\delta$, and the second term is controlled by Assumption~\ref{ass:app-fullsign-penalty}. In the bounded-precision VI setting it is at most $C_\eta R^2\|\bm\Omega\|_{\mathrm{op}}/(2\underline a)$. This proves \eqref{eq:app-Xi-truncated} simultaneously for all admissible data-dependent weights.
\end{proof}

For measurable $C\subseteq \bar{\mathbb B}(\bm 0,R)$, define the exponential population weight
\begin{align*}
Z_a(C)
:=
\int_C \exp\{aQ(\bm\theta)\}\,d\bm\theta.
\end{align*}

For measurable $C\subseteq \bar{\mathbb B}(\bm 0,R)$, define the empirical exponential weight
\begin{align*}
\widehat Z_a(C)
:=
\int_C \exp\{a\widetilde J(\bm\theta)\}\,d\bm\theta.
\end{align*}
\begin{align}
G^{\mathrm{sc}}_{\delta}
&:= Q(\bm\theta^\star)-
\sup_{\bm\theta\in S^{\mathrm{sc}}}Q_{\delta}(\bm\theta),\label{eq:app-truncated-sc-gap}\\
G^{\mathrm{bd}}_{\delta}
&:= Q(\bm\theta^\star)-
\sup_{\bm\theta\in D_\delta^{\mathrm{bd}}}Q_{\delta}(\bm\theta),\label{eq:app-truncated-bd-gap}\\
G_\delta &:= \min\bigl\{G^{\mathrm{sc}}_\delta,\,G^{\mathrm{bd}}_\delta\bigr\}.\label{eq:app-G-delta}
\end{align}
We choose $c_{\mathrm{tr}}>1$ large enough so that $G_\delta>0$.
This is possible because $G^{\mathrm{sc}}_\delta>0$ for $c_{\mathrm{tr}}$ large by the nodal-neighborhood argument in Proposition~\ref{prop:app-signchanging} applied with the truncated log, and $G^{\mathrm{bd}}_\delta>0$ for $c_{\mathrm{tr}}$ large by Proposition~\ref{prop:app-truncated-boundary-gap} below.



\begin{appproposition}[Truncated boundary-strip gap]
\label{prop:app-truncated-boundary-gap}
Suppose Assumptions~\ref{ass:app-target}--\ref{ass:app-boundary} and~\ref{ass:app-slope} hold, and that $\delta<c_0/2$. 
Suppose also that $\bm B(\cdot)$ satisfies the conditions of Lemma~\ref{lem:app-bounded-basis} and Assumption~\ref{ass:app-slope}.
Assume also that \(g_0\) is Lipschitz on \(\mathcal B\).
Define the clamped boundary-loss constant
\begin{equation}
\gamma^{\mathrm{bd}}
:=
\left(\frac{c_0}{4}\right)^2 v^{\mathrm{bd}},
\quad
v^{\mathrm{bd}}
:=
\inf_{\bm y\in\mathcal B}|\mathcal B\cap\mathbb B(\bm y,r^{\mathrm{bd}})|,
\quad
r^{\mathrm{bd}}
:=
\frac{c_0/4}{L_0+L},
\label{eq:app-gamma-bd}
\end{equation}
where $L_0$ is a Lipschitz constant of $g_0$ and $L$ is from Assumption~\ref{ass:app-slope}. Then $\gamma^{\mathrm{bd}}>0$ and
\begin{equation}
G^{\mathrm{bd}}_\delta
\ge
\gamma^{\mathrm{bd}}
- M_\varepsilon|\mathcal B|\,\varepsilon^2
- \bar\lambda\,\Delta_P(\delta)
- 2\delta^{2c_{\mathrm{tr}}}|\mathcal B|.
\label{eq:app-Gbd-lower}
\end{equation}
In particular, if
\begin{equation}
\gamma^{\mathrm{bd}} > M_\varepsilon|\mathcal B|\,\varepsilon^2 + \bar\lambda\,\Delta_P(\delta),
\label{eq:app-small-penalty-bd-trunc}
\end{equation}
then $G^{\mathrm{bd}}_\delta>0$ for all $c_{\mathrm{tr}}$ large enough. The same bound holds for $\sup_{-D_\delta^{\mathrm{bd}}}Q_\delta$ by sign symmetry.
\end{appproposition}

\begin{proof}
Let $\bm\theta\in D_\delta^{\mathrm{bd}}$ and define the truncated predictor
\[
\check s_{\bm\theta}(\bm y)
:=
s_{\bm\theta}(\bm y)\vee\delta^{c_{\mathrm{tr}}}
=
\max\{s_{\bm\theta}(\bm y),\,\delta^{c_{\mathrm{tr}}}\}.
\]
Since $s_{\bm\theta}>0$ on $D_\delta^{\mathrm{bd}}$, we have $\check s_{\bm\theta}\ge s_{\bm\theta}>0$, and $\ell_\delta(s_{\bm\theta}(\bm y))=\log\check s_{\bm\theta}(\bm y)$.

Define $\mathcal Q(s):=-\int_{\mathcal B}s(\bm y)^2\,d\bm y+2\int_{\mathcal B}\lambda_0(\bm y)\log s(\bm y)\,d\bm y$ for $s>0$.
Then
\[
Q_\delta(\bm\theta)
=
-\int_{\mathcal B} s_{\bm\theta}(\bm y)^2 d\bm y
+2\int_{\mathcal B} \lambda_0(\bm y) \log\check s_{\bm\theta}(\bm y) d\bm y
-\bar\lambda\,\bm\theta^\top\bm\Omega\bm\theta.
\]
Since $s_{\bm\theta}\le\check s_{\bm\theta}$, we have $s_{\bm\theta}^2\le\check s_{\bm\theta}^2$, and $-\int s_{\bm\theta}^2 = -\int\check s_{\bm\theta}^2+\int(\check s_{\bm\theta}^2-s_{\bm\theta}^2)$.
On $\{s_{\bm\theta}<\delta^{c_{\mathrm{tr}}}\}$, $\check s_{\bm\theta}^2-s_{\bm\theta}^2=(\delta^{c_{\mathrm{tr}}}-s_{\bm\theta})(\delta^{c_{\mathrm{tr}}}+s_{\bm\theta})<2\delta^{2c_{\mathrm{tr}}}$; elsewhere the difference is zero. Hence
$\int(\check s_{\bm\theta}^2-s_{\bm\theta}^2)\le 2\delta^{2c_{\mathrm{tr}}}|\mathcal B|$.
Therefore
\begin{equation}
Q_\delta(\bm\theta)
\le
\mathcal Q(\check s_{\bm\theta})
+2\delta^{2c_{\mathrm{tr}}}|\mathcal B|
-\bar\lambda\,\bm\theta^\top\bm\Omega\bm\theta.
\label{eq:proof-Qbd-upper}
\end{equation}

Since $\bm\theta\in D_\delta^{\mathrm{bd}}$, there exists $\bm y_0\in\mathcal B$ with $s_{\bm\theta}(\bm y_0)<\delta$.
Hence $\check s_{\bm\theta}(\bm y_0)=\max\{s_{\bm\theta}(\bm y_0),\delta^{c_{\mathrm{tr}}}\}\le\delta<c_0/2$, so
\[
g_0(\bm y_0)-\check s_{\bm\theta}(\bm y_0)
\ge c_0-\delta>c_0/2.
\]
The function $\check s_{\bm\theta}=\max\{s_{\bm\theta},\delta^{c_{\mathrm{tr}}}\}$ is Lipschitz with constant $L$ (same as $s_{\bm\theta}$, since max with a constant preserves Lipschitz). Thus $g_0-\check s_{\bm\theta}$ is Lipschitz with constant at most $L_0+L$. For $r^{\mathrm{bd}}=c_0/(4(L_0+L))$ and all $\bm y\in\mathcal B\cap\mathbb B(\bm y_0,r^{\mathrm{bd}})$,
\[
g_0(\bm y)-\check s_{\bm\theta}(\bm y)>c_0/4.
\]
Since $\check s_{\bm\theta}\ge\delta^{c_{\mathrm{tr}}}>0$ and $g_0\ge c_0>0$, both are in $(0,\infty)$ and $\ell_{\bm y}(u):=u^2-2\lambda_0(\bm y)\log u$ satisfies $\ell_{\bm y}''(u)=2+2\lambda_0(\bm y)u^{-2}\ge2$, minimized at $g_0(\bm y)$. Therefore $\ell_{\bm y}(\check s_{\bm\theta}(\bm y))-\ell_{\bm y}(g_0(\bm y))\ge(\check s_{\bm\theta}(\bm y)-g_0(\bm y))^2\ge(c_0/4)^2$ on $\mathcal B\cap\mathbb B(\bm y_0,r^{\mathrm{bd}})$.
Integrating over $\mathcal B$ (and restricting to $\mathcal B\cap\mathbb B(\bm y_0,r^{\mathrm{bd}})$, where $|\mathcal B\cap\mathbb B(\bm y_0,r^{\mathrm{bd}})|\ge v^{\mathrm{bd}}>0$):
\begin{equation}
\mathcal Q(g_0)-\mathcal Q(\check s_{\bm\theta})
=
\int_{\mathcal B}\{\ell_{\bm y}(\check s_{\bm\theta}(\bm y))-\ell_{\bm y}(g_0(\bm y))\}\,d\bm y
\ge
\left(\frac{c_0}{4}\right)^2 v^{\mathrm{bd}}
=\gamma^{\mathrm{bd}}.
\label{eq:proof-Qbd-gap}
\end{equation}

From \eqref{eq:app-witness-loss}, 
$\mathcal Q(s_{\bm\theta^+})\ge\mathcal Q(g_0)-M_\varepsilon|\mathcal B|\varepsilon^2$.
Since $Q(\bm\theta^\star)\ge Q(\bm\theta^+)=\mathcal Q(s_{\bm\theta^+})-\bar\lambda(\bm\theta^+)^\top\bm\Omega\bm\theta^+$, we get
\[
\mathcal Q(g_0)
\le
Q(\bm\theta^\star)+M_\varepsilon|\mathcal B|\varepsilon^2+\bar\lambda(\bm\theta^+)^\top\bm\Omega\bm\theta^+.
\]
Combining with \eqref{eq:proof-Qbd-upper} and \eqref{eq:proof-Qbd-gap}:
\begin{align*}
Q_\delta(\bm\theta)
&\le
\mathcal Q(g_0)-\gamma^{\mathrm{bd}}+2\delta^{2c_{\mathrm{tr}}}|\mathcal B|-\bar\lambda\,\bm\theta^\top\bm\Omega\bm\theta\\
&\le
Q(\bm\theta^\star)+M_\varepsilon|\mathcal B|\varepsilon^2+\bar\lambda\bigl[(\bm\theta^+)^\top\bm\Omega\bm\theta^+-\bm\theta^\top\bm\Omega\bm\theta\bigr]-\gamma^{\mathrm{bd}}+2\delta^{2c_{\mathrm{tr}}}|\mathcal B|.
\end{align*}
Taking the supremum over $D_\delta^{\mathrm{bd}}$ and using $\Delta_P(\delta):=\sup_{D_\delta^{\mathrm{bd}}}|(\bm\theta^+)^\top\bm\Omega\bm\theta^+-\bm\theta^\top\bm\Omega\bm\theta|$ gives \eqref{eq:app-Gbd-lower}.

Under condition \eqref{eq:app-small-penalty-bd-trunc}, the right-hand side of \eqref{eq:app-Gbd-lower} is positive for $c_{\mathrm{tr}}$ large enough since $2\delta^{2c_{\mathrm{tr}}}|\mathcal B|\to0$ as $c_{\mathrm{tr}}\to\infty$.

The claim for $-D_\delta^{\mathrm{bd}}$ follows from sign symmetry $Q_\delta(\bm\theta)=Q_\delta(-\bm\theta)$.
\end{proof}

\begin{appexample}[B-spline verification of condition~\eqref{eq:app-small-penalty-bd-trunc}]
\label{exm:bs6}
Under the setting of Examples~\ref{exm:bs3} and~\ref{exm:bs4}, $\gamma^{\mathrm{bd}}\gtrsim h^H$ and $M_\varepsilon|\mathcal B|\varepsilon^2\lesssim M_\varepsilon h^{2\alpha}$ (from Example~\ref{exm:bs3} and B-spline Lipschitz constant $L\le C h^{-1}$). If $\bar\lambda=0$ (Proposition~\ref{prop:app-gamma} setting), condition~\eqref{eq:app-small-penalty-bd-trunc} reduces to $\gamma^{\mathrm{bd}}>M_\varepsilon|\mathcal B|\varepsilon^2$, which holds for sufficiently small $h$ whenever $\alpha>H/2$.
\end{appexample}

Proposition~\ref{prop:app-pop-dominance-truncated}  shows that the truncated population weight of any measurable subset of $\mathcal X_\delta$ is exponentially dominated by the population weight near \(\bm\theta^\star\) inside \(D_\delta\). 
The exponent is governed by the gap $G_{\delta}$, reduced by the local oscillation \(\omega(r)\) of \(Q\) around \(\bm\theta^\star\).

\begin{appproposition}[Truncated population dominance of the admissible chamber]
\label{prop:app-pop-dominance-truncated}
Assume $G_{\delta}>0$. Let $r>0$ be such that $v^\star(r)>0$, with $v^\star(r)$ and $\omega(r)$ defined in the main text. 
Then, for every measurable
$C\subseteq \mathcal X_\delta$,
\begin{equation}
\frac{
\int_C \exp\{aQ_{\delta}(\bm\theta)\}\,d\bm\theta
}{
Z_a({D_\delta})
}
\le
\frac{|C|}{v^\star(r)}
\exp\left\{
-a\bigl(G_{\delta}-\omega(r)\bigr)
\right\}.
\label{eq:app-pop-ratio-sc-truncated}
\end{equation}
Consequently, if $\omega(r)<G_{\delta}$, then
\begin{equation}
\int_{\mathcal X_\delta}
\exp\{aQ_{\delta}(\bm\theta)\}\,d\bm\theta
=
o\!\left(
Z_a({D_\delta})
+
Z_a(-{D_\delta})
\right)
\qquad\text{as }a\to\infty .
\label{eq:app-pop-mass-dominance-truncated}
\end{equation}
\end{appproposition}

\begin{proof}
For $\bm\theta\in C\subseteq \mathcal X_\delta$, the definition of $G_\delta$ gives
\[
Q_{\delta}(\bm\theta)
\le
Q(\bm\theta^\star)-G_{\delta}.
\]
Therefore
\[
\int_C \exp\{aQ_{\delta}(\bm\theta)\}\,d\bm\theta
\le
|C|\exp\{a(Q(\bm\theta^\star)-G_{\delta})\}.
\]
On the other hand, for
$\bm\theta\in D_\delta\cap \mathbb B(\bm\theta^\star,r)$,
\[
Q(\bm\theta)\ge Q(\bm\theta^\star)-\omega(r),
\]
so
\[
\int_{D_\delta}\exp\{aQ(\bm\theta)\}\,d\bm\theta
\ge
v^\star(r)\exp\{a(Q(\bm\theta^\star)-\omega(r))\}.
\]
Dividing the two bounds gives the claim.

Taking $C=\mathcal X_\delta$ gives
\[
\frac{
\int_{\mathcal X_\delta}
\exp\{aQ_{\delta}(\bm\theta)\}\,d\bm\theta
}{
Z_a({D_\delta})
}
\le
\frac{|\mathcal X_\delta|}{v^\star(r)}
\exp\left\{
-a\bigl(G_{\delta}-\omega(r)\bigr)
\right\}.
\]
Since $|\mathcal X_\delta| < \infty$ and
$\omega(r)<G_{\delta}$, the right-hand side tends to zero. The denominator
in \eqref{eq:app-pop-mass-dominance-truncated} is at least
$Z_a({D_\delta})$, so the displayed ratio also implies
\eqref{eq:app-pop-mass-dominance-truncated}.
\end{proof}

\begin{appassumption}[Common population gap and interiority]
\label{ass:app-uniform-family}
The domain $\mathcal B$ is compact, $\bm B$ is continuous and finite-dimensional, $\bm M\succ0$, $\bm\Omega\succeq0$, and $\bar\lambda\ge0$. The constants $\delta\in(0,1)$, $R$, and $c_{\mathrm{tr}}>1$ are common to the compact weight family $\mathcal W$ in \eqref{eq:app-admissible-pairs}. Assume $D_\delta$ is nonempty, and let $\bm\theta^\star(w)$ be the unique maximizer of $Q_w$ on $D_\delta$, which exists by compactness and strong concavity. There exist $\rho>0$ and $\underline G_\delta>0$ such that
\begin{equation}
\overline{\mathbb B}(\bm\theta^\star(w),2\rho)\subset D_\delta
\quad(w\in\mathcal W),\qquad
\inf_{w\in\mathcal W}G_\delta(w)\ge\underline G_\delta>0.
\label{eq:app-common-interiority-gap}
\end{equation}
Here $G_\delta(w)=Q_w(\bm\theta^\star(w))-\sup_{\bm\theta\in \mathcal X_\delta}Q_{\delta,w}(\bm\theta)$, equivalently the minimum of the two gaps in the main text. The population results above provide one set of sufficient conditions; Proposition~\ref{prop:app-full-simplex} provides an alternative verification for splines, without requiring the particular slope-based sufficient inequalities used above.
\end{appassumption}

\begin{appremark}[Verification over the weight family]
\label{rem:app-uniform-family}
This assumption concerns the population geometry, not independence between the weights and the data. The population arguments above apply to each weighted target $\lambda_{0,w}$ without changing their formulas. They must be verified over the chosen $\mathcal W$; validity for one target does not imply validity for every mixture.

For a common clipping constant, $Q_w$ on $D_\delta$ and $Q_{\delta,w}$ on $\bar{\mathbb B}(\bm0,R)$ are uniformly continuous in $w$, because the logarithms are bounded there and $\lambda_{0,w}$ is linear in $w$. Uniqueness and compactness imply that $w\mapsto\bm\theta^\star(w)$ is continuous: every subsequential limit of maximizers maximizes the limiting criterion. Consequently $G_\delta(w)$ is continuous, since taking a supremum changes a function by at most its uniform perturbation. Thus positive gaps for all $w$ give a common positive gap on compact $\mathcal W$.

Similarly, pointwise interiority gives a common $\rho$ by continuity of $\bm\theta^\star(w)$ and compactness. For the positivity constraint, the strict boundary-strip gap already implies $\inf_{\bm y}s_{\bm\theta^\star(w)}(\bm y)>\delta$: equality would place $t\bm\theta^\star(w)$ in the boundary strip for $t<1$ close to one, contradicting the gap as $t\uparrow1$. A sufficient condition for strict interiority of the radius constraint is
\[
R^2\lambda_{\min}(\bm M)>\max_i\int_{\mathcal B}\lambda_i(\bm y)\,d\bm y.
\]
Indeed, if $\|\bm\theta^\star(w)\|_2=R$, then
\[
\left.\frac{d}{dt}Q_w(t\bm\theta^\star(w))\right|_{t=1}
=2\int\lambda_{0,w}-2\bm\theta^{\star\top}(\bm M+\bar\lambda\bm\Omega)\bm\theta^\star<0,
\]
so moving slightly inward increases $Q_w$, a contradiction. These observations provide sufficient checks for \eqref{eq:app-common-interiority-gap}. Proposition~\ref{prop:app-full-simplex} below verifies these requirements on the full probability simplex.
\end{appremark}

\begin{applemma}[A local polynomial loss bound]
\label{lem:app-local-polynomial-loss}
Fix $0<\delta<c_0<C_0<\infty$ and a degree $p$. Let $\mathcal P_p$ be the space of polynomials of degree at most $p$ in each coordinate on $\widehat K=[0,1]^H$. Then
\begin{equation}
\kappa_0:=
\inf_{\substack{q\in\mathcal P_p\\ \min_{\widehat K}|q|\le\delta}}
\int_{\widehat K}\operatorname{dist}(|q(\bm u)|,[c_0,C_0])^2\,d\bm u
>0.
\label{eq:app-local-polynomial-loss}
\end{equation}
Consequently, for a continuous tensor-product spline on a rectangular mesh with every cell volume at least $c_{\rm mesh}h^H$, there is $\kappa>0$, independent of the spline coefficients and $h$, such that
\begin{equation}
\int_{\mathcal B}\operatorname{dist}(|s_{\bm\theta}|,[c_0,C_0])^2
\ge\kappa h^H
\quad\text{whenever }\min_{\mathcal B}|s_{\bm\theta}|\le\delta.
\label{eq:app-cell-loss}
\end{equation}
\end{applemma}

\begin{proof}
If the infimum in \eqref{eq:app-local-polynomial-loss} were zero, choose feasible $q_j$ with integral tending to zero. The pointwise inequality
$|q_j|\le C_0+\operatorname{dist}(|q_j|,[c_0,C_0])$
implies a uniform $L^2(\widehat K)$ bound. Equivalence of norms on the fixed finite-dimensional space $\mathcal P_p$ yields a uniformly convergent subsequence, with limit $q\in\mathcal P_p$. The constraint is closed under uniform convergence, so $\min|q|\le\delta$. But the limiting integral is zero; continuity then implies $|q|\in[c_0,C_0]$ everywhere, a contradiction.

For \eqref{eq:app-cell-loss}, choose a cell containing a point where $|s_{\bm\theta}|\le\delta$. On that cell the spline is a polynomial of the stated degree. Its affine pullback to $\widehat K$ is feasible in \eqref{eq:app-local-polynomial-loss}. Changing variables gives a lower bound $\kappa_0$ times the cell volume, hence the claim with $\kappa=c_{\rm mesh}\kappa_0$. Boundary points of cells cause no difficulty because the spline is continuous.
\end{proof}

\begin{appproposition}[Spline verification on the full probability simplex]
\label{prop:app-full-simplex}
Let $\mathcal B\subset\mathbb R^H$ be a compact rectangle with positive volume. Suppose the fixed subject intensities satisfy
\begin{equation}
c_0^2\le\lambda_i(\bm y)\le C_0^2,\qquad
\max_i\|\lambda_i\|_{C^\beta(\mathcal B)}\le M_\lambda,
\label{eq:app-uniform-intensity-regularity}
\end{equation}
for constants $0<c_0<C_0<\infty$ and $\beta>H/2$. Consider tensor-product B-splines of fixed degree $p\ge2$, with $\beta\le p+1$, on quasi-uniform rectangular meshes whose cell side lengths are bounded above and below by fixed multiples of $h$. Assume the usual approximation bound
\[
\inf_{\bm\theta}\|f-\bm B_h^\top\bm\theta\|_\infty
\le C_{\rm app}h^\beta\|f\|_{C^\beta(\mathcal B)}
\]
with $C_{\rm app}$ independent of $h,f$. Take $\bar\lambda=0$, $\bm\Omega\succeq0$, and fix $0<\delta<\min\{1,c_0/4\}$.
For every sufficiently small $h$, one can choose finite $R$ and $c_{\rm tr}>1$ so that Assumption~\ref{ass:app-uniform-family} holds with $\mathcal W=\Delta_{n-1}$. The chosen basis, $R$, $\delta$, and $c_{\rm tr}$ are then held fixed as $T_i\to\infty$; this result does not require a growing basis dimension.
\end{appproposition}

\begin{proof}
Write $g_w=\sqrt{\sum_iw_i\lambda_i}$ for $w\in\Delta_{n-1}$. Convexity of the weights gives $c_0\le g_w\le C_0$. The common $C^\beta$ bound for $\sum_iw_i\lambda_i$, together with smoothness of the square-root map on $[c_0^2,C_0^2]$, gives
$\sup_w\|g_w\|_{C^\beta}<\infty$.
The approximation bound therefore supplies $\bm\theta_w^+$ with
\begin{equation}
\sup_w\|s_{\bm\theta_w^+}-g_w\|_\infty
\le\varepsilon_h:=A h^\beta.
\label{eq:app-uniform-spline-witness}
\end{equation}
Choose $h$ small enough that $\varepsilon_h<c_0/2$. The B-spline Gram matrix $\bm M_h$ is positive definite. Choose $R$ so that
\[
R^2\lambda_{\min}(\bm M_h)>|\mathcal B|(C_0+\varepsilon_h)^2.
\]
Since
$\bm\theta_w^{+\top}\bm M_h\bm\theta_w^+=\|s_{\bm\theta_w^+}\|_2^2
\le|\mathcal B|(C_0+\varepsilon_h)^2$,
all witnesses lie strictly inside this ball and satisfy $s_{\bm\theta_w^+}>2\delta$. This choice of $R$ also gives
$R^2\lambda_{\min}(\bm M_h)>\max_i\int\lambda_i$.

Define the unconstrained function-level optimum
\[
Q_w^\circ:=-\int_{\mathcal B}g_w^2+
2\int_{\mathcal B}g_w^2\log g_w
\]
and, for $g,u>0$, the loss
$D_g(u):=u^2-g^2-2g^2\log(u/g)$.
Then
\[
D_g(u)-(u-g)^2
=2g^2\{u/g-1-\log(u/g)\}\ge0.
\]
On the interval between $g_w$ and its witness, the second derivative of this loss is at most $2+8C_0^2/c_0^2$. Thus \eqref{eq:app-uniform-spline-witness} gives a constant $A_1$, independent of $w,h,R$, such that
\begin{equation}
Q_w^\circ-Q_w(\bm\theta_w^+)\le A_1h^{2\beta}.
\label{eq:app-uniform-witness-loss}
\end{equation}

Every $\bm\theta\in \mathcal X_\delta$ has $\min_{\mathcal B}|s_{\bm\theta}|\le\delta$: for sign-changing splines this follows from connectedness, and for boundary strips it follows by definition. Lemma~\ref{lem:app-local-polynomial-loss} therefore applies uniformly to all such coefficients. To include the clipped criterion, put $b=\delta^{c_{\rm tr}}<\delta$ and $u_b=|s_{\bm\theta}|\vee b$. For every $u\ge0$,
\[
\operatorname{dist}(u\vee b,[c_0,C_0])
\ge(1-b/c_0)\operatorname{dist}(u,[c_0,C_0]).
\]
Using $\bar\lambda=0$ and the definition of $Q_{\delta,w}$,
\begin{align*}
Q_{\delta,w}(\bm\theta)
&=Q_w^\circ-\int_{\mathcal B}D_{g_w}(u_b)
  +\int_{\mathcal B}(u_b^2-s_{\bm\theta}^2)\\
&\le Q_w^\circ-(1-b/c_0)^2\kappa h^H+|\mathcal B|b^2,
\qquad \bm\theta\in \mathcal X_\delta.
\end{align*}
Let $\bm\theta^\star(w)$ maximize $Q_w$ on $D_\delta$. This maximizer exists and is unique by compactness and strong concavity. Since it is at least as good as the witness,
\begin{equation}
\inf_{w\in\Delta_{n-1}}G_\delta(w)
\ge(1-b/c_0)^2\kappa h^H-A_1h^{2\beta}-|\mathcal B|b^2.
\label{eq:app-full-simplex-gap}
\end{equation}
Because $2\beta>H$, take $h$ small enough that $A_1h^{2\beta}\le\kappa h^H/8$. With this $h$ fixed, choose $c_{\rm tr}$ large enough that $|\mathcal B|b^2\le\kappa h^H/8$. As $b<\delta<c_0/4$, \eqref{eq:app-full-simplex-gap} is at least $\kappa h^H/4>0$.

The strict gap rules out $\inf_{\bm y}s_{\bm\theta^\star(w)}(\bm y)=\delta$: otherwise $u\bm\theta^\star(w)$ lies in the boundary strip for $u<1$ close to one, and its clipped criterion tends to $Q_w(\bm\theta^\star(w))$. The chosen radius rules out $\|\bm\theta^\star(w)\|_2=R$, because
\[
\left.\frac{d}{du}Q_w(u\bm\theta^\star(w))\right|_{u=1}
=2\int\lambda_{0,w}-2\bm\theta^{\star\top}\bm M_h\bm\theta^\star<0
\]
at that boundary. A small inward move remains above $\delta$ and increases $Q_w$, a contradiction. Thus every mode is internal. Continuity of $Q_w$ on the common compact core, uniqueness of its maximizer, and compactness of $\Delta_{n-1}$ imply a continuous mode map and a common $\rho>0$ with
$\overline{\mathbb B}(\bm\theta^\star(w),2\rho)\subset D_\delta$.
Together with \eqref{eq:app-full-simplex-gap}, this proves Assumption~\ref{ass:app-uniform-family}.
\end{proof}

\begin{appproposition}[Existence of admissible radius for Theorem~\ref{thm:main-dominance}]
\label{prop:app-admissible-radius}
Let $\bm\theta^\star$ maximize $Q$ on the nonempty set $D_\delta$. Suppose
$\bm M\succ0$, $\bm\Omega\succeq0$, $\bar\lambda\ge0$,
$\lambda_0$ is bounded on $\mathcal B$, and $G_\delta>0$, and suppose $\bm\theta^\star\in\operatorname{int}(D_\delta)$.
Define
\begin{equation}
L_Q(\delta)
:=
2\|\bm M\|_{\mathrm{op}}
+
2\bar\lambda\|\bm\Omega\|_{\mathrm{op}}
+
\frac{2B_{\max}^2\|\lambda_0\|_\infty|\mathcal B|}{\delta^2}
\;<\;
\infty,
\label{eq:app-LQ}
\end{equation}
and
\begin{equation}
r^\star(\delta) := \sqrt{\frac{2G_\delta}{L_Q(\delta)}}\;>\;0.
\label{eq:app-rstar}
\end{equation}
Then for every $r\in(0,r^\star(\delta))$,
\[
v^\star(r) > 0
\quad\text{and}\quad
\omega(r) < G_\delta.
\]
Under Assumption~\ref{ass:app-uniform-family}, set
\begin{equation}
\begin{split}
\overline L_Q(\delta)&:=2\|\bm M\|_{\mathrm{op}}+2\bar\lambda\|\bm\Omega\|_{\mathrm{op}}
+\frac{2B_{\max}^2}{\delta^2}\max_i\int_{\mathcal B}\lambda_i,\\
\underline r^\star(\delta)&:=\min\left\{\rho,\sqrt{\frac{2\underline G_\delta}{\overline L_Q(\delta)}}\right\}>0.
\end{split}
\label{eq:app-common-radius}
\end{equation}
For each fixed $r_0\in(0,\underline r^\star(\delta))$, one has $\inf_wv_w^\star(r_0)=|\mathbb B(\bm0,r_0)|>0$ and $\sup_w\omega_w(r_0)\le\overline L_Q(\delta)r_0^2/2<\underline G_\delta$.
\end{appproposition}

\begin{proof}
For $\bm\theta\in D_\delta$ we have $s_{\bm\theta}\ge\delta$ on $\mathcal B$, so the Hessian formula in the proof of Lemma~\ref{lem:app-strong-concavity} gives
\begin{equation}
\label{eq:app-hess-bound}
\|\nabla^2Q(\bm\theta)\|_{\mathrm{op}}\le2\|\bm M\|_{\mathrm{op}}+2\bar\lambda\|\bm\Omega\|_{\mathrm{op}}+\frac{2}{\delta^2}\int_{\mathcal B}\lambda_0(\bm y)\|\bm B(\bm y)\|_2^2\,d\bm y\le L_Q(\delta),
\end{equation}
where the last step uses $\|\bm B(\bm y)\|_2\le B_{\max}<\infty$ (Lemma~\ref{lem:app-bounded-basis}) and the boundedness of $\lambda_0$. 
Since $\bm\theta^\star$ maximizes the differentiable function $Q$ over $D_\delta$ and lies in $\operatorname{int}(D_\delta)$, the first-order condition gives $\nabla Q(\bm\theta^\star)=\bm 0$; by Lemma~\ref{lem:app-strong-concavity}, $\bm\theta^\star$ is moreover the unique maximizer, although only stationarity is used below. For $\bm\theta\in D_\delta\cap\mathbb B(\bm\theta^\star,r)$, the segment from $\bm\theta^\star$ to $\bm\theta$ lies in the convex set $D_\delta$, so a second-order Taylor expansion gives
\[
Q(\bm\theta^\star)-Q(\bm\theta)\le\frac{L_Q(\delta)}{2}\|\bm\theta-\bm\theta^\star\|_2^2\le\frac{L_Q(\delta)}{2}r^2 .
\]
Hence $\omega(r)\le L_Q(\delta)r^2/2<G_\delta$ for all $r<r^\star(\delta)$. 
Since $\bm\theta^\star\in\operatorname{int}(D_\delta)$, the set $D_\delta\cap\mathbb B(\bm\theta^\star,r)$ contains a ball of positive radius for every $r>0$, so $v^\star(r)>0$.

For the uniform assertion, fix $w\in\mathcal W$. By Assumption~\ref{ass:app-uniform-family}, $\bm\theta^\star(w)$ maximizes $Q_w$ over $D_\delta$ and $\mathbb B(\bm\theta^\star(w),r_0)\subset D_\delta$ for $r_0\le\rho$; hence $\bm\theta^\star(w)$ is an interior maximizer and $v^\star_w(r_0)=|\mathbb B(\bm 0,r_0)|$. Since $\int_{\mathcal B}\lambda_{0,w}\|\bm B\|_2^2\le B_{\max}^2\sum_iw_i\int_{\mathcal B}\lambda_i\le B_{\max}^2\max_i\int_{\mathcal B}\lambda_i$, the first inequality in \eqref{eq:app-hess-bound} yields $\|\nabla^2Q_w(\bm\theta)\|_{\mathrm{op}}\le\overline L_Q(\delta)$ on $D_\delta$ for every $w$, and the same Taylor argument gives $\omega_w(r_0)\le\overline L_Q(\delta)r_0^2/2$. For $r_0<\underline r^\star(\delta)$ this is smaller than $\underline G_\delta$, which proves the final assertion.
\end{proof}

\subsection{Proofs of the Main Text Theorems~\ref{thm:main-chamber-gap}, \ref{thm:main-dominance}, and~\ref{thm:main-laplace}}

\begin{proof}[Proof of Theorem~\ref{thm:main-chamber-gap}]
By Assumption~\ref{ass:app-uniform-family}, $Q_w$ has a unique maximizer $\bm\theta^\star(w)$ on $D_\delta$. The common gap and $Q_w\le Q_{\delta,w}$ give
\[
\sup_{\bm\theta\in \mathcal X_\delta}Q_w(\bm\theta)
\le Q_w(\bm\theta^\star(w))-\underline G_\delta.
\]
Since the truncated and original criteria coincide on $D_\delta\cup(-D_\delta)$, sign symmetry and strong concavity on $D_\delta$ show that the only maximizers over $\Theta^{\pm,\mathrm{sc}}$ are $\pm\bm\theta^\star(w)$. This proves the population statement for every $w\in\mathcal W$.
 
By Proposition~\ref{lem:app-truncated-full-sign-ulln},
$\overline\Xi_{\delta}=\sup_{\mathcal A}\Xi_{\delta,a}\xrightarrow{p}0$. 
The preceding population statements apply to every $w\in\mathcal W$.
By definitions~\eqref{eq:app-truncated-sc-gap}--\eqref{eq:app-G-delta},
\begin{align*}
\sup_{\bm\theta\in \mathcal X_\delta} Q_{\delta}(\bm\theta)
\le
Q(\bm\theta^\star)-G_{\delta}.
\end{align*}
Since $\bm\theta^\star\in D_\delta$, the truncation is inactive at
$\bm\theta^\star$, so $\widetilde J_{\delta}(\bm\theta^\star)=\widetilde J(\bm\theta^\star)$
and $Q_{\delta}(\bm\theta^\star)=Q(\bm\theta^\star)$.
Hence on the event $\{2\Xi_{\delta,a}<G_{\delta}\}$,
\begin{align*}
\sup_{\bm\theta\in \mathcal X_\delta}\widetilde J(\bm\theta)
&\le
\sup_{\bm\theta\in \mathcal X_\delta}\widetilde J_{\delta}(\bm\theta)
\\
&\le
\sup_{\bm\theta\in \mathcal X_\delta}Q_{\delta}(\bm\theta)
+\Xi_{\delta,a}
\\
&\le
Q(\bm\theta^\star)-G_{\delta}+\Xi_{\delta,a}
\\
&<
Q(\bm\theta^\star)-\Xi_{\delta,a}
\\
&\le
\widetilde J_{\delta}(\bm\theta^\star)
\\
&=
\widetilde J(\bm\theta^\star)
\le
\sup_{\bm\theta\in D_\delta\cup(-D_\delta)}\widetilde J(\bm\theta).
\end{align*}
On the common event $\{2\overline\Xi_{\delta}<\underline G_\delta\}$, this chain of inequalities holds for every $(\bm\nu,\eta)\in\mathcal A$. Its probability tends to one by Assumption~\ref{ass:app-uniform-family} and Proposition~\ref{lem:app-truncated-full-sign-ulln}, proving \eqref{eq:main-sample-gap} simultaneously over $\mathcal A$.
 
Since $\Theta^{\pm,\mathrm{sc}}=\bigl(D_\delta\cup(-D_\delta)\bigr)\cup \mathcal X_\delta$, the
separation \eqref{eq:main-sample-gap} forces every maximizer of $\widetilde J$ over
$\Theta^{\pm,\mathrm{sc}}$ to lie in $D_\delta\cup(-D_\delta)$ with probability tending
to one. Moreover, $D_\delta\cup(-D_\delta)$ is compact and $\widetilde J$ is continuous
on it, because $|s_{\bm\theta}(\cdot)|\ge\delta$ there; hence the supremum is attained
On these two chambers the truncated and original criteria coincide. Write $\mu:=\lambda_{\min}(\bm M)>0$. Strong concavity of $Q_w$ and optimality of $\widehat{\bm\theta}_\delta$ imply
\[
\mu\|\widehat{\bm\theta}_\delta-\bm\theta^\star(w)\|_2^2
\le Q_w(\bm\theta^\star(w))-Q_w(\widehat{\bm\theta}_\delta)
\le2\overline\Xi_{\delta}.
\]
Thus the positive-mode error tends to zero uniformly over $\mathcal A$. The empirical criterion is itself strongly concave on $D_\delta$, so its constrained mode is unique. On the common separation event its full-sign maximizers are exactly $\pm\widehat{\bm\theta}_\delta$, by sign symmetry. The same bound proves \eqref{eq:main-chamber-gap-argmax}, as well as the final assertion of the theorem.
\end{proof}


\begin{proof}[Proof of Theorem~\ref{thm:main-dominance}]
Let $C\subseteq \mathcal X_\delta$. For every $\bm\theta\in C$,
\begin{align*}
\widetilde J(\bm\theta)
&\le
\widetilde J_{\delta}(\bm\theta)
\\
&\le
Q_{\delta}(\bm\theta)+\Xi_{\delta,a}
\\
&\le
Q(\bm\theta^\star)-G_{\delta}+\Xi_{\delta,a}.
\end{align*}
Thus
\begin{align}
\widehat Z_a(C)
\le
|C|
\exp\!\left\{
 a\bigl(Q(\bm\theta^\star)-G_{\delta}+\Xi_{\delta,a}\bigr)
\right\}.
\label{eq:app-sample-Z-upper}
\end{align}
For $\bm\theta\in D_\delta\cap \mathbb B(\bm \theta^\star,r)$, \eqref{eq:app-omega} implies
$Q(\bm\theta)\ge Q(\bm\theta^\star)-\omega(r)$. Proposition~\ref{lem:app-truncated-full-sign-ulln} gives
\[
\widetilde J(\bm\theta)
=
\widetilde J_{\delta}(\bm\theta)
\ge
Q(\bm\theta)-\Xi_{\delta,a}
\ge
Q(\bm\theta^\star)-\omega(r)-\Xi_{\delta,a}.
\]
Therefore
\begin{align}
\widehat Z_a(D_\delta)
\ge
v^\star(r)
\exp\!\left\{
 a\bigl(Q(\bm\theta^\star)-\omega(r)-\Xi_{\delta,a}\bigr)
\right\}.
\label{eq:app-sample-Z-lower}
\end{align}
Combining \eqref{eq:app-sample-Z-upper} and \eqref{eq:app-sample-Z-lower} yields
\eqref{eq:main-sample-ratio-bound}.
To make the convergence uniform, choose a fixed $r_0\in(0,\rho)$ sufficiently small that $\overline L_Q(\delta)r_0^2/2\le\underline G_\delta/4$. Proposition~\ref{prop:app-admissible-radius} gives $v_w^\star(r_0)=v_{r_0}:=|\mathbb B(\bm0,r_0)|$ and $\omega_w(r_0)\le\underline G_\delta/4$ for every $w$. On the event $\{\overline\Xi_{\delta}\le\underline G_\delta/8\}$, the preceding pathwise bound yields, simultaneously over $\mathcal A$,
\[
\frac{\widehat Z_a(\mathcal X_\delta)}{\widehat Z_a(D_\delta)}
\le\frac{|\bar{\mathbb B}(\bm0,R)|}{v_{r_0}}
\exp\{-a\underline G_\delta/2\}
\le\frac{|\bar{\mathbb B}(\bm0,R)|}{v_{r_0}}
\exp\{-\underline a\underline G_\delta/2\}.
\]
The event has probability tending to one and the last bound tends to zero, proving \eqref{eq:main-sample-ratio-goes-zero}. This argument allows $a$ to be random after substituting fitted responsibilities.
The negative-core statement follows from \eqref{eq:app-sign-symmetry}.
\end{proof}

\begin{proof}[Proof of Theorem~\ref{thm:main-laplace}]
All bounds below hold on common events whose probability tends to one. Write $\mu:=\lambda_{\min}(\bm M)>0$ and $\widehat{\bm\theta}:=\widehat{\bm\theta}_\delta$. Theorem~\ref{thm:main-chamber-gap} and Assumption~\ref{ass:app-uniform-family} imply
\begin{equation}
\mathbb B(\widehat{\bm\theta},\rho)\subset D_\delta,
\qquad \nabla J(\widehat{\bm\theta})=\bm0,
\quad\text{for all }(\bm\nu,\eta)\in\mathcal A.
\label{eq:app-empirical-interiority}
\end{equation}
Let $V_\nu:=\sum_i\nu_iN_i(\mathcal B)$ and $\Lambda_{\max}:=\max_i\int\lambda_i$. For every admissible pair,
\begin{equation}
\frac{V_\nu}{a}
=\sum_iw_i\frac{N_i(\mathcal B)}{T_i}
\le\max_i\frac{N_i(\mathcal B)}{T_i}
\xrightarrow{p}\Lambda_{\max}.
\label{eq:app-weighted-count-bound}
\end{equation}
Thus the last maximum is bounded by a fixed $K_{\Lambda}>\Lambda_{\max}$ with probability tending to one. Differentiating $J$ on $D_\delta$ gives
\[
-\nabla^2J(\bm\theta)
=2a\bm M+\eta\bm\Omega
+2\sum_i\nu_i\int\frac{\bm B(\bm y)\bm B(\bm y)^\top}{s_{\bm\theta}(\bm y)^2}\,dN_i(\bm y).
\]
Since $s_{\bm\theta}\ge\delta$, $\eta\le C_\eta$, and $a\ge\underline a\ge1$ for sufficiently large exposures, there is a fixed $C_H<\infty$ such that
\begin{equation}
2a\mu\bm I_d\preceq-\nabla^2J(\bm\theta)\preceq C_Ha\bm I_d
\quad(\bm\theta\in D_\delta),
\label{eq:app-uniform-hessian}
\end{equation}
uniformly over $\mathcal A$. One may take $C_H=2\|\bm M\|_{\mathrm{op}}+C_\eta\|\bm\Omega\|_{\mathrm{op}}+2K_{\Lambda}B_{\max}^2/\delta^2$. Moreover,
\begin{align*}
|D^3J(\bm\theta)[\bm u,\bm u,\bm u]|
&=\left|4\sum_i\nu_i\int\frac{(\bm B(\bm y)^\top\bm u)^3}{s_{\bm\theta}(\bm y)^3}\,dN_i(\bm y)\right|\\
&\le\frac{4B_{\max}^3}{\delta^3}V_\nu\|\bm u\|_2^3
\le C_3a\|\bm u\|_2^3
\end{align*}
for a fixed $C_3$.

Set $\bm K_a:=-\bm H(\widehat{\bm\theta})\succ0$ and use the affine coordinates $\bm z=\bm K_a^{1/2}(\bm\theta-\widehat{\bm\theta})$. Define
\[
E_a:=\bm K_a^{1/2}(D_\delta-\widehat{\bm\theta}),\qquad
\tilde q_a(\bm z):=
\exp\{J(\widehat{\bm\theta}+\bm K_a^{-1/2}\bm z)-J(\widehat{\bm\theta})\}\mathbf1_{E_a}(\bm z).
\]
Equation~\eqref{eq:app-empirical-interiority} implies that $E_a$ contains the ball of radius $\rho\sqrt{2a\mu}$. Hence every fixed ball in $\bm z$ is contained in all the $E_a$ for sufficiently large exposures. Taylor's theorem, the third-derivative bound, and \eqref{eq:app-uniform-hessian} yield
\begin{equation}
\log \tilde q_a(\bm z)=-\frac12\|\bm z\|_2^2+R_a(\bm z),
\qquad
|R_a(\bm z)|\le\frac{C_4}{\sqrt a}\|\bm z\|_2^3
\quad(\bm z\in E_a),
\label{eq:app-uniform-taylor}
\end{equation}
where $C_4$ is independent of the admissible pair. The segment used in Taylor's theorem stays in the convex set $D_\delta$.

For the tails, strong concavity and stationarity at $\widehat{\bm\theta}$ give
\[
J(\widehat{\bm\theta}+\bm K_a^{-1/2}\bm z)-J(\widehat{\bm\theta})
\le-a\mu\|\bm K_a^{-1/2}\bm z\|_2^2
\le-\frac{\mu}{C_H}\|\bm z\|_2^2
\quad(\bm z\in E_a).
\]
Since $\tilde q_a=0$ outside $E_a$, this is a common integrable Gaussian envelope on $\mathbb R^d$. Let $\tilde q^{\mathrm L}(\bm z):=\exp(-\|\bm z\|_2^2/2)$ in this proof only. For every fixed $L<\infty$, \eqref{eq:app-uniform-taylor} bounds the integral of $|\tilde q_a-\tilde q^{\mathrm L}|$ on $\{\|\bm z\|_2\le L\}$ by $C_L/\sqrt{\underline a}$, uniformly over $\mathcal A$. The integral outside this ball is bounded by
\[
\int_{\|\bm z\|_2>L}
\left\{e^{-(\mu/C_H)\|\bm z\|_2^2}+e^{-\|\bm z\|_2^2/2}\right\}\,d\bm z,
\]
which tends to zero as $L\to\infty$, independently of the exposures and the weights. Taking $T_i\to\infty$ for every $i$ first, and then $L\to\infty$, proves
\[
\sup_{\mathcal A}\int_{\mathbb R^d}|\tilde q_a-\tilde q^{\mathrm L}|\,d\bm z\xrightarrow{p}0.
\]
The normalizing constants therefore converge uniformly to $(2\pi)^{d/2}>0$. More explicitly, since $\int\tilde q^{\mathrm L}=(2\pi)^{d/2}$, we have
\[
\left\|\frac{\tilde q_a}{\int\tilde q_a}-\frac{\tilde q^{\mathrm L}}{(2\pi)^{d/2}}\right\|_1
\le\frac{2}{(2\pi)^{d/2}}\|\tilde q_a-\tilde q^{\mathrm L}\|_1.
\]
These normalized densities are the affine images of $q^*_{a,+}$ and $q_a^{\mathrm L}$. Invariance of total variation under invertible affine transformations proves \eqref{eq:main-laplace-tv}.

Finally, for any data-dependent set of VI iterations whose pairs lie in $\mathcal A$, the supremum of their TV errors is bounded by the supremum in \eqref{eq:main-laplace-tv}. Thus no union bound over iterations, or convergence of the iteration sequence, is required.
\end{proof}

\begin{appcorollary}[Simultaneous validity at all active VI updates]
\label{cor:app-active-iterations}
Suppose the conditions of Proposition~\ref{prop:app-full-simplex} hold for each mark, $n$ and the truncation level $K$ are fixed, and $T_i\to\infty$ for every subject. Use the fixed basis and constants supplied by that proposition. Initialize the responsibilities as probability vectors and the variance parameters by the conjugate update in Algorithm~\ref{alg:cavi}, and consider its well-defined coefficient updates over $D_\delta$.

Let $t$ index these updates, with $(\bm\nu_k^{(t)},\eta_{km}^{(t)})$ denoting the current inputs, and define
\[
\mathcal K^{(t)}:=\left\{k:\sum_i\nu_{ik}^{(t)}>1\right\},
\qquad T_{\min}:=\min_iT_i.
\]
Take $\mathcal W=\Delta_{n-1}$ and $\underline a=T_{\min}$ in \eqref{eq:app-admissible-pairs}. Then, on every sample,
\begin{equation}
(\bm\nu_k^{(t)},\eta_{km}^{(t)})\in\mathcal A
\quad\text{for all }t,\ k\in\mathcal K^{(t)},\ m.
\label{eq:app-active-membership}
\end{equation}
Consequently the separation event of Theorem~\ref{thm:main-chamber-gap} holds simultaneously at these updates with probability tending to one. With all quantities evaluated at the corresponding current block,
\begin{align}
\sup_{\substack{t,\ k\in\mathcal K^{(t)},\ m}}
\|\widehat{\bm\theta}_{\delta,km}^{(t)}-\bm\theta_m^\star(w_k^{(t)})\|_2
&\xrightarrow{p}0,\nonumber\\
\sup_{\substack{t,\ k\in\mathcal K^{(t)},\ m}}
\frac{\widehat Z_{a,km}^{(t)}(\mathcal X_\delta)}
{\widehat Z_{a,km}^{(t)}(D_\delta)}
&\xrightarrow{p}0,\nonumber\\
\sup_{\substack{t,\ k\in\mathcal K^{(t)},\ m}}
\|q_{a,+,km}^{*,(t)}-q_{a,km}^{\mathrm L,(t)}\|_{\mathrm{TV}}
&\xrightarrow{p}0.
\label{eq:app-active-uniform-results}
\end{align}
The supremum over no active updates is defined as zero. The set of active blocks may change with $t$; a fixed block need not remain active.
\end{appcorollary}

\begin{proof}
Responsibility normalization gives $0\le\nu_{ik}^{(t)}\le1$ at every update. The conjugate variance update and $\bm\Omega\succeq0$ give
\[
0<\eta_{km}^{(t)}
=\frac{a_0+d/2}{b_0+\frac12\{
\operatorname{Tr}(\bm\Omega\bm\Sigma_{km})
+\bm\mu_{km}^\top\bm\Omega\bm\mu_{km}\}}
\le C_\eta,
\]
where the mean and covariance are those used in the latest variance update. The same inequality holds at initialization. For an active block,
\[
a_k^{(t)}=\sum_i\nu_{ik}^{(t)}T_i
\ge T_{\min}\sum_i\nu_{ik}^{(t)}>T_{\min}.
\]
Finally, $w_{ik}^{(t)}=\nu_{ik}^{(t)}T_i/a_k^{(t)}$ is nonnegative and sums to one over $i$. This proves \eqref{eq:app-active-membership} at every iteration, without a probabilistic argument about the VI trajectory.

Proposition~\ref{prop:app-full-simplex} verifies the common population assumption on this entire simplex. Since $\underline a=T_{\min}\to\infty$, Proposition~\ref{prop:app-gamma} supplies uniform penalty control with $\bar\lambda=0$, and Lemma~\ref{lem:app-weighted-ulln} supplies uniform empirical control. Thus Theorems~\ref{thm:main-chamber-gap}--\ref{thm:main-laplace} apply over $\mathcal A$. For each mark and each of their errors $E$, pathwise inclusion gives
\[
\sup_{t,\ k\in\mathcal K^{(t)}}E(\bm\nu_k^{(t)},\eta_{km}^{(t)})
\le\sup_{(\bm\nu,\eta)\in\mathcal A}E(\bm\nu,\eta).
\]
The right-hand side tends to zero in probability. The finitely many marks can be handled simultaneously, which proves \eqref{eq:app-active-uniform-results} and the separation assertion. The argument includes data-dependent stopping times and requires neither convergence of the VI iterates nor a union bound over iterations. It concerns each current coefficient update; the coefficient-box implementation is covered under the mode-inclusion condition stated after Algorithm~\ref{alg:cavi}.
\end{proof}




\section{INTENSITY ESTIMATION EXPERIMENT}
\label{app:intensity}

\subsection{Purpose and Design}

The model of Section~\ref{sec:model} combines three components: a Dirichlet process mixture
over subjects, a mark layer, and a squared-link intensity whose basis
coefficients are updated by the constrained Laplace step of Section~\ref{subsec:laplace}. 
This experiment removes the first two components. A single cluster containing
a single subject, observed without marks, reduces the algorithm to exactly the
coefficient block. 
What is measured is therefore the quality
of the intensity estimator alone, which is the part of the method that can be
compared directly against the point-process intensity-estimation literature.
Each dataset is one unmarked point pattern, and each method is asked for one
intensity surface together with a nominal $95\%$ credible interval.

All patterns are generated on the rectangular domain $\mathcal{B}$.
Six true intensity surfaces are used, shown in the top row of
Figure~\ref{fig:synthetic results}(a); they span multimodal, banded, cross-shaped,
star-shaped and highly irregular geometries, normalized to unit integral and
scaled by the target count.
Every surface is observed at four expected counts
$n_{\mathrm{exp}} \in \{500,\,2\,000,\,10\,000,\,50\,000\}$, and each of the $24$ resulting
configurations is replicated $100$ times with independent seeds, giving
$2\,400$ datasets and $14\,400$ fits.
Every method returns its estimate on one shared evaluation grid and every fit is
pinned to a single core.

The intensity estimation part of the proposed model (DPM-MPPP) is fitted with concentration parameter $\alpha=1$, and hyperparameters $a_0=1$ and $b_0=0.005$. We use a tensor-product cubic B-spline basis with 10 interior knots on each axis, yielding $d=14\times 14=196$ basis functions, and the precision matrix $\Omega$ is set to $\bm\Omega = \bm{\mathrm{P}} + c_{\bm\Omega}\bm{I}_d$ where P is taken to be the corresponding first-order Bayesian $P$-spline penalty matrix. For result summaries, each subject $i$ is assigned to the cluster with the highest responsibility $\nu_{ik}$.

\subsection{Baselines}

Five baselines are used, chosen to span the three link functions 
and both approximate and exact posterior inference. Each is run from its
authors' released implementation.

\textbf{INLA} \citep{Simpson2016GoingOffGrid, Rue2009INLA} is the log-link
LGCP under the SPDE representation, which replaces gridding with a
finite-element basis on a triangulated mesh. Fits are obtained with the
\textsf{R-INLA} package \citep{Rue2009INLA}. We use a fixed mesh across all
$n_{\mathrm{exp}}$, with $1\,890$ nodes, and penalized-complexity priors on the range and
marginal standard deviation. The log link leaves the likelihood integral
doubly intractable, so inference has to proceed on an approximating basis; the
SPDE construction is a considerable improvement on a regular grid, but a
triangulation still has to be chosen, and mesh construction becomes
progressively harder as the geometry of the domain grows irregular or its
dimension increases \citep{Simpson2016GoingOffGrid, Teng2017Bayesian}. The posterior
is then a nested Laplace approximation with numerical integration over the
hyperparameters.

\textbf{LGCP} \citep{Taylor2013lgcp} fits the same log-link LGCP by
Metropolis-adjusted Langevin (MALA) Markov chain Monte Carlo (MCMC) sampling
on discretized lattices. Fits are obtained with the \textsf{lgcp} package
\citep{Taylor2013lgcp}. It is the only baseline method whose posterior is
neither a variational nor a Laplace approximation. The
price is that the lattice is a genuine discretization whose cell width must be
fixed in advance 
and that sampling makes it far more
expensive.

\textbf{VBPP} \citep{lloyd2015variational} is the squared-link variational
method, with a sparse Gaussian process
prior on $400$ inducing points arranged on a $20 \times 20$ grid. The squared
link makes the integrated intensity available in closed form, so neither
gridding nor thinning is required; this is the property our method inherits.
Its two limitations are the ones discussed in Section~\ref{subsec:squared-link}. Closed-form
coordinate updates for the basis coefficients are unavailable, so the full
variational mean and covariance are optimized numerically, which becomes
expensive and unstable as the latent dimension grows. And the squared link
leaves a global sign ambiguity together with nodal lines where the linear
predictor crosses zero, so the optimizer is sensitive to initialization and
susceptible to undesirable local modes \citep{John2018LargeScaleCP}.

\textbf{SGCP} \citep{Donner2018Sigmoidal} uses the sigmoid link
$\lambda(y) = \lambda^{\star}\sigma\{g(y)\}$ with Pólya--Gamma augmentation
and a mean-field variational posterior, also with $400$ inducing points. The
sigmoid link avoids the intractable integral through thinning-based
augmentation \citep{Adams2009Tractable}, but inference must then be carried
out over the augmented latent thinned points as well, and their number grows
with the size of the domain and with the intensity level.

\textbf{STVB} \citep{Aglietti2019Structured} shares SGCP's model, prior and
augmentation and differs only in that its variational posterior is structured
rather than mean-field, retaining the dependence between the latent count, the
maximum intensity and the thinned locations. The pair therefore isolates the
effect of the posterior factorization at a fixed model. The richer posterior
inherits SGCP's augmentation cost and adds to it, making STVB the slowest
method among baselines; and, as in the
reference implementation and in the original paper, its kernel hyperparameters
are held fixed rather than learned.


\subsection{Metrics}

All quantities below are evaluated on one shared lattice of
$G = 188 \times 200 = 37\,600$ equally spaced grid points over $\mathcal{B}$,
with cell area $A$, so that every integral is the midpoint sum
$\int g \approx \sum_i g(y_i)A$ and every method is scored on identical
quadrature. Write $\lambda$ for the true intensity and $\hat{\lambda}$ for an
estimate, both evaluated on $\{y_i\}_{i=1}^{G}$, and let $R = 100$ index
replicates.

The primary metric is the relative integrated squared error
\begin{equation*}
\mathrm{RISE}
= \frac{\sum_{i}\{\hat{\lambda}(y_i)-\lambda(y_i)\}^{2}A}
       {\sum_{i}\lambda(y_i)^{2}A},
\end{equation*}
which is scale-free and therefore comparable across surfaces and sample sizes.
One value $\mathrm{RISE}^{(r)}$ is obtained per replicate, and we report the
mean over replicates together with
$\mathrm{SD}_{\mathrm{RISE}}=\{(R-1)^{-1}\sum_{r}(\mathrm{RISE}^{(r)}
-\overline{\mathrm{RISE}})^{2}\}^{1/2}$,
which measures how stable the accuracy of a method is from one dataset to the
next.

For coverage, let $[L^{(r)}_{j},U^{(r)}_{j}]$ be the nominal $95\%$ interval
returned at grid point $y_j$ in replicate $r$, and let
$I^{(r)}_{j}=\mathbf{1}\{L^{(r)}_{j}\le\lambda(y_j)\le U^{(r)}_{j}\}$. The
coverage attained at a single grid point is the proportion of replicates in
which its interval contains the truth,
\begin{equation*}
c_j = \frac{1}{R}\sum_{r=1}^{R} I^{(r)}_{j},
\qquad j = 1,\dots,G,
\end{equation*}
and the reported pointwise coverage is the average of these over the surface,
$\mathrm{Cov}=G^{-1}\sum_{j}c_j$. Write
$p_r = G^{-1}\sum_{j} I^{(r)}_{j}$ for the fraction of the grid covered in
replicate $r$. 

\begin{equation*}
\mathrm{SD}_{\mathrm{Cov}}
=\Big\{\tfrac{1}{R-1}\sum_{r=1}^{R}(p_r-\mathrm{Cov})^{2}\Big\}^{1/2},
\qquad
\mathrm{MSE}_{\mathrm{Cov}}
=\frac{1}{G}\sum_{j=1}^{G}\big(c_j-0.95\big)^{2}.
\end{equation*}

The two answer different questions. $\mathrm{SD}_{\mathrm{Cov}}$ varies over
replicates at a fixed surface and therefore reports how reproducible a
method's calibration is from one dataset to the next.
In contrast, $\mathrm{MSE}_{\mathrm{Cov}}$ evaluates the spatial variation of coverage across different locations on the surface. It measures how consistently the method achieves accurate coverage: a method might achieve a 0.95 average by over-covering some regions while under-covering others. While the simple average cannot detect this spatial imbalance, $\mathrm{MSE}_{\mathrm{Cov}}$ directly penalizes any local deviation from the nominal level.

Finally, CI width is the mean width of the pointwise intervals,
$\,(RG)^{-1}\sum_{r}\sum_{j}(U^{(r)}_{j}-L^{(r)}_{j})$, which is what makes
coverage interpretable.
Computation time is single-core wall clock for one fit, including prediction.

\subsection{Results}

\begin{table}[t!]
\centering
\caption{Full intensity-estimation results for Scenario A. RISE and SD$_{\mathrm{RISE}}$ are $\times 10^{2}$; Cov is the empirical pointwise coverage of the nominal $95\%$ band, SD$_{\mathrm{Cov}}$ its standard deviation across replicates and MSE$_{\mathrm{Cov}}$ its mean squared deviation from $0.95$ across grid points; CI width is the mean width of the pointwise intervals; Time is the mean total wall-clock time per fit and SD$_{\mathrm{Time}}$ its standard deviation. $n_{\mathrm{exp}}$ is the expected number of observations per pattern. Lowest RISE in bold.}
\label{tab:rise-A}
\scriptsize
\setlength{\tabcolsep}{2pt}
\renewcommand{\arraystretch}{1.05}
\begin{tabular*}{\textwidth}{@{\extracolsep{\fill}}l rrrrrrrr rrrrrrrr@{}}
\toprule
 & \multicolumn{8}{c}{$n_{\mathrm{exp}}=500$} & \multicolumn{8}{c}{$n_{\mathrm{exp}}=2{,}000$} \\
\cmidrule(lr){2-9}\cmidrule(lr){10-17}
Method & RISE & SD$_{\mathrm{RISE}}$ & Cov & SD$_{\mathrm{Cov}}$ & MSE$_{\mathrm{Cov}}$ & \shortstack{CI\\width} & \shortstack{Time\\(s)} & SD$_{\mathrm{Time}}$ & RISE & SD$_{\mathrm{RISE}}$ & Cov & SD$_{\mathrm{Cov}}$ & MSE$_{\mathrm{Cov}}$ & \shortstack{CI\\width} & \shortstack{Time\\(s)} & SD$_{\mathrm{Time}}$ \\
\midrule
INLA & 15.41 & 1.79 & 0.819 & 0.074 & 0.058 & 0.295 & 41.7 & 17.3 & 5.64 & 0.71 & 0.921 & 0.024 & 0.018 & 0.886 & 42.5 & 17.1 \\
LGCP & 16.74 & 1.32 & 0.971 & 0.012 & 0.003 & 0.421 & 166.2 & 31.5 & 9.31 & 0.74 & 0.969 & 0.009 & 0.002 & 1.252 & 163.3 & 28.1 \\
SGCP & 17.24 & 2.32 & 0.577 & 0.131 & 0.226 & 0.163 & 43.7 & 28.5 & 5.07 & 0.81 & 0.769 & 0.037 & 0.093 & 0.518 & 49.4 & 24.5 \\
STVB & 18.29 & 1.28 & 0.834 & 0.025 & 0.065 & 0.286 & 312.8 & 25.7 & 10.16 & 1.47 & 0.899 & 0.016 & 0.041 & 1.056 & 484.5 & 54.3 \\
VBPP & 37.44 & 15.55 & 0.817 & 0.065 & 0.027 & 0.389 & 59.0 & 21.5 & 22.84 & 3.77 & 0.714 & 0.097 & 0.124 & 1.169 & 47.6 & 16.2 \\
Ours & \textbf{13.44} & 1.57 & 0.924 & 0.023 & 0.009 & 0.304 & 33.9 & 1.3 & \textbf{4.75} & 0.64 & 0.931 & 0.019 & 0.005 & 0.752 & 32.0 & 0.7 \\
\midrule[\heavyrulewidth]
 & \multicolumn{8}{c}{$n_{\mathrm{exp}}=10{,}000$} & \multicolumn{8}{c}{$n_{\mathrm{exp}}=50{,}000$} \\
\cmidrule(lr){2-9}\cmidrule(lr){10-17}
Method & RISE & SD$_{\mathrm{RISE}}$ & Cov & SD$_{\mathrm{Cov}}$ & MSE$_{\mathrm{Cov}}$ & \shortstack{CI\\width} & \shortstack{Time\\(s)} & SD$_{\mathrm{Time}}$ & RISE & SD$_{\mathrm{RISE}}$ & Cov & SD$_{\mathrm{Cov}}$ & MSE$_{\mathrm{Cov}}$ & \shortstack{CI\\width} & \shortstack{Time\\(s)} & SD$_{\mathrm{Time}}$ \\
\midrule
INLA & 1.84 & 0.18 & 0.914 & 0.014 & 0.028 & 2.517 & 44.5 & 17.2 & 0.66 & 0.05 & 0.891 & 0.011 & 0.043 & 6.932 & 47.6 & 19.2 \\
LGCP & 5.32 & 0.30 & 0.956 & 0.011 & 0.007 & 4.939 & 154.3 & 30.5 & 5.88 & 0.79 & 0.967 & 0.009 & 0.014 & 32.609 & 156.5 & 27.2 \\
SGCP & 2.12 & 0.16 & 0.682 & 0.021 & 0.134 & 1.341 & 88.1 & 32.3 & 2.07 & 0.12 & 0.459 & 0.015 & 0.355 & 3.628 & 138.0 & 41.3 \\
STVB & 8.30 & 1.78 & 0.880 & 0.020 & 0.063 & 4.746 & 1586.4 & 250.6 & 18.94 & 0.19 & 0.549 & 0.002 & 0.400 & 21.965 & 8300.0 & 1959.9 \\
VBPP & 16.56 & 9.66 & 0.979 & 0.033 & 0.002 & 15.148 & 160.3 & 46.0 & 27.03 & 15.89 & 0.996 & 0.009 & 0.002 & 83.794 & 889.7 & 306.4 \\
Ours & \textbf{1.33} & 0.14 & 0.912 & 0.018 & 0.009 & 1.925 & 18.3 & 0.7 & \textbf{0.47} & 0.03 & 0.827 & 0.012 & 0.053 & 4.582 & 46.4 & 1.6 \\
\bottomrule
\end{tabular*}
\end{table}

\begin{table}[t!]
\centering
\caption{Full intensity-estimation results for Scenario B. RISE and SD$_{\mathrm{RISE}}$ are $\times 10^{2}$; Cov is the empirical pointwise coverage of the nominal $95\%$ band, SD$_{\mathrm{Cov}}$ its standard deviation across replicates and MSE$_{\mathrm{Cov}}$ its mean squared deviation from $0.95$ across grid points; CI width is the mean width of the pointwise intervals; Time is the mean total wall-clock time per fit and SD$_{\mathrm{Time}}$ its standard deviation. $n_{\mathrm{exp}}$ is the expected number of observations per pattern. Lowest RISE in bold.}
\label{tab:rise-B}
\scriptsize
\setlength{\tabcolsep}{2pt}
\renewcommand{\arraystretch}{1.05}
\begin{tabular*}{\textwidth}{@{\extracolsep{\fill}}l rrrrrrrr rrrrrrrr@{}}
\toprule
 & \multicolumn{8}{c}{$n_{\mathrm{exp}}=500$} & \multicolumn{8}{c}{$n_{\mathrm{exp}}=2{,}000$} \\
\cmidrule(lr){2-9}\cmidrule(lr){10-17}
Method & RISE & SD$_{\mathrm{RISE}}$ & Cov & SD$_{\mathrm{Cov}}$ & MSE$_{\mathrm{Cov}}$ & \shortstack{CI\\width} & \shortstack{Time\\(s)} & SD$_{\mathrm{Time}}$ & RISE & SD$_{\mathrm{RISE}}$ & Cov & SD$_{\mathrm{Cov}}$ & MSE$_{\mathrm{Cov}}$ & \shortstack{CI\\width} & \shortstack{Time\\(s)} & SD$_{\mathrm{Time}}$ \\
\midrule
INLA & 3.50 & 1.04 & 0.742 & 0.029 & 0.200 & 0.246 & 41.4 & 17.2 & 1.35 & 0.32 & 0.765 & 0.020 & 0.185 & 0.613 & 40.0 & 16.5 \\
LGCP & 11.14 & 2.69 & 0.794 & 0.018 & 0.165 & 0.480 & 159.0 & 28.4 & 6.97 & 1.09 & 0.756 & 0.016 & 0.177 & 1.215 & 166.5 & 32.9 \\
SGCP & 5.88 & 1.36 & 0.522 & 0.054 & 0.283 & 0.087 & 74.4 & 20.7 & 2.48 & 0.57 & 0.507 & 0.040 & 0.291 & 0.207 & 154.6 & 38.3 \\
STVB & 10.41 & 3.01 & 0.711 & 0.019 & 0.230 & 0.314 & 313.1 & 24.9 & 6.06 & 3.04 & 0.769 & 0.012 & 0.185 & 1.154 & 475.8 & 51.3 \\
VBPP & 9.29 & 6.69 & 0.939 & 0.021 & 0.030 & 0.418 & 30.7 & 9.0 & 2.10 & 4.44 & 0.963 & 0.067 & 0.017 & 0.702 & 55.8 & 14.4 \\
Ours & \textbf{3.33} & 0.97 & 0.949 & 0.013 & 0.035 & 0.237 & 23.5 & 1.3 & \textbf{1.21} & 0.29 & 0.956 & 0.012 & 0.024 & 0.540 & 19.4 & 0.8 \\
\midrule[\heavyrulewidth]
 & \multicolumn{8}{c}{$n_{\mathrm{exp}}=10{,}000$} & \multicolumn{8}{c}{$n_{\mathrm{exp}}=50{,}000$} \\
\cmidrule(lr){2-9}\cmidrule(lr){10-17}
Method & RISE & SD$_{\mathrm{RISE}}$ & Cov & SD$_{\mathrm{Cov}}$ & MSE$_{\mathrm{Cov}}$ & \shortstack{CI\\width} & \shortstack{Time\\(s)} & SD$_{\mathrm{Time}}$ & RISE & SD$_{\mathrm{RISE}}$ & Cov & SD$_{\mathrm{Cov}}$ & MSE$_{\mathrm{Cov}}$ & \shortstack{CI\\width} & \shortstack{Time\\(s)} & SD$_{\mathrm{Time}}$ \\
\midrule
INLA & 0.44 & 0.08 & 0.775 & 0.014 & 0.180 & 1.741 & 41.1 & 17.2 & 0.14 & 0.02 & 0.788 & 0.011 & 0.170 & 4.884 & 46.5 & 17.8 \\
LGCP & 3.13 & 0.43 & 0.775 & 0.018 & 0.165 & 4.498 & 158.8 & 30.4 & 12.89 & 5.37 & 0.873 & 0.019 & 0.106 & 41.826 & 164.7 & 34.4 \\
SGCP & 1.73 & 0.22 & 0.439 & 0.029 & 0.368 & 0.588 & 188.1 & 45.3 & 2.58 & 0.25 & 0.350 & 0.015 & 0.463 & 2.320 & 413.0 & 100.7 \\
STVB & 6.95 & 2.83 & 0.768 & 0.014 & 0.181 & 5.178 & 1587.1 & 279.4 & 25.20 & 0.19 & 0.373 & 0.001 & 0.563 & 21.300 & 8567.7 & 2073.0 \\
VBPP & 5.73 & 5.44 & 0.943 & 0.058 & 0.036 & 6.018 & 131.1 & 40.2 & 7.53 & 2.05 & 0.942 & 0.008 & 0.051 & 49.137 & 592.0 & 183.4 \\
Ours & \textbf{0.35} & 0.07 & 0.960 & 0.011 & 0.014 & 1.367 & 24.6 & 0.6 & \textbf{0.09} & 0.02 & 0.960 & 0.011 & 0.008 & 3.330 & 62.1 & 2.4 \\
\bottomrule
\end{tabular*}
\end{table}

\begin{table}[t!]
\centering
\caption{Full intensity-estimation results for Scenario C. RISE and SD$_{\mathrm{RISE}}$ are $\times 10^{2}$; Cov is the empirical pointwise coverage of the nominal $95\%$ band, SD$_{\mathrm{Cov}}$ its standard deviation across replicates and MSE$_{\mathrm{Cov}}$ its mean squared deviation from $0.95$ across grid points; CI width is the mean width of the pointwise intervals; Time is the mean total wall-clock time per fit and SD$_{\mathrm{Time}}$ its standard deviation. $n_{\mathrm{exp}}$ is the expected number of observations per pattern. Lowest RISE in bold.}
\label{tab:rise-C}
\scriptsize
\setlength{\tabcolsep}{2pt}
\renewcommand{\arraystretch}{1.05}
\begin{tabular*}{\textwidth}{@{\extracolsep{\fill}}l rrrrrrrr rrrrrrrr@{}}
\toprule
 & \multicolumn{8}{c}{$n_{\mathrm{exp}}=500$} & \multicolumn{8}{c}{$n_{\mathrm{exp}}=2{,}000$} \\
\cmidrule(lr){2-9}\cmidrule(lr){10-17}
Method & RISE & SD$_{\mathrm{RISE}}$ & Cov & SD$_{\mathrm{Cov}}$ & MSE$_{\mathrm{Cov}}$ & \shortstack{CI\\width} & \shortstack{Time\\(s)} & SD$_{\mathrm{Time}}$ & RISE & SD$_{\mathrm{RISE}}$ & Cov & SD$_{\mathrm{Cov}}$ & MSE$_{\mathrm{Cov}}$ & \shortstack{CI\\width} & \shortstack{Time\\(s)} & SD$_{\mathrm{Time}}$ \\
\midrule
INLA & 3.44 & 0.94 & 0.836 & 0.060 & 0.082 & 0.231 & 40.4 & 17.4 & 1.33 & 0.30 & 0.898 & 0.045 & 0.039 & 0.583 & 40.3 & 15.8 \\
LGCP & 12.31 & 3.35 & 0.867 & 0.026 & 0.085 & 0.479 & 166.8 & 37.7 & 7.72 & 1.03 & 0.810 & 0.027 & 0.107 & 1.175 & 171.1 & 39.4 \\
SGCP & 5.62 & 0.83 & 0.369 & 0.064 & 0.402 & 0.055 & 95.9 & 26.1 & 2.91 & 0.93 & 0.465 & 0.114 & 0.273 & 0.175 & 152.4 & 47.1 \\
STVB & 9.89 & 3.34 & 0.680 & 0.019 & 0.250 & 0.310 & 317.5 & 24.8 & 5.79 & 4.24 & 0.761 & 0.020 & 0.179 & 1.113 & 497.8 & 53.4 \\
VBPP & 6.35 & 3.49 & 0.984 & 0.018 & 0.003 & 0.367 & 35.4 & 8.1 & 2.40 & 4.07 & 0.969 & 0.041 & 0.004 & 0.686 & 51.7 & 13.7 \\
Ours & \textbf{2.92} & 0.76 & 0.983 & 0.011 & 0.004 & 0.224 & 20.2 & 1.0 & \textbf{1.08} & 0.28 & 0.980 & 0.012 & 0.002 & 0.506 & 18.9 & 0.8 \\
\midrule[\heavyrulewidth]
 & \multicolumn{8}{c}{$n_{\mathrm{exp}}=10{,}000$} & \multicolumn{8}{c}{$n_{\mathrm{exp}}=50{,}000$} \\
\cmidrule(lr){2-9}\cmidrule(lr){10-17}
Method & RISE & SD$_{\mathrm{RISE}}$ & Cov & SD$_{\mathrm{Cov}}$ & MSE$_{\mathrm{Cov}}$ & \shortstack{CI\\width} & \shortstack{Time\\(s)} & SD$_{\mathrm{Time}}$ & RISE & SD$_{\mathrm{RISE}}$ & Cov & SD$_{\mathrm{Cov}}$ & MSE$_{\mathrm{Cov}}$ & \shortstack{CI\\width} & \shortstack{Time\\(s)} & SD$_{\mathrm{Time}}$ \\
\midrule
INLA & 0.41 & 0.07 & 0.917 & 0.033 & 0.029 & 1.655 & 40.0 & 17.3 & 0.14 & 0.02 & 0.923 & 0.029 & 0.027 & 4.576 & 44.9 & 18.3 \\
LGCP & 4.05 & 0.87 & 0.872 & 0.021 & 0.079 & 5.394 & 163.7 & 36.5 & 11.21 & 4.48 & 0.950 & 0.017 & 0.036 & 40.238 & 166.5 & 37.0 \\
SGCP & 1.46 & 0.19 & 0.495 & 0.036 & 0.291 & 0.534 & 151.1 & 39.9 & 2.18 & 0.12 & 0.360 & 0.021 & 0.440 & 1.643 & 357.2 & 95.3 \\
STVB & 5.54 & 3.19 & 0.762 & 0.014 & 0.180 & 5.054 & 1705.4 & 318.9 & 25.29 & 0.20 & 0.272 & 0.001 & 0.654 & 21.118 & 9085.2 & 2268.9 \\
VBPP & 4.66 & 4.30 & 0.988 & 0.043 & 0.002 & 7.006 & 154.0 & 63.1 & 5.51 & 2.48 & 0.998 & 0.001 & 0.004 & 41.998 & 638.0 & 225.5 \\
Ours & \textbf{0.31} & 0.06 & 0.972 & 0.013 & 0.001 & 1.281 & 26.9 & 0.7 & \textbf{0.09} & 0.02 & 0.955 & 0.013 & 0.002 & 3.112 & 66.6 & 3.1 \\
\bottomrule
\end{tabular*}
\end{table}

\begin{table}[t!]
\centering
\caption{Full intensity-estimation results for Scenario D. RISE and SD$_{\mathrm{RISE}}$ are $\times 10^{2}$; Cov is the empirical pointwise coverage of the nominal $95\%$ band, SD$_{\mathrm{Cov}}$ its standard deviation across replicates and MSE$_{\mathrm{Cov}}$ its mean squared deviation from $0.95$ across grid points; CI width is the mean width of the pointwise intervals; Time is the mean total wall-clock time per fit and SD$_{\mathrm{Time}}$ its standard deviation. $n_{\mathrm{exp}}$ is the expected number of observations per pattern. Lowest RISE in bold.}
\label{tab:rise-D}
\scriptsize
\setlength{\tabcolsep}{2pt}
\renewcommand{\arraystretch}{1.05}
\begin{tabular*}{\textwidth}{@{\extracolsep{\fill}}l rrrrrrrr rrrrrrrr@{}}
\toprule
 & \multicolumn{8}{c}{$n_{\mathrm{exp}}=500$} & \multicolumn{8}{c}{$n_{\mathrm{exp}}=2{,}000$} \\
\cmidrule(lr){2-9}\cmidrule(lr){10-17}
Method & RISE & SD$_{\mathrm{RISE}}$ & Cov & SD$_{\mathrm{Cov}}$ & MSE$_{\mathrm{Cov}}$ & \shortstack{CI\\width} & \shortstack{Time\\(s)} & SD$_{\mathrm{Time}}$ & RISE & SD$_{\mathrm{RISE}}$ & Cov & SD$_{\mathrm{Cov}}$ & MSE$_{\mathrm{Cov}}$ & \shortstack{CI\\width} & \shortstack{Time\\(s)} & SD$_{\mathrm{Time}}$ \\
\midrule
INLA & \textbf{4.55} & 0.95 & 0.909 & 0.033 & 0.039 & 0.254 & 25.4 & 11.3 & \textbf{1.76} & 0.32 & 0.934 & 0.021 & 0.025 & 0.639 & 25.0 & 10.1 \\
LGCP & 8.83 & 1.97 & 0.952 & 0.020 & 0.029 & 0.451 & 144.0 & 22.4 & 5.71 & 0.91 & 0.947 & 0.011 & 0.030 & 1.381 & 154.5 & 27.1 \\
SGCP & 5.33 & 1.29 & 0.669 & 0.054 & 0.165 & 0.093 & 56.6 & 15.3 & 3.08 & 0.47 & 0.584 & 0.045 & 0.222 & 0.213 & 62.6 & 17.0 \\
STVB & 8.75 & 2.03 & 0.885 & 0.025 & 0.067 & 0.279 & 300.4 & 19.1 & 4.33 & 1.73 & 0.929 & 0.014 & 0.041 & 1.071 & 454.1 & 37.7 \\
VBPP & 7.98 & 3.57 & 0.971 & 0.044 & 0.002 & 0.366 & 27.2 & 6.5 & 4.48 & 3.30 & 0.978 & 0.026 & 0.002 & 1.195 & 46.8 & 10.6 \\
Ours & 4.73 & 1.00 & 0.983 & 0.013 & 0.001 & 0.258 & 30.6 & 1.2 & 1.78 & 0.33 & 0.976 & 0.012 & 0.001 & 0.604 & 23.1 & 0.8 \\
\midrule[\heavyrulewidth]
 & \multicolumn{8}{c}{$n_{\mathrm{exp}}=10{,}000$} & \multicolumn{8}{c}{$n_{\mathrm{exp}}=50{,}000$} \\
\cmidrule(lr){2-9}\cmidrule(lr){10-17}
Method & RISE & SD$_{\mathrm{RISE}}$ & Cov & SD$_{\mathrm{Cov}}$ & MSE$_{\mathrm{Cov}}$ & \shortstack{CI\\width} & \shortstack{Time\\(s)} & SD$_{\mathrm{Time}}$ & RISE & SD$_{\mathrm{RISE}}$ & Cov & SD$_{\mathrm{Cov}}$ & MSE$_{\mathrm{Cov}}$ & \shortstack{CI\\width} & \shortstack{Time\\(s)} & SD$_{\mathrm{Time}}$ \\
\midrule
INLA & 0.55 & 0.09 & 0.942 & 0.016 & 0.022 & 1.827 & 27.1 & 11.0 & 0.18 & 0.02 & 0.936 & 0.014 & 0.029 & 5.128 & 30.7 & 12.2 \\
LGCP & 3.43 & 0.49 & 0.940 & 0.013 & 0.029 & 4.929 & 145.8 & 24.4 & 2.45 & 0.76 & 0.967 & 0.025 & 0.014 & 23.587 & 150.6 & 31.8 \\
SGCP & 1.59 & 0.20 & 0.540 & 0.036 & 0.258 & 0.644 & 91.3 & 26.0 & 1.46 & 0.14 & 0.439 & 0.025 & 0.351 & 2.276 & 223.8 & 56.6 \\
STVB & 4.65 & 2.10 & 0.943 & 0.019 & 0.026 & 4.912 & 1402.2 & 55.4 & 18.56 & 0.46 & 0.528 & 0.004 & 0.419 & 22.201 & 7335.9 & 722.6 \\
VBPP & 11.37 & 5.37 & 0.978 & 0.031 & 0.003 & 9.743 & 185.6 & 74.5 & 10.20 & 8.79 & 0.999 & 0.002 & 0.002 & 58.538 & 1181.5 & 451.1 \\
Ours & \textbf{0.50} & 0.09 & 0.972 & 0.011 & 0.001 & 1.557 & 24.8 & 0.9 & \textbf{0.13} & 0.02 & 0.960 & 0.014 & 0.001 & 3.837 & 59.8 & 1.9 \\
\bottomrule
\end{tabular*}
\end{table}

\begin{table}[t!]
\centering
\caption{Full intensity-estimation results for Scenario E. RISE and SD$_{\mathrm{RISE}}$ are $\times 10^{2}$; Cov is the empirical pointwise coverage of the nominal $95\%$ band, SD$_{\mathrm{Cov}}$ its standard deviation across replicates and MSE$_{\mathrm{Cov}}$ its mean squared deviation from $0.95$ across grid points; CI width is the mean width of the pointwise intervals; Time is the mean total wall-clock time per fit and SD$_{\mathrm{Time}}$ its standard deviation. $n_{\mathrm{exp}}$ is the expected number of observations per pattern. Lowest RISE in bold.}
\label{tab:rise-E}
\scriptsize
\setlength{\tabcolsep}{2pt}
\renewcommand{\arraystretch}{1.05}
\begin{tabular*}{\textwidth}{@{\extracolsep{\fill}}l rrrrrrrr rrrrrrrr@{}}
\toprule
 & \multicolumn{8}{c}{$n_{\mathrm{exp}}=500$} & \multicolumn{8}{c}{$n_{\mathrm{exp}}=2{,}000$} \\
\cmidrule(lr){2-9}\cmidrule(lr){10-17}
Method & RISE & SD$_{\mathrm{RISE}}$ & Cov & SD$_{\mathrm{Cov}}$ & MSE$_{\mathrm{Cov}}$ & \shortstack{CI\\width} & \shortstack{Time\\(s)} & SD$_{\mathrm{Time}}$ & RISE & SD$_{\mathrm{RISE}}$ & Cov & SD$_{\mathrm{Cov}}$ & MSE$_{\mathrm{Cov}}$ & \shortstack{CI\\width} & \shortstack{Time\\(s)} & SD$_{\mathrm{Time}}$ \\
\midrule
INLA & 2.66 & 0.88 & 0.607 & 0.043 & 0.311 & 0.226 & 40.2 & 17.2 & 1.05 & 0.25 & 0.619 & 0.035 & 0.303 & 0.543 & 40.0 & 16.7 \\
LGCP & 15.51 & 3.64 & 0.593 & 0.023 & 0.323 & 0.455 & 177.3 & 41.7 & 7.75 & 1.07 & 0.544 & 0.018 & 0.358 & 1.192 & 179.6 & 39.8 \\
SGCP & 6.54 & 2.13 & 0.342 & 0.040 & 0.457 & 0.066 & 154.5 & 46.1 & 1.47 & 0.37 & 0.288 & 0.031 & 0.508 & 0.153 & 644.9 & 236.1 \\
STVB & 17.42 & 5.85 & 0.434 & 0.016 & 0.475 & 0.313 & 301.4 & 20.4 & 12.15 & 5.61 & 0.488 & 0.012 & 0.429 & 1.052 & 465.2 & 40.3 \\
VBPP & 8.92 & 7.58 & 0.864 & 0.049 & 0.088 & 0.484 & 31.6 & 7.8 & 1.93 & 1.70 & 0.848 & 0.032 & 0.098 & 0.518 & 49.9 & 15.0 \\
Ours & \textbf{2.31} & 0.91 & 0.847 & 0.009 & 0.130 & 0.208 & 14.5 & 0.6 & \textbf{0.76} & 0.23 & 0.879 & 0.010 & 0.098 & 0.451 & 15.3 & 0.7 \\
\midrule[\heavyrulewidth]
 & \multicolumn{8}{c}{$n_{\mathrm{exp}}=10{,}000$} & \multicolumn{8}{c}{$n_{\mathrm{exp}}=50{,}000$} \\
\cmidrule(lr){2-9}\cmidrule(lr){10-17}
Method & RISE & SD$_{\mathrm{RISE}}$ & Cov & SD$_{\mathrm{Cov}}$ & MSE$_{\mathrm{Cov}}$ & \shortstack{CI\\width} & \shortstack{Time\\(s)} & SD$_{\mathrm{Time}}$ & RISE & SD$_{\mathrm{RISE}}$ & Cov & SD$_{\mathrm{Cov}}$ & MSE$_{\mathrm{Cov}}$ & \shortstack{CI\\width} & \shortstack{Time\\(s)} & SD$_{\mathrm{Time}}$ \\
\midrule
INLA & 0.34 & 0.07 & 0.624 & 0.026 & 0.303 & 1.504 & 42.4 & 17.8 & 0.14 & 0.02 & 0.636 & 0.020 & 0.292 & 4.122 & 46.3 & 18.0 \\
LGCP & 4.17 & 0.61 & 0.574 & 0.015 & 0.349 & 5.461 & 174.6 & 37.6 & 8.08 & 1.54 & 0.614 & 0.008 & 0.335 & 39.167 & 161.0 & 30.5 \\
SGCP & 1.75 & 0.23 & 0.227 & 0.024 & 0.598 & 0.457 & 278.3 & 85.7 & 2.47 & 0.14 & 0.177 & 0.011 & 0.664 & 1.718 & 600.3 & 181.2 \\
STVB & 11.24 & 5.80 & 0.444 & 0.014 & 0.463 & 4.894 & 1527.9 & 199.6 & 41.01 & 0.21 & 0.134 & 0.001 & 0.780 & 18.637 & 8077.6 & 1961.2 \\
VBPP & 2.54 & 1.56 & 0.854 & 0.039 & 0.084 & 4.435 & 176.8 & 75.6 & 3.37 & 2.48 & 0.855 & 0.012 & 0.126 & 34.786 & 666.2 & 265.8 \\
Ours & \textbf{0.20} & 0.05 & 0.910 & 0.008 & 0.066 & 1.101 & 22.7 & 0.7 & \textbf{0.05} & 0.01 & 0.928 & 0.010 & 0.044 & 2.588 & 61.2 & 1.3 \\
\bottomrule
\end{tabular*}
\end{table}

\begin{table}[t!]
\centering
\caption{Full intensity-estimation results for Scenario F. RISE and SD$_{\mathrm{RISE}}$ are $\times 10^{2}$; Cov is the empirical pointwise coverage of the nominal $95\%$ band, SD$_{\mathrm{Cov}}$ its standard deviation across replicates and MSE$_{\mathrm{Cov}}$ its mean squared deviation from $0.95$ across grid points; CI width is the mean width of the pointwise intervals; Time is the mean total wall-clock time per fit and SD$_{\mathrm{Time}}$ its standard deviation. $n_{\mathrm{exp}}$ is the expected number of observations per pattern. Lowest RISE in bold.}
\label{tab:rise-F}
\scriptsize
\setlength{\tabcolsep}{2pt}
\renewcommand{\arraystretch}{1.05}
\begin{tabular*}{\textwidth}{@{\extracolsep{\fill}}l rrrrrrrr rrrrrrrr@{}}
\toprule
 & \multicolumn{8}{c}{$n_{\mathrm{exp}}=500$} & \multicolumn{8}{c}{$n_{\mathrm{exp}}=2{,}000$} \\
\cmidrule(lr){2-9}\cmidrule(lr){10-17}
Method & RISE & SD$_{\mathrm{RISE}}$ & Cov & SD$_{\mathrm{Cov}}$ & MSE$_{\mathrm{Cov}}$ & \shortstack{CI\\width} & \shortstack{Time\\(s)} & SD$_{\mathrm{Time}}$ & RISE & SD$_{\mathrm{RISE}}$ & Cov & SD$_{\mathrm{Cov}}$ & MSE$_{\mathrm{Cov}}$ & \shortstack{CI\\width} & \shortstack{Time\\(s)} & SD$_{\mathrm{Time}}$ \\
\midrule
INLA & 14.17 & 1.83 & 0.827 & 0.039 & 0.092 & 0.327 & 37.9 & 17.2 & 6.03 & 0.60 & 0.863 & 0.023 & 0.070 & 0.881 & 38.2 & 17.3 \\
LGCP & 17.55 & 1.80 & 0.935 & 0.019 & 0.030 & 0.460 & 167.6 & 45.3 & 12.32 & 1.84 & 0.911 & 0.024 & 0.036 & 1.491 & 173.4 & 45.0 \\
SGCP & 20.12 & 2.34 & 0.593 & 0.060 & 0.245 & 0.168 & 32.4 & 13.3 & 7.66 & 1.01 & 0.723 & 0.033 & 0.119 & 0.504 & 92.2 & 31.1 \\
STVB & 21.79 & 1.67 & 0.778 & 0.021 & 0.137 & 0.295 & 316.1 & 18.6 & 13.67 & 2.49 & 0.827 & 0.017 & 0.102 & 1.081 & 491.5 & 59.1 \\
VBPP & 32.56 & 14.10 & 0.845 & 0.044 & 0.026 & 0.358 & 44.8 & 16.9 & 18.32 & 3.01 & 0.689 & 0.058 & 0.164 & 0.917 & 57.8 & 15.5 \\
Ours & \textbf{11.81} & 1.66 & 0.932 & 0.024 & 0.008 & 0.304 & 31.3 & 1.2 & \textbf{4.27} & 0.53 & 0.940 & 0.016 & 0.004 & 0.738 & 30.6 & 0.8 \\
\midrule[\heavyrulewidth]
 & \multicolumn{8}{c}{$n_{\mathrm{exp}}=10{,}000$} & \multicolumn{8}{c}{$n_{\mathrm{exp}}=50{,}000$} \\
\cmidrule(lr){2-9}\cmidrule(lr){10-17}
Method & RISE & SD$_{\mathrm{RISE}}$ & Cov & SD$_{\mathrm{Cov}}$ & MSE$_{\mathrm{Cov}}$ & \shortstack{CI\\width} & \shortstack{Time\\(s)} & SD$_{\mathrm{Time}}$ & RISE & SD$_{\mathrm{RISE}}$ & Cov & SD$_{\mathrm{Cov}}$ & MSE$_{\mathrm{Cov}}$ & \shortstack{CI\\width} & \shortstack{Time\\(s)} & SD$_{\mathrm{Time}}$ \\
\midrule
INLA & 2.30 & 0.17 & 0.868 & 0.014 & 0.067 & 2.557 & 39.3 & 18.0 & 1.08 & 0.06 & 0.847 & 0.013 & 0.079 & 7.060 & 42.5 & 17.8 \\
LGCP & 7.42 & 0.60 & 0.940 & 0.008 & 0.029 & 6.487 & 162.8 & 40.3 & 8.56 & 1.30 & 0.960 & 0.009 & 0.025 & 38.057 & 166.5 & 39.7 \\
SGCP & 3.51 & 0.23 & 0.670 & 0.021 & 0.143 & 1.432 & 211.3 & 65.8 & 3.79 & 0.13 & 0.462 & 0.014 & 0.345 & 3.954 & 292.0 & 80.0 \\
STVB & 11.77 & 2.64 & 0.820 & 0.013 & 0.118 & 4.978 & 1615.6 & 281.0 & 26.79 & 0.12 & 0.522 & 0.003 & 0.424 & 21.791 & 8046.6 & 2131.5 \\
VBPP & 22.14 & 6.12 & 0.954 & 0.049 & 0.006 & 13.020 & 182.8 & 72.6 & 25.88 & 5.11 & 0.980 & 0.012 & 0.010 & 60.964 & 823.0 & 373.0 \\
Ours & \textbf{1.07} & 0.12 & 0.948 & 0.015 & 0.001 & 1.888 & 21.4 & 0.8 & \textbf{0.22} & 0.03 & 0.957 & 0.013 & 0.001 & 4.473 & 45.7 & 1.8 \\
\bottomrule
\end{tabular*}
\end{table}

Figure~\ref{fig:synthetic results}(a) summarizes all $24$ configurations. 
Tables~\ref{tab:rise-A}--\ref{tab:rise-F} give the
underlying numbers for each scenario, together with the interval width and the
coverage MSE that the figure does not display.
 
The proposed method attains the lowest RISE in $22$ of the $24$
configurations. Both exceptions occur in Scenario~D, at
$n_{\mathrm{exp}}=500$ and $2\,000$, where INLA is lower by $0.19$ and
$0.02$ on the $\times 10^{2}$ scale. These gaps are $1.4$ and $0.5$ standard
errors of the difference, so neither is distinguishable from replicate noise,
and the ordering reverses at the two larger sample sizes of the same scenario.
The margin widens with $n_{\mathrm{exp}}$: at $n_{\mathrm{exp}}=50\,000$ the
ratio of the best baseline's RISE to ours runs from $1.4$ to $4.9$ across
Scenarios~A--F. The two sigmoidal methods move in the opposite direction and
deteriorate at the largest sample size, STVB's mean RISE rising from $8.07$ at
$n_{\mathrm{exp}}=10\,000$ to $25.97$ at $50\,000$. 
 
Averaged over the $24$ configurations, pointwise coverage is $0.939$ for the
proposed method, against $0.919$ for VBPP, $0.852$ for LGCP, $0.825$ for
INLA, $0.666$ for STVB and $0.488$ for SGCP. The average on its own
understates the difference, because a method can hold a reasonable average
while over-covering some regions and under-covering others. The
coverage MSE separates them: $0.022$ for the proposed method against $0.038$
for VBPP and $0.106$ or worse for every other method, so ours is close to
nominal not only on average but across the surface. It reaches that
calibration with the narrowest intervals of any method that does not
undercover, with a mean band width of $1.51$ against $2.09$ for INLA and
$16.34$ for VBPP. SGCP's bands are narrower still at $0.96$, but its coverage
of $0.488$ shows that this reflects undercoverage rather than precision, which
is why width and coverage have to be read together.
Ultimately, while INLA narrowly edges out our method in two specific scenarios 
in terms of RISE, our approach proves to be conclusively superior 
when both estimation accuracy and reliable coverage are considered together.
 
Mean single-core runtime is $32.2$ seconds for the proposed method, against
$39.4$ for INLA, $163.1$ for LGCP, $193.6$ for SGCP, $263.3$ for
VBPP and $2\,648.7$ for STVB. It is also the most predictable of the six: the
standard deviation of the runtime is $1.1$ seconds for ours, against $16.4$
for INLA and $538.8$ for STVB, so the cost of a fit is essentially
independent of the dataset it is given. 

\section{CLUSTERING ACCURACY EXPERIMENT}
\label{app:clustering}

\subsection{Purpose and Design}

This experiment measures subject-level clustering accuracy against published
competitors for replicated point patterns. \citep{Yin2023MatrixClustering} does not model
event-level marks, so switching marks off puts the comparison on the clustering
task alone.

Six clusters are defined by the six surfaces shown in the top row of
Figure~\ref{fig:synthetic results}(a).
Every subject in a cluster shares that cluster's intensity, scaled by its own
expected count, and contributes one independent point pattern per replicate.
In the High-High (HH) configuration, the number of subjects per
cluster ranges from $20$ to $60$. Each subject's expected count is randomly
drawn from a cluster-specific interval, varying broadly from $80$ to $550$
depending on the cluster, and is subsequently held fixed. Because the number
of subjects and their expected counts remain constant across replicates, only
the realized counts and point locations are random, ensuring each subject
retains its identity and exposure throughout.

We cross a high and a low setting for
both the number of subjects and the expected counts, yielding four size
configurations (HH, LH, HL, LL) as summarized in Table~\ref{tab:regimes}, each
replicated $100$ times. For the low settings, the values are approximately
halved: the number of subjects reduces to $10$--$30$ per cluster, and the
expected counts decrease to ranges between $15$ and $250$. Across the four
regimes this spans $110$ to $220$ subjects and roughly $15\,700$ to $70\,700$
events per replicate. 



\begin{table}[t!]
\centering
\caption{Design of the four regimes. \emph{Subj.} is the number of subjects in the cluster, fixed across replicates, and \emph{Events} the range (minimum, maximum) of the number of events per subject over the $100$ replicates. The last row gives the mean total number of events per replicate. 
}
\label{tab:regimes}
\small
\setlength{\tabcolsep}{3pt}
\renewcommand{\arraystretch}{1.05}
\begin{tabular*}{\textwidth}{@{\extracolsep{\fill}}ll rc rc rc rc@{}}
\toprule
 & & \multicolumn{2}{c}{HH} & \multicolumn{2}{c}{LH} & \multicolumn{2}{c}{HL} & \multicolumn{2}{c}{LL} \\
\cmidrule(lr){3-4} \cmidrule(lr){5-6} \cmidrule(lr){7-8} \cmidrule(lr){9-10}
Cluster & Surface & Subj. & Events & Subj. & Events & Subj. & Events & Subj. & Events \\
\midrule
1 & Scenario A & 50 & (302, 605) & 25 & (315, 606) & 50 & (121, 294) & 25 & (124, 288) \\
2 & Scenario B & 60 & (342, 588) & 30 & (345, 554) & 60 & (160, 298) & 30 & (162, 306) \\
3 & Scenario C & 20 & (263, 556) & 10 & (267, 554) & 20 & (80, 226) & 10 & (79, 226) \\
4 & Scenario D & 20 & (64, 160) & 10 & (57, 148) & 20 & (7, 72) & 10 & (5, 62) \\
5 & Scenario E & 40 & (59, 141) & 20 & (60, 138) & 40 & (27, 88) & 20 & (29, 82) \\
6 & Scenario F & 30 & (123, 357) & 15 & (120, 344) & 30 & (14, 176) & 15 & (17, 165) \\
\midrule
\multicolumn{2}{@{}l}{All clusters} & 220 & (59, 605) & 110 & (57, 606) & 220 & (7, 298) & 110 & (5, 306) \\
\multicolumn{2}{@{}l}{Events per replicate} & \multicolumn{2}{c}{70{,}748.03} & \multicolumn{2}{c}{34{,}859.69} & \multicolumn{2}{c}{32{,}060.20} & \multicolumn{2}{c}{15{,}716.10} \\
\bottomrule
\end{tabular*}
\end{table}

For the proposed method (DPM-MPPP) we use the truncation level $K=30$, 
concentration parameter $\alpha=1$, and the same tensor B-spline basis as in
Appendix~\ref{app:intensity}, namely degree $3$ with $10$ interior knots per axis. 

\subsection{Baselines}


\textbf{MM-MPP} \citep{Yin2021RowClustering} models each subject as a multi-level
log-Gaussian Cox process and clusters subjects through a semiparametric
expectation-solution algorithm driven by functional principal component
analysis of the estimated second-order structure. Following the authors' approach, 
we select the number of clusters by BIC
over a range of candidate values. The code for MM-MPP is published, and released, for
one-dimensional temporal data only; the spatial case is named as future work
in the original paper. The two-dimensional extension is ours. It changes only
the index set: the kernel becomes a product of the two coordinates' Gaussian
kernels, the edge correction becomes the product of the corresponding
univariate corrections, and the quadrature weight becomes a cell area instead
of a cell width. Both product forms are exact rather than approximations,
since the reference kernel is already a product Gaussian and the domain is a
rectangle. 

Two structural characteristics of MM-MPP are directly relevant to our 
comparative analysis. First, the method relies on an expectation-solution step
that solves estimating equations, thus producing solely point estimates. 
Consequently, it lacks inherent uncertainty quantification for cluster assignments
and intensity surfaces—a limitation addressed by the variational posterior 
in our proposed method. Second, MM-MPP processes subjects via kernel-smoothed 
summaries on a lattice rather than utilizing a continuous-domain likelihood, 
necessitating the choice of a bandwidth parameter.

\textbf{MFM-MxN} \citep{Yin2023MatrixClustering} is a two-stage procedure. The first
stage estimates a per-subject log-intensity surface and discretizes it onto a coarse grid; the
second clusters the resulting matrices under a mixture of finite mixtures of
matrix normal distributions, with cluster labels summarized by Dahl's method.
For MFM-MxN only the second stage is published. The first stage is ours, and
reuses the INLA-LGCP construction of Appendix~\ref{app:intensity} unchanged
except for the prediction grid, which is the $24\times 25$ lattice the second
stage expects. We set the initial number of clusters to $30$, matching the
truncation level used by the proposed method. 
Computation time for this baseline is charged as the sum of both stages; the second stage
accounts for roughly $92\%$ of the total.

Being two-stage is the structural limitation. The posterior uncertainty in the
first-stage surfaces is discarded at the hand-off, so the clustering is
conditional on point estimates and no intensity uncertainty can propagate into
the cluster assignments, and conversely the cluster structure cannot inform
the intensity estimates by borrowing strength across subjects. The proposed
method infers both jointly in a single pass. The discretization is a second
cost: the continuous domain is reduced to a $24\times 25$ matrix per subject
before any clustering takes place.

\subsection{Metrics}

Let $z$ be the true partition of the $n$ subjects and $\hat{z}$ the estimated
one, and let $n_{ij}$ count subjects in true cluster $i$ and estimated cluster
$j$, with row and column sums $a_i$ and $b_j$. Writing
$P=\sum_{ij}\binom{n_{ij}}{2}$, $Q=\sum_{i}\binom{a_i}{2}$,
$R=\sum_{j}\binom{b_j}{2}$ and $E=QR/\binom{n}{2}$, the adjusted Rand index
\citep{Hubert1985ARI} is
\begin{equation*}
\mathrm{ARI}= \frac{P-E}{\tfrac12(Q+R)-E},
\end{equation*}
which is the Rand index corrected for the agreement expected by chance, so
that an arbitrary partition scores $0$ in expectation and a perfect one
scores $1$. Clustering purity and clustering accuracy are
\begin{equation*}
\mathrm{Purity}=\frac{1}{n}\sum_{j}\max_{i} n_{ij},
\qquad
\mathrm{Accuracy}=\frac{1}{n}\max_{\sigma}\sum_{i} n_{i\sigma(i)},
\end{equation*}
where $\sigma$ ranges over one-to-one maps from true to estimated clusters.
Purity and accuracy differ in exactly this restriction: purity can be raised
by splitting a true cluster into several estimated ones, since each fragment
contributes its own majority count, whereas accuracy cannot, because a
one-to-one map forbids assigning several estimated clusters to the same true
one. Reading the two together therefore separates over-segmentation from
genuine confusion: a method that splits true clusters without mixing them
keeps a high purity while its accuracy falls, and a method that merges
distinct clusters loses both.

Each metric is computed once per replicate, and we report its mean over the
$100$ replicates together with the standard deviation across them. 
Computation time is single-core wall clock, on the same
convention as Appendix~\ref{app:intensity}.

\subsection{Results}

\begin{table*}[t!]
\centering
\caption{Clustering performance over $100$ replicates in four settings, reported as mean (standard deviation).}
\label{tab:cluster-results}
\footnotesize
\setlength{\tabcolsep}{4pt}
\renewcommand{\arraystretch}{0.95}
\begin{tabular*}{\textwidth}{@{\extracolsep{\fill}}lrrrr@{}}
\toprule
Method & ARI & Purity & Accuracy & Time (s) \\
\midrule
\multicolumn{5}{@{}l}{\textit{HH}}\\
MM-MPP & 0.687 (0.190) & 0.799 (0.192) & 0.735 (0.153) & 464.8 (101.4) \\
MFM-MxN & 0.586 (0.128) & 0.734 (0.076) & 0.701 (0.080) & 7787.9 (1006.4) \\
Ours & \textbf{0.974} (0.018) & \textbf{1.000} (0.000) & \textbf{0.955} (0.031) & \textbf{130.1} (74.7) \\
\midrule
\multicolumn{5}{@{}l}{\textit{LH}}\\
MM-MPP & 0.684 (0.093) & 0.824 (0.096) & 0.732 (0.075) & 185.1 (35.5) \\
MFM-MxN & 0.404 (0.136) & 0.625 (0.083) & 0.581 (0.091) & 4003.6 (144.6) \\
Ours & \textbf{0.808} (0.068) & \textbf{0.867} (0.044) & \textbf{0.849} (0.052) & \textbf{50.9} (4.8) \\
\midrule
\multicolumn{5}{@{}l}{\textit{HL}}\\
MM-MPP & 0.671 (0.136) & 0.799 (0.139) & 0.723 (0.114) & 416.8 (65.0) \\
MFM-MxN & 0.468 (0.170) & 0.640 (0.093) & 0.631 (0.093) & 6817.4 (226.5) \\
Ours & \textbf{0.893} (0.061) & \textbf{0.993} (0.006) & \textbf{0.902} (0.055) & \textbf{70.9} (11.9) \\
\midrule
\multicolumn{5}{@{}l}{\textit{LL}}\\
MM-MPP & 0.629 (0.086) & 0.788 (0.068) & 0.707 (0.069) & 267.6 (51.3) \\
MFM-MxN & 0.293 (0.170) & 0.535 (0.094) & 0.532 (0.094) & 5121.9 (648.0) \\
Ours & \textbf{0.724} (0.129) & \textbf{0.859} (0.071) & \textbf{0.788} (0.096) & \textbf{37.2} (6.3) \\
\bottomrule
\end{tabular*}
\end{table*}

Table~\ref{tab:cluster-results} reports the three accuracy metrics and the
computation time for all four regimes along with their standard deviations.

The proposed method is first on every metric in every regime.  On
ARI the margin over the better of the two baselines is $0.287$ in HH, $0.124$
in LH, $0.222$ in HL and $0.095$ in LL. The margin is set by the number of
subjects rather than by the total number of events: it is large in the two
regimes with $220$ subjects and roughly half as large in the two with $110$,
even though HL has fewer events in total than LH. 
MM-MPP is the stronger baseline in all four regimes, with ARI between $0.629$
and $0.687$ and little variation across the four settings. MFM-MxN is
consistently weakest and degrades sharply as information is removed, from
$0.586$ in HH to $0.293$ in LL. 
The pattern is consistent with the structural limitations noted above: 
the first-stage surfaces are noisier with fewer subjects or lower exposure, and 
that noise is absorbed by over-merging rather than by a better partition.
 
The proposed method's purity exceeds its accuracy in every regime, by $0.045$
in HH, $0.018$ in LH, $0.091$ in HL and $0.071$ in LL. The gap is the
signature of mild over-segmentation: the method opens more components than
there are true clusters, but the components it opens are largely homogeneous subsets, thereby significantly reducing the merging of distinct true clusters. It widens both
as the number of subjects rises and as the per-subject exposure falls. 
 
Single-core computation time is $37$ to $130$ seconds for the proposed method,
$185$ to $465$ seconds for MM-MPP and $4\,004$ to $7\,788$ seconds for
MFM-MxN. The proposed method is between $3.6$ and $7.2$ times faster than
MM-MPP and between $60$ and $138$ times faster than MFM-MxN, and both margins
widen as the problem gets smaller. Its runtime is also stable across replicates,
with a standard deviation between $4.8$ and $11.9$ seconds, except in HH. 
The proposed method
is therefore the most accurate and the fastest in every regime.

\section{Marked Clustering Experiment}\label{app:marked}
 
\subsection{Purpose and Design}
 
The experiment of Appendix~\ref{app:clustering} switches the marks off, so that the comparison with its two baselines isolates the clustering task. This experiment restores the mark layer of Section~\ref{sec:model} and examines what it contributes. No baseline is run, since not all competitors model event-level marks; the question is whether the proposed method recovers a partition that is identifiable through the marks.
 
The design is built on swapped pairs. Let $\lambda_A,\dots,\lambda_F$ denote the six surfaces of Scenarios A--F in the top row of Figure~\ref{fig:synthetic results}(a). The six clusters assign one surface to each mark,
\[
(\lambda_{c0},\lambda_{c1})=(\lambda_A,\lambda_B),\ (\lambda_B,\lambda_A),\ (\lambda_C,\lambda_D),\ (\lambda_D,\lambda_C),\ (\lambda_E,\lambda_F),\ (\lambda_F,\lambda_E),
\]
labelled AB, BA, CD, DC, EF and FE, the first letter giving the mark-$0$ surface and the second the mark-$1$ surface. Within a pair, the two clusters share one exposure distribution, and each subject has the same expected count $n_{\mathrm{exp}}$ for both marks. An AB subject and a BA subject with the same $n_{\mathrm{exp}}$ therefore have the same location intensity $n_{\mathrm{exp}}(\lambda_A+\lambda_B)$, and by the superposition and marking theorems (Section~\ref{sec:model}) their event locations have the same law. The two clusters differ only in their mark probabilities,
\[
p_{AB}(\bm y)=\frac{\lambda_B(\bm y)}{\lambda_A(\bm y)+\lambda_B(\bm y)},\qquad p_{BA}(\bm y)=1-p_{AB}(\bm y),
\]
and similarly for the other pairs. Without the marks the two members of a pair are indistinguishable, so a method that ignores the marks cannot separate them beyond chance. The three pairs combine different geometries, so the experiment also involves mark-specific surfaces that differ sharply within a subject.
 
\subsection{Data Generation}

We use a $2\times2$ design analogous to that of Appendix~\ref{app:clustering}, crossing a high and a low number of subjects with a high and a low exposure (HH, LH, HL, LL), each replicated $100$ times. 
In the High-High (HH) configuration, the number of subjects
per cluster ranges from $20$ to $60$. Each subject's expected total count $\kappa_i$
is randomly drawn from an interval shared by the two clusters of its pair,
namely $(550,800)$, $(800,1200)$ and $(300,550)$ for A--B/B--A, C--D/D--C and
E--F/F--E, and is subsequently held fixed.
Its expected count per mark, $\lceil \kappa_i/2 \rceil$, varies from
$150$ to $600$ depending on the pair, so the expected total count
is split equally between the marks. Because the number of
subjects and their expected counts remain constant across replicates, only the
realized counts and point locations are random, ensuring each subject retains
its identity and exposure throughout.
For the low settings, the number of subjects is reduced to $15$--$30$ per cluster,
and the intervals become $(300,450)$, $(400,600)$ and $(200,350)$, so the
expected counts decrease to between $100$ and $300$ per mark. Across the four
regimes this spans $120$ to $230$ subjects and roughly $48\,600$ to $166\,800$
events per replicate.

 
\subsection{Model Settings}
 
The proposed method is fitted with the full marked model of Section~\ref{sec:model} and the settings of Appendix~\ref{app:clustering}: truncation level $K=30$ and concentration parameter $\alpha=1$; for every $\bm\theta_{km}$, a tensor-product cubic B-spline basis with $10$ interior knots per axis ($d=196$) and $\bm\Omega=\bm P+c_\Omega\bm I_d$ with $\bm P$ the first-order P-spline penalty; inverse-gamma hyperparameters $a_0=1$ and $b_0=0.005$.
 
The variational iterations stop when the relative change of the ELBO falls below $10^{-4}$, with a cap of $100$ iterations, and each fit starts from a single random initialization with a fixed seed. All $400$ fits met the stopping criterion within $17$ iterations. The estimated partition is $\hat z_i=\arg\max_k\nu_{ik}$, and $\hat K$ is the number of distinct labels in $\hat z$.
 
\subsection{Metrics}
 
The adjusted Rand index, purity and accuracy are defined as in Appendix~\ref{app:clustering}. 
 
Each metric is computed once per replicate. Table~\ref{tab:marked-results} reports its mean over the $100$ replicates and the standard deviation across them. Computation time is single-core wall clock, on the same convention as Appendix~\ref{app:intensity}.
 
\subsection{Results}

\begin{table}[t]
\centering
\small
\caption{Clustering performance of the proposed method over $100$ replicates in the four regimes, reported as mean (standard deviation).}
\label{tab:marked-results}
\begin{tabular*}{\linewidth}{@{\extracolsep{\fill}}lrrrr@{}}
\toprule
Regime & ARI & Purity & Accuracy & Time (s)\\
\midrule
HH & 0.938 (0.036) & 1.000 (0.000) & 0.934 (0.041) & 188.3 (21.9)\\
LH & 0.952 (0.038) & 0.993 (0.028) & 0.937 (0.043) & 109.8 (13.4)\\
HL & 0.936 (0.034) & 1.000 (0.000) & 0.922 (0.043) & 126.5 (15.1)\\
LL & 0.879 (0.045) & 1.000 (0.000) & 0.866 (0.049) & \phantom{0}85.2 (7.5)\\
\bottomrule
\end{tabular*}
\end{table}

ARI ranges from $0.879$ in LL to $0.952$ in LH. The method therefore separates clusters whose event locations alone have the same distribution. Mean purity is $1.000$ in HH, HL and LL and $0.993$ in LH, so the estimated components are almost always subsets of a single true cluster and the two members of a swapped pair are almost never placed together; the shortfall of ARI from $1$ is therefore due mainly to splitting rather than merging.
 

At high exposure, ARI is essentially unaffected by the number of subjects ($0.952$ in LH and $0.938$ in HH), whereas at low exposure it increases with the number of subjects ($0.879$ in LL and $0.936$ in HL), and falls in LL ($0.879$), where both the number of subjects and the per-subject exposure are low. Single-core computation time is $85$ to $188$ seconds per fit, with standard deviation $7.5$ to $21.9$ seconds, about $1.4$ to $2.3$ times the $37$ to $130$ seconds of the unmarked experiment, reflecting both the larger data size ($48{,}600$ to $166{,}800$ events per replicate, against $15{,}700$ to $70{,}700$) and the two mark-specific coefficient blocks per component.

\section{NBA SHOT-CHART ANALYSIS}
\label{app:nba_real_data}

\paragraph{Literature Review of Basketball Shot-chart Analysis.}
\citet{Reich2006Spatial} modeled shot frequencies and efficiencies on a discretized court, and \citet{Miller2014Factorized} used shared nonnegative basis functions to represent shot-location intensities. \citet{Yin2022SpatialHomogeneity} proposed a flexible nonparametric estimator for individual shot-location intensities with spatial homogeneity, but their focus is intensity estimation rather than clustering or marks. \citet{Jiao2021Bayesian} jointly modeled shot intensity and shot outcome for individual players, but the clustering of players was conducted only as a secondary analysis of fitted parameters for the top 50 shooters. \citet{Hu2021GroupLearning} performed Bayesian group learning for shot selection via LGCPs, yet the procedure is explicitly two-stage and excludes lower-volume players. \citet{Yin2023MatrixClustering} cluster discretized shot-intensity matrices over rectangular court regions, \citet{Hu2023ZIPClustered} model zone-level shot counts through a zero-inflated Poisson regression mixture, and \citet{WongToi2023Joint} jointly analyze shot attempts and field-goal percentages over 12 predefined front-court regions while retaining only players with at least four shots in each region. More broadly, matrix- and region-based representations do not explicitly separate spatial allocation from total shooting volume, so differences in overall attempt frequency can enter the clustering target together with differences in spatial shape.

We apply the proposed DPM-MPPP model to the NBA shot-chart data from the 2025--2026 regular season. The data were collected via the \texttt{nba\_api} Python package\footnote{\url{https://github.com/swar/nba_api}} from the official NBA stats website and consist of 219{,}160 shot attempts taken by 582 players. For each player $i$, the observed shot chart is represented as marked point data $(\bm{Y}_i,\bm{m}_i)=\{(\bm y_{ij},m_{ij})\}_{j=1}^{N_i}$, where $\bm y_{ij}\in\mathcal{B}\subset\mathbb{R}^2$ denotes the shot location on the offensive half-court and $m_{ij}\in\{0,1\}$ indicates whether the shot was missed or made. To account for varying levels of exposure, we define the offset $T_i$ as the total regular-season minutes played by player $i$. The total number of shot attempts per player $N_i$ ranges from 1 to 1{,}543, with mean 376.6 and median 292. Although the raw data contain additional contextual variables such as team, game, quarter, time, shot region, and shot type, the present analysis uses only the player identifier, the shot location, and the made/missed indicator.

The model is fitted with truncation level $K=50$, concentration parameter $\alpha=1$, and hyperparameters $a_0=1$ and $b_0=0.005$. We use a tensor-product cubic B-spline basis with 10 interior knots on each axis, yielding $d=14\times 14=196$ basis functions, and the precision matrix $\Omega$ is set to $\bm\Omega = \bm{\mathrm{P}} + c_{\bm\Omega}\bm{I}_d$ where P is taken to be the corresponding first-order Bayesian $P$-spline penalty matrix. For visualization, we summarize each fitted component by the plug-in intensity surface
$\widehat{\lambda}_{km}(\bm y) := \bigl(\bm B(\bm y)^\top \widehat{\bm\theta}_{km}\bigr)^2$,
together with the total shot-selection surface $\widehat{\lambda}_k(\bm y)=\widehat{\lambda}_{k0}(\bm y)+\widehat{\lambda}_{k1}(\bm y)$ and the induced spatial field-goal probability surface
$\widehat{p}_k(\bm y)= \frac{\widehat{\lambda}_{k1}(\bm y)}{\widehat{\lambda}_{k0}(\bm y)+\widehat{\lambda}_{k1}(\bm y)}$.
Under the activity criterion $\sum_{i=1}^n \nu_{ik}>1$, 15 clusters are active in the fitted model. 
We report the 13 major clusters containing at least 5 players, as the remaining 2 clusters contain only 1 player each.

For each cluster, Table~\ref{tab:app-representative-players} lists the top five players with the
largest number of field goal attempts as representative members, and
Figure~\ref{fig:app-cluster-summary} displays the observed shot locations together with the
estimated total intensity, make intensity, miss intensity, and success-probability surfaces.
The central empirical finding is that the clustering recovers broad offensive roles while also
resolving meaningful within-role heterogeneity.
In particular, players with broadly similar roles are further separated according to where they
concentrate shot attempts and where they convert more efficiently.
In other words, the fitted clusters are determined by a finer joint pattern of spatial preference
and location-specific success probabilities, rather than by overall shot volume or by a single
summary such as the share of three-point attempts.

\begin{table}[t!]
\centering
\caption{Representative players for the 13 major clusters. For each cluster, the five players with the largest field goal attempts are listed.}
\label{tab:app-representative-players}
\small
\setlength{\tabcolsep}{2.4pt}
\renewcommand{\arraystretch}{0.95}
\begin{tabularx}{\textwidth}{@{}c>{\raggedright\arraybackslash}X@{}}
\toprule
Cluster & Top-5 representative players \\
\midrule
1 & Matas Buzelis, Toumani Camara, OG Anunoby, Quentin Grimes, Max Christie \\
2 & Donovan Mitchell, LaMelo Ball, James Harden, Derrick White, Payton Pritchard \\
3 & Donte DiVincenzo, Tim Hardaway Jr., Klay Thompson, Duncan Robinson, Royce O'Neale \\
4 & Jaylen Brown, Tyrese Maxey, Jalen Brunson, Luka Don\v{c}i\'{c}, Kevin Durant \\
5 & Ronald Holland II, Jake LaRavia, Bruce Brown, Isaac Okoro, Dominick Barlow \\
6 & Miles Bridges, Saddiq Bey, Mikal Bridges, Naz Reid, Andrew Wiggins \\
7 & Jamal Shead, Kris Dunn, Marcus Smart, Cameron Johnson, Luke Kennard \\
8 & Sion James, Will Richard, Sidy Cissoko, Dean Wade, Nicolas Batum \\
9 & Nickeil Alexander-Walker, Ryan Rollins, Jabari Smith Jr., Ace Bailey, Russell Westbrook \\
10 & Jalen Johnson, Julius Randle, Desmond Bane, Cooper Flagg, Paolo Banchero \\
11 & Onyeka Okongwu, Keldon Johnson, Ayo Dosunmu, Donovan Clingan, Kyle Filipowski \\
12 & Amen Thompson, Naji Marshall, Dyson Daniels, Maxime Raynaud, Deandre Ayton \\
13 & Zion Williamson, Jalen Duren, Nic Claxton, Ausar Thompson, Jarrett Allen \\
\bottomrule
\end{tabularx}
\end{table}

\begin{figure*}[p]
  \centering
  \includegraphics[width=0.9612\textwidth]{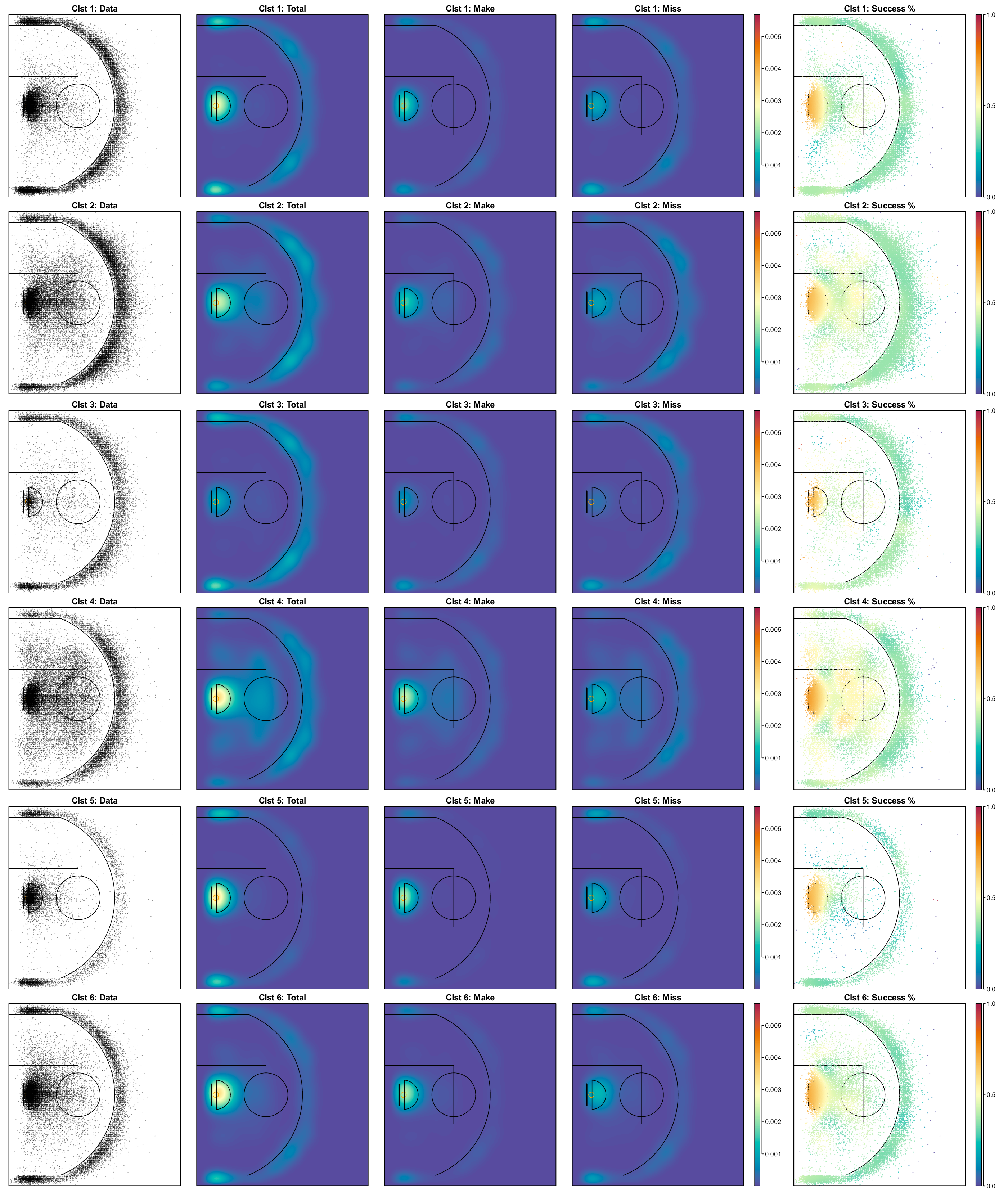}
  \caption{Observed shot locations and the estimated total intensity, make intensity, miss
  intensity, and implied success-probability surfaces for the 13 major clusters containing at least
  five players. For each cluster, the five panels summarize the observed point pattern, total shot
  allocation, made-shot allocation, missed-shot allocation, and the induced success pattern.}
  \label{fig:app-cluster-summary}
\end{figure*}

\begin{figure*}[p]
  \ContinuedFloat
  \centering
  \includegraphics[width=0.9612\textwidth]{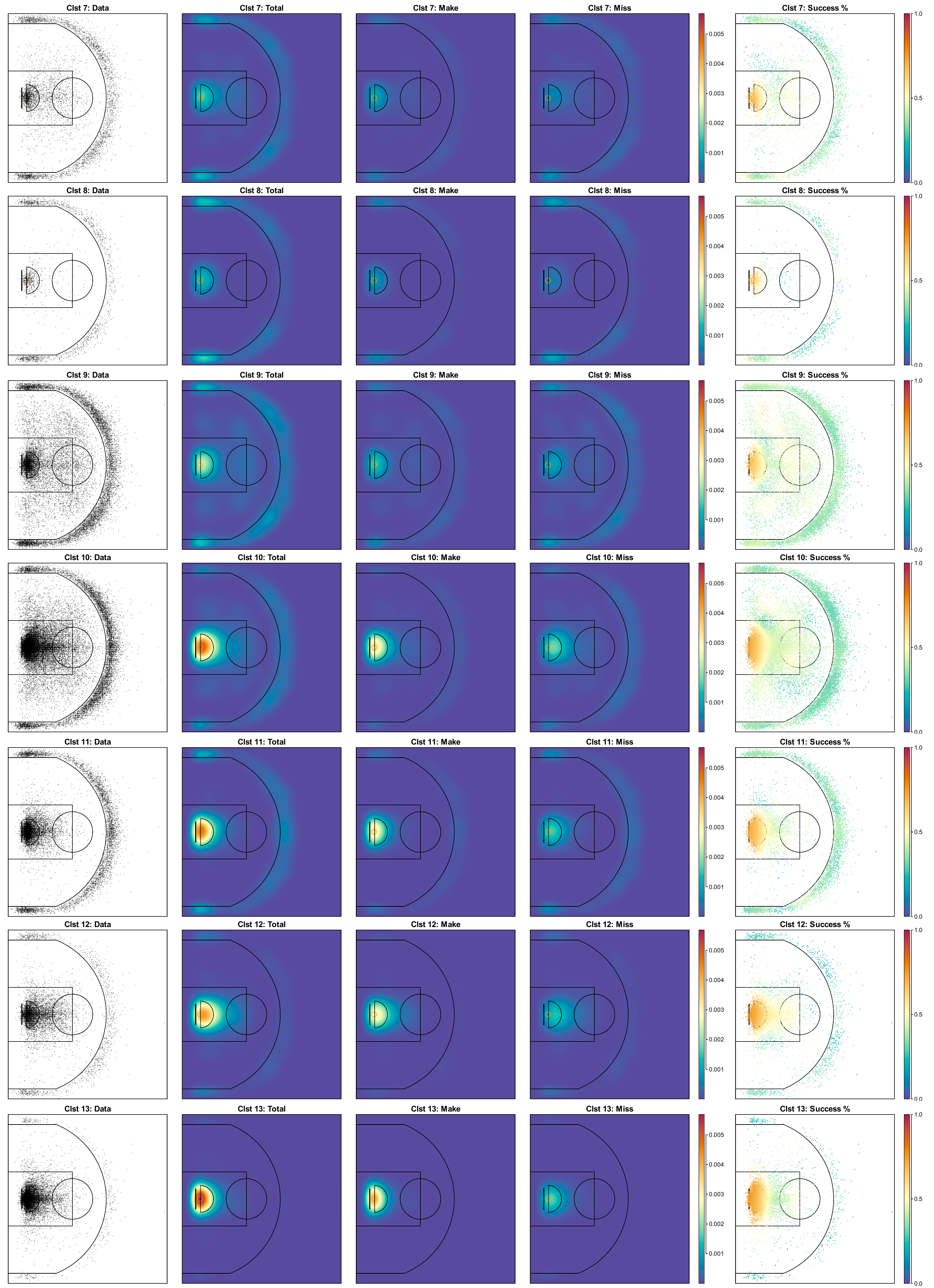}
  \caption{(Continued)}
  \label{fig:app-cluster-summary2}
\end{figure*}

Figures~\ref{fig:app-scorer}, \ref{fig:app-frontcourt}, and~\ref{fig:app-shooter} illustrate how
the fitted clusters resolve this within-role heterogeneity for three broad role families:
frontcourt players, ace scorers, and shooters.

\paragraph{Ace-scorer Clusters.}
Clusters 2, 4, and 9 all correspond to high-usage primary scoring profiles: each places substantial
mass at the rim, in the middle area, and beyond the arc, so that the total intensity is broadly
distributed over the entire half court (Figure~\ref{fig:app-scorer}).
The clusters are nevertheless clearly separated by how that mass is allocated.
Cluster 2 carries the largest three-point share, concentrated at the top of the key and on the
wings, and its perimeter mass extends farthest from the basket, indicating a willingness to attempt
deep threes well beyond the three-point line.
Cluster 4 places the greatest weight on the middle area and also converts middle-area attempts most
efficiently; more generally, its success-probability surface is the highest overall among the
scorer clusters, so that this cluster isolates efficient shot creators.
Cluster 9 has an angular profile similar to that of Cluster 2, but its perimeter attempts are
concentrated close to the arc rather than extending behind it, and it allocates a relatively larger
share of attempts to the corners and short corners; secondary handlers and wing scorers who operate within a more restricted shooting range fall into this cluster.

\begin{figure*}[t!]
  \centering
  \includegraphics[width=0.94\textwidth]{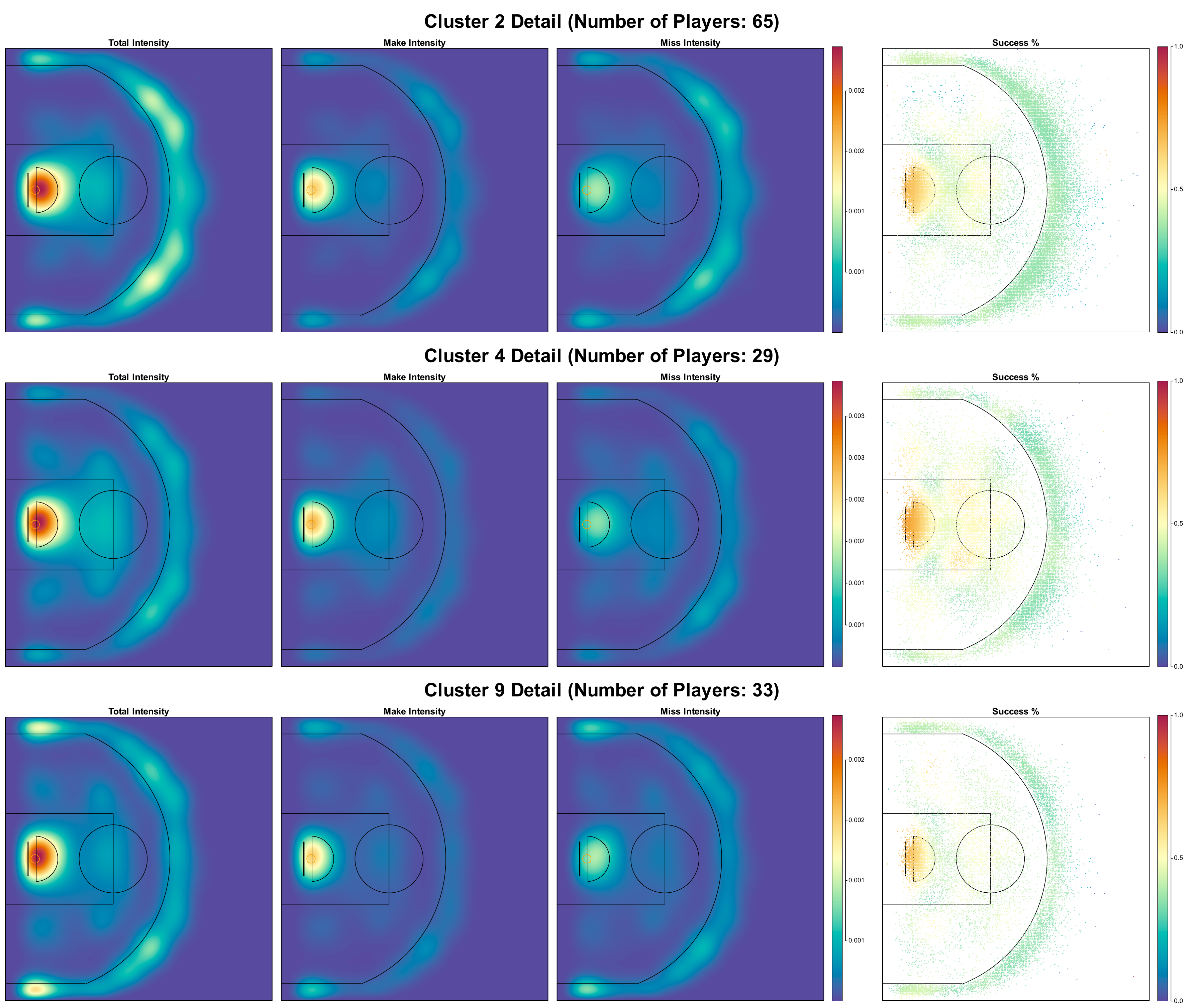}
  \caption{Detailed total, make, miss intensity surfaces, and success-probability maps for the three
  ace-scorer clusters, Clusters 2, 4, and 9. These clusters represent primary scoring profiles that
  attempt shots at the rim, in the middle area, and beyond the arc, but they differ in the relative
  concentration and extent of their perimeter attempts: Cluster 2 has the largest and deepest
  three-point component at the top of the key and on the wings, Cluster 4 has the highest
  middle-area share and the highest overall success probability, and Cluster 9 shoots from a more
  restricted range close to the arc with a relatively larger corner share.}
  \label{fig:app-scorer}
\end{figure*}

\paragraph{Frontcourt Clusters.}
Clusters 10, 12, and 13 are all strongly paint-centered, with the dominant mode of the total
intensity located inside the restricted area, but they differ systematically in how much offensive
mass extends beyond the paint (Figure~\ref{fig:app-frontcourt}).
Cluster 13 is almost entirely concentrated near the basket, with essentially no perimeter component,
and it also attains the highest near-rim success probability among all clusters; this profile
corresponds to traditional interior finishers whose attempts are generated almost exclusively at
the rim.
Cluster 12 remains predominantly rim-oriented but adds a thin, low-volume three-point component,
a pattern consistent with frontcourt role players who take a small number of spot-up perimeter
attempts while still deriving most of their scoring from the paint.
Cluster 10 retains a dominant rim component as well, but exhibits the broadest secondary spread of
the three, including short-middle attempts and a visible above-the-break perimeter band; this is
the profile of versatile frontcourt players who are key facilitators of the offense and also capable of ball-handling.

\begin{figure*}[t!]
  \centering
  \includegraphics[width=0.94\textwidth]{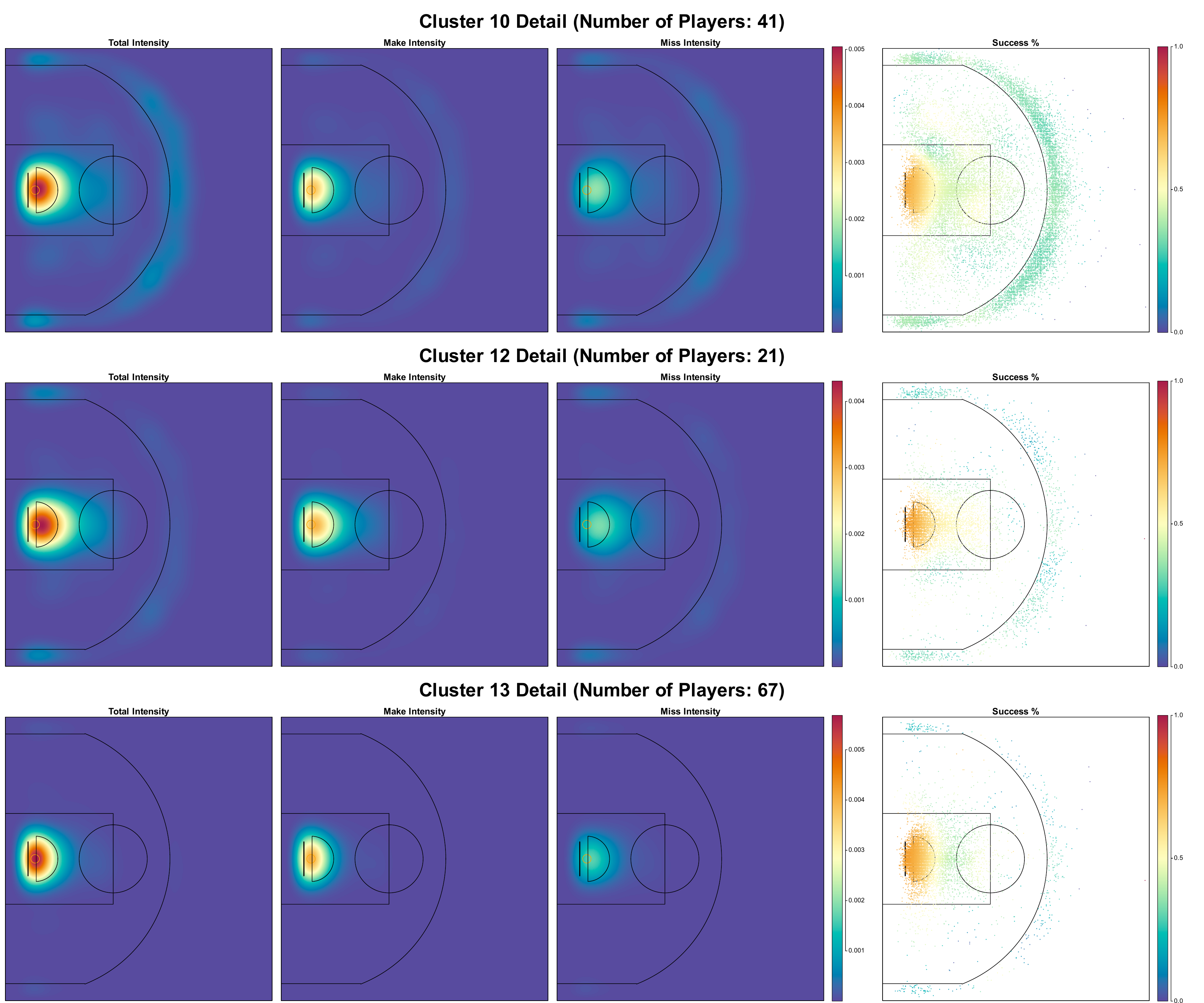}
  \caption{Detailed total, make, miss intensity surfaces, and success-probability maps for the three
  frontcourt clusters, Clusters 10, 12, and 13. All three are paint-dominated, but Cluster 13 is
  almost entirely concentrated near the basket and attains the highest near-rim success probability,
  Cluster 12 is primarily rim-oriented with a small three-point component, and Cluster 10 adds
  short-middle and above-the-break perimeter activity characteristic of ball-handling frontcourt
  players.}
  \label{fig:app-frontcourt}
\end{figure*}

\paragraph{Shooter Clusters.}
Clusters 1, 3, 5, and 7 share a broadly perimeter-oriented profile, but differ in how attempts
are distributed along the arc and in the interior (Figure~\ref{fig:app-shooter}).
Cluster 1 combines a substantial perimeter component with the highest at-rim mass among the shooter
clusters, corresponding to three-and-drive wings who mix catch-and-shoot perimeter attempts with
direct attacks on the basket.
Cluster 3 is almost exclusively a three-point profile, isolating pure perimeter specialists,
whereas Cluster 5 concentrates its perimeter mass in the corners with sparse activity above the break, distinguishing corner specialists from the more general perimeter profile of Cluster 3.
Cluster 7 retains a non-negligible middle-area component while spreading perimeter attempts fairly
evenly along the arc, yielding a more balanced shooting profile.

\begin{figure*}[t!]
  \centering
  \includegraphics[width=0.94\textwidth]{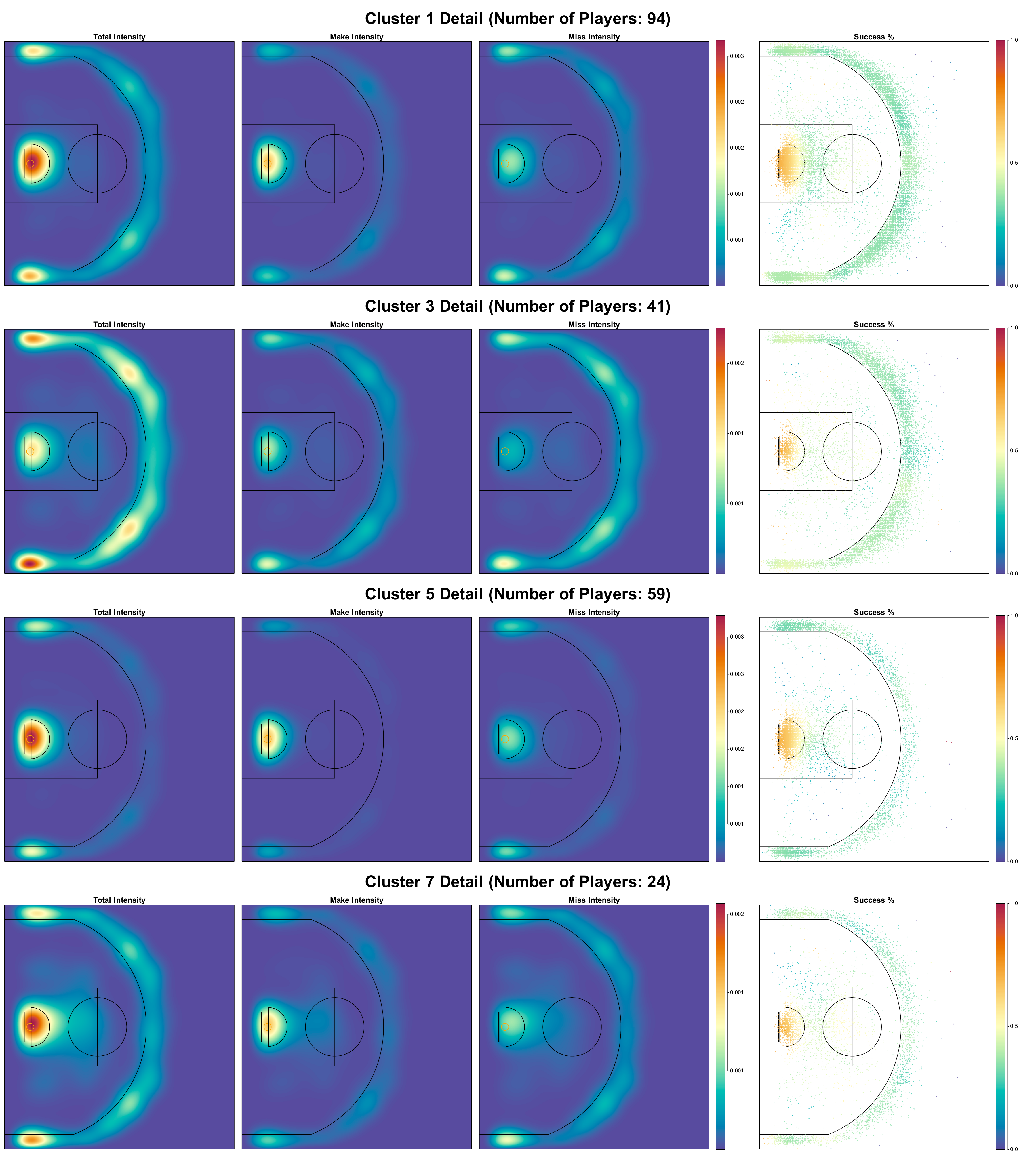}
  \caption{Detailed total, make, miss intensity surfaces, and success-probability maps for the four
  shooter clusters, Clusters 1, 3, 5, and 7. The four clusters share a perimeter-oriented profile
  but differ in interior involvement and in the allocation of perimeter attempts: Cluster 1
  has the largest at-rim component, Cluster 3 is almost exclusively a three-point profile, Cluster 5
  is concentrated in the corners, and Cluster 7 retains middle-area activity while spreading
  perimeter attempts evenly along the arc.}
  \label{fig:app-shooter}
\end{figure*}

\end{document}